\def\paper{\ignore}   
\def\report{}         

\documentclass{article}
\newenvironment{proof}{\begin{trivlist}\item[]\hspace{\parindent}{\em Proof.}\ }
                     {\hfill \rule{.2em}{.2em}\end{trivlist}}
\usepackage{times}
\usepackage{graphicx}
\usepackage{epic}
\usepackage{eepic}
\usepackage{epsfig}
\usepackage{latexsym}
\usepackage{stmaryrd}
\usepackage{amssymb}
\usepackage{url}
\usepackage{prooftree}
\usepackage{float}
\usepackage{thmtools} 
\usepackage{thm-restate}
\usepackage{hyperref}
\usepackage[noabbrev,capitalize]{cleveref}

\usepackage[usenames]{color}
\usepackage{multicol} 
\usepackage[f]{esvect} 
\usepackage{ifthen} 
\input xy                        
\xyoption{all} 
\usepackage[dvips]{xy}           

\newcommand{\delia}[1]{{\color{blue} #1}}

\newtheorem{theorem}{Theorem}[section]
\newtheorem{lemma}[theorem]{Lemma}
\newtheorem{corollary}[theorem]{Corollary}
\newtheorem{proposition}[theorem]{Proposition}
\newtheorem{definition}[theorem]{Definition}
\newtheorem{remark}[theorem]{Remark}

\long\def\ignore#1{\relax}

\newcommand{\mathsmall}[1]{\ensuremath{\mbox{\small{$#1$}}}}
\newcommand{\usm}{{\tt R_{\mX}}}
\newcommand{\lamBig}{'\!\Lambda\mbox{\textsc{big}}}

\newcommand{\ShaneIgnore}[1]{}

\newcommand{\fresh}[1]{\widehat{#1}}
\newcommand{\U}{{\cal U}}
\newcommand{\Er}{{\cal E}}
\newdir{> }{{}*!/26pt/@{>}}
\newdir{>' }{{}*!/16pt/@{>}}
\newdir{>'' }{{}*!/3pt/@{>}}
\newdir{>''' }{{}*!/7pt/@{>}}
\newdir{>> }{{}*!/26pt/@{>>}}

\newboolean{ShowFullProofs}
\setboolean{ShowFullProofs}{true} 
\newcommand{\ShortLongProof}[2]{\ifthenelse{\boolean{ShowFullProofs}}{#2}{#1}}

\def\itmath#1{\leavevmode\ifmmode{\mbox{\it#1} }\else{\it#1 }\fi}
\def\sfmath#1{\leavevmode\ifmmode{\mbox{\sf#1} }\else{\sf#1 }\fi}
\def\condmath#1{\leavevmode\ifmmode{#1}\else{$#1$}\fi}

\newcommand{\lpar}     {\condmath{\bindnasrepma}}

\newcommand{\ie}{{\em i.e.~}}
\newcommand{\eg}{{\em e.g.~}}

\newcommand{\sep}{\hspace*{0.5cm}}
\renewcommand{\>}{\rightarrow}

\def\l{\lambda}

\def\Gam{\Gamma}
\def\Del{\Delta}

\newcommand{\Rew}[1]{\rightarrow_{#1}}
\newcommand{\Rewn}[2][*]{\rightarrow^{#1}_{#2}}

\newcommand{\s}{{\tt s}}

\newcommand{\es}{{\tt es}}

\newcommand{\les}{\l \es}

\newcommand{\lm}{\lambda_{sub}}
\newcommand{\mterms}{\Lambda}
\newcommand{\m}{{sub}}

\newcommand{\Es}{{\tt E}_{\s}}
\newcommand{\EqEs}{=_{\Es}}

\newcommand{\llxr}{\l {\tt lxr}}

\newcommand{\isubs}[2]{\{#1 / #2\}}

\newcommand{\SN}[1]{{\cal SN}_{#1}}

\newcommand{\lambdasigma}{\l_{\sigma}}

\newcommand{\B}{{\tt B}}
\newcommand{\C}{{\tt C}}

\newcommand{\R}{{\cal R}}

\newcommand{\tra}[1]{{\tt T}(#1)}
\newcommand{\pntra}[1]{{\tt W}(#1)}
\newcommand{\pnlestra}[1]{{\tt Z}(#1)}

\newcommand{\nombretra}{{\tt T}}

\newcommand{\un}[1]{\underline{#1}}

\def\l{\lambda}
\def\int{int}

\newcommand{\Rewplus}[1]{\rightarrow^{+}_{#1}}

\newcommand{\LRew}[1]{\; \mbox{}_{#1}{\leftarrow}\ }
\newcommand{\LRewn}[1]{\; \mbox{}^{*}_{#1}{\leftarrow}\ }

\newcommand{\irule}[2]
   {\renewcommand{\arraystretch}{1.2}
    \begin{array}{c} \mbox{\(  #1 \)} \\ \hline \mbox{\( #2 \)} \end{array}}

\newcommand{\fv}{{\tt fv}}

\newcommand{\bv}{{\tt bv}}

\newcommand{\vd}{\vdash}

\newcommand{\ini}{\cap \; {\tt I}}
\newcommand{\ine}{\cap \; {\tt E}}

\newcommand{\axiom}{ {\tt ax}}
\newcommand{\aaxiom}{ {\tt ax}^{+}}
\newcommand{\maxiom}{ {\tt ax}^{*}}

\newcommand{\abs}{ {\tt abs}}
\newcommand{\aabs}{ {\tt abs}^{+}}
\newcommand{\mabs}{ {\tt abs}^{*}}
\newcommand{\absi}{ {\tt abs}_1}
\newcommand{\absii}{ {\tt abs}_2}
\newcommand{\mabsi}{ {\tt abs}^{*}_1}
\newcommand{\mabsii}{ {\tt abs}^{*}_2}
\newcommand{\app}{ {\tt app}}
\newcommand{\aapp}{ {\tt app}^{+}}
\newcommand{\mapp}{ {\tt app}^{*}}

\newcommand{\asubstr}{ {\tt subs}^{+}}
\newcommand{\msubstr}{ {\tt subs}^{*}}
\newcommand{\subsi}{ {\tt subs}_1}
\newcommand{\subsii}{ {\tt subs}_2}
\newcommand{\msubsi}{ {\tt subs}^{*}_1}
\newcommand{\msubsii}{ {\tt subs}^{*}_2}
\newcommand{\Var}{{\tt Var}}
\newcommand{\Gc}{{\tt Gc}}
\newcommand{\BGc}{{\tt BGc}}
\newcommand{\bp}{\beta_{p}}

\newcommand{\us}{{\tt R}}

\newcommand{\Appi}{{\tt App}_1}
\newcommand{\Appii}{{\tt App}_2}
\newcommand{\Appiii}{{\tt App}_3}

\newcommand{\Compi}{{\tt Comp}_1}
\newcommand{\Compii}{{\tt Comp}_2}
\newcommand{\Lamb}{{\tt Lamb}}

\newcommand{\Com}{{\tt C}}

\newcommand{\ov}[1]{\overline{#1}}

\newcommand{\dom}{{\tt dom}}

\newcommand{\sm}{{\tt sm}}

\newcommand{\set}[1]{ \{ #1 \}}

\newcommand{\size}{{\tt s}}

\newcommand{\mul}[2]{{\tt M}_{#1}(#2)}

\newcommand{\ems}{\emptyset}

\newcommand{\paralp}[1]{\Rrightarrow_{#1}}
\newcommand{\lparalp}[1]{\; \mbox{}_{#1}{\Lleftarrow}\ }
\newcommand{\paralnp}[1]{\Rrightarrow^{*}_{#1}}
\newcommand{\mX}{\mathbb{X}}

\newcommand{\ALC}{{\tt ALC}}
\newcommand{\ih}{i.h.}

\renewcommand{\vv}[1]{\overline{#1}}
\renewcommand{\vec}[1]{\overline{#1}}
\newcommand{\capp}[1]{\cap_{#1}}
\newcommand{\saddl}{{\tt add}_{\l}}
\newcommand{\smultl}{{\tt mul}_{\l}}
\newcommand{\saddls}{{\tt add}_{\lm}}
\newcommand{\smultls}{{\tt mul}_{\lm}}

\newcommand{\addl}{{\tt add}^i_{\l}}
\newcommand{\multl}{{\tt mul}^i_{\l}}
\newcommand{\addls}{{\tt add}^i_{\lm}}
\newcommand{\multls}{{\tt mul}^i_{\lm}}
\newcommand{\cwc}[1]{{[} \! {[} #1 {]}  \! {]}}
\renewcommand{\lpar}{\l_{\beta_{p}}}
\newcommand{\ldef}{\l_{def}}
\newcommand{\parlm}[1]{{\tt U}(#1)}
\newcommand{\lmpar}[1]{{\tt V}(#1)}
\newcommand{\type}{{\tt type}}

\newcommand{\defn}[1]{{\bf #1}}
\newcommand{\TST}{{\cal T}}

\restylefloat{figure}

\begin{document}

\paper{\frontmatter}

\title{Milner's Lambda-Calculus with Partial Substitutions}

\author{
\begin{tabular}{cc}
Delia Kesner & Shane \'O Conch\'uir \\
PPS, CNRS and Universit\'{e} Paris 7 & Trinity College Dublin\\
France &  Ireland \\
\end{tabular}}

\date{This is a technical report developed in 2008 \\ but remained unpublished until now. }

\maketitle

\begin{abstract}
  We study  Milner's lambda-calculus with  partial substitutions.
  Particularly,  we  show  confluence  on  terms  and  metaterms,
  preservation     of     $\beta$-strong    normalisation     and
  characterisation   of  strongly   normalisable  terms   via  an
  intersection typing  discipline. The results  on terms transfer
  to  Milner's  bigraphical  model  of the  calculus.  We  relate
  Milner's  calculus  to calculi  with  definitions,
  to calculi with  explicit
  substitutions, and to MELL Proof-Nets.
\end{abstract}

\section{Introduction}

The  $\lm$-calculus  was  introduced  by  Milner as  a  means  to
modelling   the   $\l$-calculus  in   bigraphs~\cite{Milner2006}. However, the  $\lm$-calculus is  interesting apart  from the
model; it enjoys confluence  on terms, step-by-step simulation of
$\beta$-reduction~\cite{OConchuir:2006a},   and  preservation  of
$\beta$-strong    normalisation     (PSN)    \emph{i.e.}    every
$\l$-calculus term which  is $\beta$-strongly normalising is also
$\lm$-strongly normalising ~\cite{ShaneHOR06}.

In this paper we study many remaining open questions about the
$\lm$-calculus.
The first  of them concerns \emph{confluence  on metaterms} which
are terms containing  \emph{metavariables} usually used to denote
\emph{incomplete}  programs  and/or   proofs  in  higher-order
frameworks~\cite{HuetThEtat}.   To obtain  a  confluent
reduction  relation  on metaterms  we  need  to extend  the
existing notion of reduction  on terms. We develop a proof
of confluence  for this extended  new relation by using Tait
and  Martin-L\"of's  technique.  This  proof  includes  a  formal
argument  to show  that the  calculus of  substitution  itself is
terminating.

Our main  contribution lies  in studying the  connections between
the $\lm$-calculus and other formalisms.  We start by considering calculi with
definitions,       namely,      the       partial 
$\l$-calculus~\cite{deBruijn87,Ned92}, which we call $\lpar$, and
the  $\l$-calculus  with  definitions~\cite{severi94definitions},
which we call $\ldef$.  We distinguish arbitrary terms of
  the   calculi   with    definitions,   which   we   call
    $\mterms$-terms,  from  (pure) terms   without  definition,  which  are
  ordinary   $\l$-terms.     We   show   that    the   sets   of
strongly-normalising  $\l$-terms  in $\lm$  and  $\lpar$ are  the
same.  Similarly,  we show that the  sets of strongly-normalising
$\mterms$-terms in $\lm$ and $\ldef$ are equal.  Thus, we demonstrate
that partial substitutions and definitions are similar notions.

  We also relate  $\lm$-strongly normalising terms
  to  typed $\mterms$-terms.
  For that, we start by introducing an
  intersection type discipline for  $\mterms$-terms.
  We  then give a simple
  (and    constructive)   argument   to    prove   $\lpar$-strong
  normalisation for   typed  $\l$-terms.  This  argument turns
  out to be sufficient  to conclude $\lm$-strong normalisation for
  \emph{intersection} typed $\mterms$-terms. 
  By proving 
  the converse \emph{i.e.} $\lm$-strongly normalising $\mterms$-terms can be  typed
  in the intersection type discipline, we also provide a characterisation
  of $\lm$-strongly normalising $\mterms$-terms.
  
The relation between typable and $\lm$-strongly normalising
$\mterms$-terms also gives an alternative proof of PSN for the
$\lm$-calculus, which is self-contained, and which simplifies previous
work~\cite{ShaneHOR06} considerably.
Indeed, the existing proof is quite
involved, and uses  a translation of $\lm$ into a rather complex calculus,
obtained by modifying a language with explicit resources
  inspired from Linear Logic's Proof-Nets. 
 
Another contribution of the paper
is the  study of the relation  between
partial   substitutions   and   explicit   substitutions.    More
precisely, we define a translation  from $\lm$ to a calculus with
explicit    substitutions    called   $\les$~\cite{K07}.     This
translation preserves  reduction and  has at least  two important
consequences.  On one hand, we  obtain a simple proof of
$\lm$-strong
normalisation 
for  \emph{simply typed}  $\mterms$-terms.   A second
consequence is  that the existing simulation of  the simply typed
$\les$-calculus  into MELL  Proof-Nets~\cite{K07} also gives a 
natural interpretation for the simply typed $\lm$-calculus by composition.  As a
corollary, $\lm$-strong normalisation  for simply typed $\mterms$-terms can
also be inferred from strong normalisation of MELL Proof-Nets.

Finally, we transfer our confluence and strong normalisation proofs on $\mterms$-terms without metavariables 
in $\lm$ to Milner's model using an existing result. 

{\bf Road map.} 
\Cref{s:lm-calculus}  introduces the $\lm$-calculus.
Metaterms are introduced in~\cref{s:metaterms}: 
some preliminary properties are discussed in~\cref{s:basic-properties-lm}
and confluence on metaterms is proved
using     Tait      and     Martin-L\"of's     technique
in~\cref{s:confluence}.      In~\cref{s:partial-definitions}   we  relate  $\lm$   to
two 
calculi     with     definitions,    $\lpar$     and     $\ldef$.
In~\cref{s:PSN},  we present  the translation  from  $\lm$ to
$\les$ and  prove that  reduction in the  former is  simulated by
non-empty  reduction  sequences   in  the  latter.  \cref{s:charact-les-lm}  presents a  neat characterisation
of  $\lm$-strongly  normalising  terms  using  intersection  type
systems   as  well   as   the  PSN   property   for  untyped
$\mterms$-terms of $\lm$.     We  conclude
$\lm$-strong  normalisation  for  simply typed  $\mterms$-terms  from
strong $\les$-normalisation  for simply typed  $\mterms$-terms. 
 Last
but  not  least,  we  discuss a relation between  $\lm$  and  MELL
Proof-Nets  and transfer results to
the bigraphical setting
in~\cref{s:graphical}.

\section{The $\lm$-calculus}
\label{s:lm-calculus}

The $\lm$-calculus  was introduced by Milner to present a model of
the $\lambda$-calculus  in local bigraphs.   The    calculus   was
  inspired by $\lambdasigma$~\cite{ACCL91}
  although  it is a  \emph{named} calculus  and has  turned out  to have
  stronger properties as we show in this paper.  Terms 
  of the $\lm$-calculus, called \defn{$\mterms$-terms}, are  given by the following grammar: 

\[ t ::= x \mid t\ t \mid \l x. t \mid t[x/t] \]

The set of terms includes \defn{variables}, \defn{abstractions},
\defn{applications} and \defn{closures} respectively.  The piece of
syntax $[x/t]$, which is not a term itself, is called an
\defn{explicit substitution}. A term $t$ is said to be \defn{pure} if $t$
does not contain any explicit substitution. \\

\defn{Free} and \defn{bound} variables are defined as usual, by assuming 
the terms $\l x. t$ and $t[x/u]$ bind $x$ in $t$. Formally,
\[ \begin{array}{llllll}
   \fv(x) & = & \set{x}  & \bv(x) & = & \ems \\
   \fv(t\ u) & = & \fv(t) \cup \fv(u) & \bv(t\ u) & = & \bv(t)
   \cup \bv(u)\\
   \fv(\l x. t) & = & \fv(t) \setminus \set{x} & \bv(\l x. t) & = &
   \bv(t) \cup \set{x} \\
   \fv(t[x/u]) & = & \fv(t) \setminus \set{x} \cup \fv(u) & \bv(t[x/u])
   & = & \bv(t) \cup \set{x} \cup \bv (u) \\
   \end{array} \] 
We consider $\alpha$-conversion which is the congruence generated by
renaming of bound variables. Thus for example $(\l y. x)[x/y]
=_{\alpha} (\l z. x')[x'/y]$. We work with $\alpha$-equivalence
classes so that two bound variables of the same term are assumed to be
distinct, and no free and bound variable of the same term have the
same name. Thus, $\alpha$-conversion avoids capture of variables.
We use notation $\l \vec{y}.s$ for $\l y_1. \ldots. \l
  y_n. s$, where $s$ is not a lambda abstraction.  \defn{Implicit
substitution} on $\mterms$-terms can be defined modulo $\alpha$-conversion
in such a way that capture of variables is avoided:

\[ \begin{array}{llllll}
   x\isubs{x}{v} & :=  & v \\
   y\isubs{x}{v} & :=  & y & \mbox{ if } y \neq x \\
   (\l y. t)\isubs{x}{v}  & := & \l y. t\isubs{x}{v} \\
   (tu)\isubs{x}{v} & := & t\isubs{x}{v}u\isubs{x}{v} \\
  t[y/u]\isubs{x}{v}& := &  t\isubs{x}{v}[y/u\isubs{x}{v}]\\
   \end{array} \]

The set  of \defn{$\mterms$-contexts}  can be  defined by the following grammar:

\[ C ::= \Box \mid C\ t \mid t\ C \mid \l x. C \mid C[x/t]
\mid t[x/C] \]

We use the notation $C\cwc{u}_\phi$
to mean that the hole $\Box$ in the context $C$ has been replaced
by the term $u$ \emph{without capture of the variables in the set
  $\phi$}. Thus for example, if $C = \l z. \Box$, then
$C\cwc{x}_{\phi}$
with $x \in \phi$
means in particular that $z \neq x$.

The reduction rules  of  the $\lm$-calculus are given
in~\cref{f:lm-rules}. 

\begin{figure}[htp]
\[ \begin{array}{|lcll|}
\hline 
(\l x. t)\ u    &\Rew{\B}   & t[x/u] & \\
t[x/u]          &\Rew{\Gc}  & t      & \mbox{ if } x \notin \fv(t) \\ 
C\cwc{x}_\phi[x/u]   &\Rew{\us}   & C\cwc{u}_\phi[x/u] & \mbox{ if } \{x\} \cup \fv(u) \subseteq \phi\\
\hline
\end{array} \]
\caption{Reduction Rules for $\mterms$-Terms}
\label{f:lm-rules} 
\end{figure}

As Milner  describes, an explicit substitution  
$[x/u]$ acts  \emph{`at  a distance'} on each free occurrence
of $x$ in turn, rather than migrating a copy of itself towards
each such occurrence \eg\  the reduction step
  $(\l x.  x\ (y\ y))[y/t] \Rew{\us}  (\l x. x\  (t\ y))[y/t]$
implements  a \emph{partial}  substitution.
Partial substitution is atypical and therefore presents novel
challenges: traditional methods of proving simulation or normalisation
properties need to be adapted to this setting.  However, it also
exhibits interesting properties: the
$\lm$-calculus retains PSN whilst having full
composition of substitutions.  This is remarkable since Melli\`{e}s'
counterexample of PSN \cite{Mellies1995a} for $\lambdasigma$ was based
on full composition of substitutions.

\ignore{
  Bloo and  Rose's subsequent  work~\cite{Bloo95}
  led to the intuition that

\begin{quote}
  \emph{``the   essential  property   for  losing   PSN   is  the
    possibility of moving one  substitution from outside a second
    substitution to the \emph{inside} of the latter by means of a
    composition of substitutions.'' \cite{Bloo1997}}
\end{quote}

The composition  described in~\cite{Bloo1997} is  given by a local reduction rule, by pushing  one  substitution
inside  the one appearing on its left-hand side  (see for example rules
$\Compi$ and $\Compii$ in~\cref{s:PSN}). Composition  of  substitutions
in the  $\lm$-calculus is different in the sense that it is implicit via
the non-local substitution rule.
}

We denote by $\sm$ (resp. $\B\sm$)
the  reduction relation generated by
the reduction  rules $\us$ and $\Gc$ in~\cref{f:lm-rules} 
(resp. $\B$, $\us$ and $\Gc$)
and closed by \emph{all} contexts. 
The \defn{reduction relations  of the $\lm$-calculus for $\mterms$-terms},
defined on $\mterms$-terms,  are generated by the previous
reduction relations  $\sm$ (resp. $\B\sm$) \emph{ modulo the
  equivalence relation 
$\alpha$}, they are denoted by  $\Rew{\m}$ (resp. $\Rew{\lm}$):   
\[  \begin{array}{llll} 
    t \Rew{\m} t' & \mbox{ iff there are } s,s' \mbox{ s.t.}  & t
    =_{\alpha} s \Rew{\sm}  s'=_{\alpha} t' \\
   t \Rew{\lm} t' & \mbox{ iff there are } s,s' \mbox{ s.t.}  & t
   =_{\alpha} s \Rew{\B\sm}  s'=_{\alpha} t'  \\
   \end{array} \]

Thus, the reduction relation 
acts on  $\alpha$-equivalence classes.
For any reduction relation $\R$, we use the notation $\Rewn{\R}$
(resp.  $\Rewplus{\R}$) to denote the reflexive (resp. reflexive
and
transitive) closure of $\Rew{\R}$. As a consequence if $t
\Rewn{\m} t'$ (resp, $t
\Rewn{\lm} t'$) in $0$ steps, then $t=_{\alpha} t'$ (and not $t = t'$). \\

  Reduction enjoys   the  following   properties.

\begin{lemma}
  [Preservation of Free Variables]
  \label{l:preservation-variables}
Let $t,t'$ be $\mterms$-terms. If $t \Rew{\lm} t'$, then  $\fv(t') \subseteq \fv(t)$.
\end{lemma}

\begin{proof}
  By induction on $t \Rew{\lm} t'$.
\end{proof}

\begin{lemma}[Full Composition for Terms]
\label{l:full-composition}
Let $t,u$ be $\mterms$-terms. Then $t[x/u] \Rewplus{\lm} t \isubs{x}{u}$.
\end{lemma}

\begin{proof}
  By induction on $t$.
\end{proof}

Full composition   guarantees
  that explicit substitution  implements the implicit one.  While
  this  property seems reasonable/natural,  it is  worth noticing
  that many calculi with  explicit substitutions do not enjoy it.

\begin{lemma}[One-Step $\beta$-Simulation]
\label{l:one-step-beta-simulation}
Let $t,u$ be $\l$-terms.
If $t \Rew{\beta} u$, then $t \Rewplus{\lm} u$. 
\end{lemma}

\begin{proof}
  By induction on $t \Rew{\beta} u$.
\end{proof}


\section{Metaterms}
\label{s:metaterms}

We now introduce
\emph{metaterms},  usually used to denote  \emph{incomplete} programs/proofs
  in   higher-order  frameworks~\cite{HuetThEtat}.  Metavariables
  come  with a minimal amount of  information to guarantee
  that some  basic operations such  as instantiation (replacement
  of metavariables  by metaterms) are sound in  a typing context.
  An example can  be given by the (non  annotated) metaterm $t=\l
  y. y\ \mX\  (\l z. \mX)$, for which  the instantiation of $\mX$
  by a term containing a  free occurrence of $z$ would be unsound
  (see~\cite{MunozTh,DHK2000,FP07} for details).
  The  set   of  \defn{$\mterms$-metaterms}  is  obtained   by  adding
    \defn{annotated metavariables} of  the form $\mX_\Del$ (where
    $\Del$ is  a set of  variables) to the grammar
   generating the 
    $\mterms$-terms introduced in~\cref{s:lm-calculus}.
    The  notion of \defn{free
    variable}  is extended  to $\mterms$-metaterms by $\fv(\mX_\Del)= \Del$. As a
    consequence, $\alpha$-conversion can also be defined on
    $\mterms$-metaterms and thus for example $\l x. \mX_{x,y} =_{\alpha} \l z. \mX_{z,y}$.
We also extend the  standard notion of \defn{implicit substitution} to $\mterms$-metaterms as follows:
$$
\begin{array}{llll}
\mX_\Del \isubs{x}{v} & := & \mX_\Del & \mbox{if } x \notin \Del\\
\mX_\Del \isubs{x}{v} & := & \mX_\Del[x/v] & \mbox{if } x \in \Del\\
\end{array}
$$

It is worth noticing that Milner's original presentation did
    not consider metaterms, as  the bigraphical system did not model
    them. However, all properties we prove here involving metaterms
    hold also for terms.

    Throughout  this section,  we include a new rule
    in the reduction   relation  as   well   as  a new equation in the equivalence relation. Indeed, we add  the equation  $\Com$
    and the reduction rule $\usm$, presented in~\cref{f:extra-lm-rules},
    to the ones in~\cref{f:lm-rules}.

\begin{figure}[htp]
\[ \begin{array}{|llll|}
\hline  
\mbox{{\bf Equation}}: &&&\\
t[x/u][y/v]  & =_{\Com} &   t[y/v][x/u] &  \mbox{ if } y \notin \fv(u)\ \&\
x \notin \fv(v)\\ &&&\\ 
\mbox{{\bf Reduction Rule}}: &&&\\
C\cwc{\mX_\Del}_{\phi} [x/u]  & \Rew{\usm} &    
C\cwc{\mX_\Del[x/u]}_{\phi}[x/u]\  & \mbox{ if }   x \in
\Del\ \&\  x \cup \fv(u) \subseteq \phi \\ 
&&& \;\;\;\;  \&\ C \neq \Box[y_1/v_1]\ldots[y_n/v_n]\ (n \geq 0) \\
\hline
\end{array} \]  
\caption{Extra Equation and Reduction Rule for $\mterms$-Metaterms}
\label{f:extra-lm-rules} 
\end{figure}

Remark in particular that $\usm$ cannot be applied if the
  context is  empty. Remark also that the equation $\Com$ can
always  be postponed w.r.t. reduction if only
terms (and not metaterms) are considered.

The equation $\C$ is not part of the original presentation
  of Milner's  $\lambda$-calculus but we  include it here  for at
  least two reasons. The first one is that in bigraphs as well as
  in  proof-nets, which are  graphical representation  of $\mterms$-terms,
  some  syntactic  details ---such  as  for  example  the order  of
  appearence  of \emph{independent} substitutions---  is 
  extensionally  irrelevant.  The   second  reason  is  that  the
  reduction   relation    on   $\mterms$-metaterms   we    study   in~\cref{s:confluence}  turns out to  be confluent  only
  with the  equation $\C$.

The \emph{equivalence  relation}  generated  by  the conversions  $\alpha$  and
$\Com$ is denoted by $\EqEs$. We now   denote by $\sm$ (resp. $\B\sm$)
the  reduction relation generated by
the reduction  rules $\{\us, \Gc, \usm\}$ (resp. $\{\B, \us, \Gc, \usm \}$)
and closed by \emph{all} contexts.
The \defn{reduction relations  of the $\lm$-calculus for $\mterms$-metaterms} are generated by the
reduction relations  $\sm$ (resp. $\B\sm$) \emph{ modulo the
  equivalence relation 
$\Es$}, always  denoted by  $\Rew{\m}$ (resp. $\Rew{\lm}$):   
\[  \begin{array}{llll} 
    t \Rew{\m} t' & \mbox{ iff there are } s,s' \mbox{ s.t.}  & t
    \EqEs s \Rew{\sm}  s' \EqEs t' \\
   t \Rew{\lm} t' & \mbox{ iff there are } s,s' \mbox{ s.t.}  & t
   \EqEs s \Rew{\B\sm}  s'\EqEs t'  \\
   \end{array} \]

\subsection{Preliminary Properties}
\label{s:basic-properties-lm}

In this section we prove some preliminary properties of $\mterms$-metaterms. First, full composition still holds for $\mterms$-metaterms:

\begin{lemma}[Full Composition for Metaterms]
Let $t, u$ be $\mterms$-metaterms. Then
$t[x/u] \Rewn{\m} t\isubs{x}{u}$.
\end{lemma}

\begin{proof}
By induction on the number $n_{x,t}$ of free occurrences of 
$x$ in $t$.

\begin{itemize}
\item If $n_{x,t}=0$, then $t[x/u] \Rew{\Gc} t = t\isubs{x}{u}$.
\item If $n_{x,t}>0$, then we have different cases.

  \begin{enumerate}
  \item Suppose  $t$ can be written as $C\cwc{x}$, for some
    context $C$. Then
        $t[x/u] \Rew{\us} C\cwc{u}[x/u]$ and 
        $n_{x,C\cwc{u}} < n_{x,t}$. By the \ih\
        $C\cwc{u}[x/u] \Rewn{\m} C\cwc{u}\isubs{x}{u}$.
        Since $t\isubs{x}{u} = C\cwc{u}\isubs{x}{u}$,
        then   $t[x/u] \Rewn{\m} t\isubs{x}{u}$. 
            
      \item Otherwise, suppose  $t$ can  be written as
        $t=C\cwc{\mX_{\Del}}\ (x \in \Del)$, for some context
        $C \neq \Box[y_1/v_1]\ldots[y_n/v_n]\ (n \geq 0)$.
        Then $t[x/u] \Rew{\usm} C\cwc{\mX_{\Del}[x/u]}[x/u]$ and 
        $n_{x,C\cwc{\mX_{\Del}[x/u]}} < n_{x,t}$. By the \ih\
        $C\cwc{\mX_{\Del}[x/u]}[x/u] \Rewn{\m} C\cwc{\mX_{\Del}[x/u]}\isubs{x}{u}$.
        Since $t\isubs{x}{u} = C\cwc{\mX_{\Del}[x/u]}\isubs{x}{u}$,
        then   $t[x/u] \Rewn{\m} t\isubs{x}{u}$.

      \item Otherwise,  $t$ can only be written as
        $\mX_{\Del}[y_1/u_1]..[y_n/u_n]\ (x \in \Del)$
        for some $n \geq 0$.
        Remark that $x \notin \fv(u_i)$ for all $1 \leq i \leq n$,
        otherwise we would be in the previous case. 
        Then $t[x/u]= \mX_{\Del}[y_1/u_1]..[y_n/u_n][x/u] =_{\C}  
            \mX_{\Del}[x/u][y_1/u_1]..[y_n/u_n]
            = \mX_{\Del}\isubs{x}{u}[y_1/u_1\isubs{x}{u}]..[y_n/u_n\isubs{x}{u}]
            = \mX_{\Del}[y_1/u_1]..[y_n/u_n]\isubs{x}{u}$.

      \end{enumerate}
\end{itemize}
\end{proof}

\ignore{So here there is a  subtle problem since $t \Rewn{\m} t'$
  means $t = t'$ and not $t =_{\C} t'$ when the number of steps
  from
$t$ to $t'$ is $0$. I've just realised that I
  have the same problem in  my $\les$-paper. The solution is just
  to talk about the reflexive closure on equivalence classes. Do
  you agree? }
  
\ignore{Yes,  I agree. We rewrite modulo the equation in the
  transitive case so I think we should also do the same in the reflexive case (it's funny how it's always the outlying cases). Besides solving the problem, this seems the correct notion to me. I've added the definition in the syntax section but there is probably a nicer way to do it.}

\ignore{If there's space, explain $\m$...
 {system $\sm = \{ \us, \usm, \Gc\}$ modulo
$\Es$ can be used as a function on $\Es$-equivalence classes.}}

We now  remark that the system $\sm = \{ \us, \usm, \Gc\}$ modulo
$\Es$ can be used as a function on $\Es$-equivalence classes.

\begin{lemma}
\label{l:unique-es-nf}
The $\m$-normal forms of $\mterms$-metaterms exist and are unique modulo $\Es$. 
\end{lemma}

\begin{proof}
  The reduction relation  $\Rew{\m}$ can be shown to be  terminating
  by associating to each $\mterms$-metaterm 
a measure 
which does not change by $\Es$ but strictly decreases by
$\Rew{\sm}$ (\cref{l:sub-terminating-metaterms} in ~\cref{app:confluence}).   Thus, $\m$-normal forms  of $\mterms$-metaterms exist. 
Moreover,   $\Rew{\m}$  is
locally confluent and locally coherent
(\cref{l:lm-local-confluent-coherent}
in~\cref{app:types}).  Therefore, by~\cite{JK86},
  $\Rew{\m}$  is  confluent  on $\mterms$-metaterms  and  hence $\m$-normal forms
  of $\mterms$-metaterms are unique
  modulo $\Es$-equivalence.
\end{proof}

Moreover, the following properties are  straightforward:

\begin{lemma}
\label{l:s-distributes} Let $u$ and $v$ be metaterms. Then, 
\begin{enumerate}
\item $\m(u\ v) = \m(u)\ \m(v)$.
\item $\m(\l x.u) = \l x.\m(u)$.
\end{enumerate}
\end{lemma}

\begin{lemma}
\label{l:es-application-lambda-composition}
Let  $t=f(t_1,\ldots,t_n)$, where $f$  is a  $\l$-abstraction, an
application     or      a     substitution     operator.     Then
$\m(f(t_1,\ldots,t_n)[x/u]) = \m(f(t_1[x/u],\ldots,t_n[x/u]))$.
\end{lemma}

\ignore{
  \begin{proof} First of all remark that $t=f(t_1,\ldots,t_n)$
means either $t=\l x. t_1$ or $t = t_1\ t_2$ or $t= t_1[y/t_2]$.
We reason by induction on the number  $n_{x,t}$ of free occurrences of 
$x$ in $t$.
\begin{itemize}
\item If $n=0$, then $\m(t[x/u]) =_{\Gc} \m(t) =_{\Gc}
  \m(f(t_1[x/u],\ldots,t_n[x/u]))$. 
\item If $n_{x,t}>0$, then we have two cases.
      \begin{itemize}
      \item There is a reduction 
            $t[x/u] \Rew{\lm} t'[x/u]$ such that 
            $n_{x,t'} < n_{x,t}$ so that 
           $t_i = C\cwc{x}$ or $t_i = C\cwc{\mX_{\Del}}$ with $x
           \in \Del$ and $C \neq \Box[y_1/v_1]\ldots[y_n/v_n], n
           \geq 0 $
           for some $i \in [1..n]$. As a consequence
           $t' = f(t_1,\ldots,t'_i,\ldots,t_n)$ with 
           $t'_i = C\cwc{u}$ or $t'_i = C\cwc{\mX_{\Del}[x/u]}$. Thus
           \[ \begin{array}{ll}
              \m(f(t_1,\ldots,t_i,\ldots,t_n)[x/u]) & = \\
              \m(f(t_1,\ldots,t'_i,\ldots,t_n)[x/u]) & =_{\ih} \\
              \m(f(t_1[x/u],\ldots,t'_i[x/u],\ldots,t_n[x/u])) & =_{\ih} \\
              \m(f(t_1[x/u],\ldots,t_i[x/u],\ldots,t_n[x/u])) & =_{\ih} \\
              \end{array}\]
      \item There is no such reduction so that is of the form
            $t = t_1[y/t_2]$ with $x \in
            \fv(t_1)$ and $x \notin
            \fv(t_2)$. We thus have
\[ \begin{array}{ll}
  \m(f(t_1,t_2)[x/u]) & = \\
  \m( t_1[y/t_2][x/u]) &   =_{\C} \\
  \m(t_1[x/u][y/t_2]) & =_{\Gc} \\
  \m(t_1[x/u][y/t_2[x/u]]) & = \\
  \m(f(t_1[x/u], t_2[x/u]))\\
  \end{array} \]
\end{itemize}
\end{itemize}
\end{proof}
}

\begin{lemma}
\label{l:forma-snf}
Let $t$ be a $\mterms$-metaterm in  $\m$-normal form.
Then it has one of the following forms:
\begin{itemize}
\item $t =x$, $t= t_1\ t_2$, or $t = \l y. t_1$ where $t_1$ and $t_2$ are in $\m$-normal form.
\item $t = \mX_\Del [x_1/u_1]\ldots[x_n/u_n]$, where $n
  \geq 0$ and every $u_i$ is in $\m$-normal form and $x_i \in
  \Del$ and $x_i \notin \fv(u_j)$ for all $i, j \in [1,n]$ s.t. $i < j$.
\end{itemize}
\end{lemma}

Remark that metaterms   in  $\m$-normal  form   have  all  explicit
  substitutions  directly  above  metavariables. Thus in particular   terms
  without metavariables in $\m$-normal form have no explicit substitutions at all.
From now on, we write $\m(t)$ to denote the (unique)
$\m$-normal form of the $\mterms$-metaterm $t$.

\subsection{Confluence}
\label{s:confluence}

While confluence on terms    always holds for calculi  with   explicit
substitutions, confluence on metaterms is often based 
on 
some possible form of  interaction between   substitutions,   such  as  in
$\lambda\sigma$~\cite{ACCL91}   or  $\lambda_{ws}$~\cite{guillaume01}. 
To illustrate this requirement, let  us consider the typical diverging
example adapted to $\lm$-reduction:

\[  t\isubs{y}{v}[x/u\isubs{y}{v}]  \LRewn{\lm} ((\l x. t)\ u)[y/v] \Rew{\B} t[x/u][y/v] \] 

This diagram can be closed using  full composition with the sequence
$t[x/u][y/v] \linebreak \Rewplus{\lm} t[x/u]\isubs{y}{v} =
t\isubs{y}{v}[x/u\isubs{y}{v}]$. 

However,  while  de  Bruijn   notation  for  $\l$-terms  allows  a
canonical  representation  of bound  variables   given by  a
certain  order  on  their  natural numbers,  calculi  with  named
variables  suffer  from the  following  (also typical)  diverging
example:

\[  \mX_{x,y}[y/v][x/z]  \LRewn{\lm} ((\l x. \mX_{x,y})\ z)[y/v] \Rew{\B} \mX_{x,y}[x/z][y/v] \] 
 
The $\mterms$-metaterms $\mX_{x,y}[y/v][x/z]$ and $\mX_{x,y}[x/z][y/v]$ are
equal modulo permutation of \emph{independent substitutions},
thus justifying the introduction of the equation $\C$ in the definition
of the calculus for metaterms. 

One possible technical tool to show confluence for $\mterms$-metaterms
is the use of another confluent calculus well-related to the
$\lm$-calculus. We prefer to give a self-contained argument, and so
adapt a proof based on Tait and Martin-L\"of's technique: define a
simultaneous reduction relation denoted $\paralp{\lm}$; prove that $\lm$
can be projected to $\paralp{\lm}$ on $\m$-normal forms; show that
$\paralnp{\lm}$ has the diamond property; and finally conclude.

\begin{definition}
\label{d:paralp}
The relation $\paralp{}$ on $\mterms$-metaterms in $\m$-normal form is
given by:
\begin{itemize}
\item $x \paralp{} x$
\item If $t  \paralp{} t'$, then $\l x. t  \paralp{} \l x. t'$
\item If $t  \paralp{} t'$ and 
         $u \paralp{} u'$, then $t\ u  \paralp{} t'\ u'$
\item If $t  \paralp{} t'$ and 
         $u \paralp{} u'$, then $(\l x. t)\ u  \paralp{} \m(t'[x/u'])$
\item If $u_i \paralp{} u'_i$ and $x_j \notin \fv(u_i)$ for all $i, j \in [1,n]$, then $\mX_\Del [x_1/u_1]\ldots[x_n/u_n] \paralp{} \mX_\Del [x_1/u'_1]\ldots[x_n/u'_n]$
\end{itemize}
\end{definition}

The relation $\paralp{\lm}$ is defined by $t  \paralp{\lm} t' \mbox{ iff } \exists s,s' \mbox{ s.t. }
t=_{\Es} s \paralp{} s' =_{\Es} t'$.
We use  $\paralnp{\lm}$
to 
  denote the reflexive closure of $\paralp{\lm}$ and thus
  $t \paralnp{\lm} t'$  in $0$ steps means  $t=_{\Es} t'$.
The following properties are straightforward.

\begin{remark}
  \label{l:simple-properties-es} The reduction relation $ \paralp{\lm}$
  enjoys the following properties: 
\begin{description}
\item[(Reflexivity)] $t \paralp{\lm} t$ for every $\mterms$-metaterm $t$ in $\m$-normal form. 
\ShaneIgnore{Induct on  the structure of $t$, using the second two cases in the definition of $\paralp{}$ for the induction cases.}
\item[(Closure by contexts)] if
$t \paralp{\lm} t'$, then
$u=C\cwc{t} \paralp{\lm} C\cwc{t'}=u'$ 
whenever $u$ and $u'$ are $\m$-normal forms. 
\ShaneIgnore{Use your proof for $\les$ replacing $\es$ with $\s$.}
\ShaneIgnore{I induct on  $n$, proving the base case $n=1$. The induction case is easier then since we are guaranteed to have an application or substitution where some $t_i$s are in one branch and the rest are in the other. We apply the induction hypothesis on both sides and then the third or fifth rule defining $\paralp{}$ to join the application/substitution. We then ``build'' the rest of the term $u'$ using the other cases in the definition of $\paralp{}$.}
\end{description}
\end{remark}

\begin{lemma}
\label{l:standard} \mbox{}
The reflexive and transitive closures
of $\paralp{\lm}$ and $\Rew{\lm}$ are the same relation.
\end{lemma}

\begin{proof}
To show $\paralnp{\lm} \subseteq\  \Rewn{\lm}$
we first show that $t \paralp{} t'$ implies $t \Rewn{\lm}
  t'$ by induction on the definition of $\paralp{}$. 
  
\begin{itemize}
\item $t = x \paralp{} x = t'$. Then $t \Rewn{\lm} t'$.
\item $t = \lambda x. u \paralp{} \l x. u' = t'$ where $u
  \paralp{} u'$. By the \ih\  $u \Rewn{\lm} u'$. Therefore, $\l x. u \Rewn{\lm} \l x. u'$.
\item $t = u\ v \paralp{} u'\ v' = t'$ where $u \paralp{} u'$ and
  $v \paralp{} v'$. By the \ih\  $u \Rewn{\lm} u'$ and $v \Rewn{\lm} v'$. Therefore, $u\ v \Rewn{\lm} u'\ v'$.
\item $t  = (\l  x. u)  v \paralp{} \m(u'[x/v'])  = t'$  where $u
  \paralp{} u'$ and $v \paralp{}  v'$. By \ih\, $u \Rewn{\lm} u'$
  and $v \Rewn{\lm} v'$. Therefore,  $$(\l x. u) v \Rewn{\lm} (\l
  x.  u')   v'  \Rew{\lm}  u'[x/v']   \Rewn{\lm}  \m(u'[x/v']).$$
  \item  $t = \mX_\Del  [x_1/u_1]\ldots[x_n/u_n] \paralp{}
    \mX_\Del   [x_1/u'_1]\ldots[x_n/u'_n]   =   t'$  where   $u_i
    \paralp{} u'_i$ and  $x_i \notin \fv(u_j)$ for all  $i, j \in
    [1,n]$.  By the  \ih\  $u_i  \Rewn{\lm} u'_i$  for all  $i \in
    [1,n]$ therefore $t \Rewn{\lm} t'$.

\end{itemize}

Now, suppose $t \paralp{\lm}^n t'$. We reason by
induction on  the number of
steps $n$. 
For $n=0$, $t =_{\Es} t'$ so that $t \Rewn{\lm} t'$ and we are done. 
Assume $n = k+1$ and $t \paralp{\lm} s \paralp{}^k t'$.
Then $t =_{\Es} t_1 \paralp{} s_1 =_{\Es} s \paralp{}^k t'$. 
We have $t =_{\Es} t_1 \Rewn{\lm} s_1 =_{\Es} s  \Rewn{\lm} t'$
by the previous point and the \ih\ 
so we conclude $t \Rewn{\lm}
  t'$.

\medskip

To show $\Rewn{\lm} \subseteq\ \paralnp{\lm}$
we first show that $t \Rew{\lm} t'$ implies $t \paralp{\lm} t'$ by
induction on $\Rew{\lm}$ using~\cref{l:simple-properties-es}. 
Now, suppose $t \Rewn{\lm} t'$ in $n$ steps. We conclude by
a simple  induction on  $n$.

\end{proof}

A consequence of~\cref{l:preservation-variables} and the previous  lemma is that
$t \paralp{\lm} t'$ implies $\fv(t') \subseteq \fv(t)$.
  \begin{lemma}
\label{l:stable-paralp}
Let $t,t',u,u'$ be $\mterms$-metaterms.
If $t \paralp{\lm} t'$ and $u \paralp{\lm} u'$,
then $\m(t[x/u]) \paralp{\lm} \m(t'[x/u'])$.
\end{lemma}

\begin{proof} We have $t=_{\Es} s_1 \paralp{} s_2 =_{\Es}
    t'$ so that it is sufficient to show that $s_1 \paralp{} s_2$
and $u \paralp{\lm} u'$
imply $\m(s_1[x/u]) \paralp{\lm} \m(s_2[x/u'])$ since then
$\m(t[x/u]) =_{\Es} \m(s_1[x/u])
\paralp{\lm} \m(s_2[x/u']) =_{\Es} \m(t'[x/u'])$.
We reason by induction on $s_1 \paralp{} s_2$. 
\begin{itemize}
  \item If $x \paralp{} x$, then 
           $\m(x[x/u]) = \m(u) = u \paralp{\lm} u' = 
                          \m(u') = \m(x[x/u'])$. 
  \item If $y \paralp{} y$, then 
        $\m(y[x/u]) = y  \paralp{\lm}
                       y = \m(y[x/u'])$ holds by definition. 
    
  \item If $t_1\ t_2 \paralp{} t'_1\ t'_2$, where
          $t_1 \paralp{} t'_1$ and $t_2 \paralp{} t'_2$,  then 
           \[ \begin{array}{ll}
            \m((t_1\ t_2)[x/u]) & = (\mbox{\cref{l:es-application-lambda-composition}})\\
            \m(t_1[x/u])\ \m(t_2[x/u]) &  \paralp{\lm} (\ih)\\
            \m(t'_1[x/u'])\ \m(t'_2[x/u']) & = (\mbox{\cref{l:es-application-lambda-composition}})\\
            \m((t'_1\ t'_2)[x/u'])
            \end{array} \]               
  \item If $\l y. v \paralp{} \l y. v'$, where
          $v \paralp{} v'$, then
           \[ \begin{array}{ll}
            \m((\l y. v)[x/u]) & = (\mbox{\cref{l:es-application-lambda-composition},~\cref{l:s-distributes}}) \\
            \l y. \m(v[x/u])   & \paralp{\lm} (\ih) \\ 
            \l y. \m(v'[x/u']) & =(\mbox{\cref{l:es-application-lambda-composition},~\cref{l:s-distributes}}) \\
            \m((\l y. v')[x/u']) 
            \end{array} \]  

  \item If $(\l y. t_1)\ v \paralp{} \m(t'_1[y/v'])$, where
          $t_1 \paralp{} t'_1$ and $v \paralp{} v'$, then 

           \[ \begin{array}{ll}
            \m(((\l y. t_1)\ v)[x/u]) & = (\mbox{\cref{l:es-application-lambda-composition}}) \\
            \m((\l y. t_1)[x/u])\ \m(v[x/u]) & = (\mbox{\cref{l:es-application-lambda-composition},~\cref{l:s-distributes}})\\ 
            (\l y. \m(t_1[x/u]))\ \m(v[x/u]) & \paralp{\lm}
            (\ih) \\
           \m(\m(t'_1[x/u'])[y/\m(v'[x/u'])]) & = \\ 
           \m( t'_1[x/u'][y/v'[x/u']]) & = (\mbox{\cref{l:es-application-lambda-composition}}) \\
           \m( t'_1[y/v'][x/u']) & = \\
           \m(\m( t'_1[y/v'])[x/u'])
           \end{array} \]  
        \item If $\mX_\Del [x_1/u_1]\ldots[x_n/u_n] \paralp{} \mX_\Del [x_1/u'_1]\ldots[x_n/u'_n]'$ where $u_i \paralp{} u'_i$ and $x_i \notin \fv(u_j)$ for all $i, j \in [1,n]$ then we reason by induction on $n$.

Note that $\m(u)=u$ and $\m(u')=u'$. We use~\cref{l:simple-properties-es} throughout.

\begin{description}
\item[--] 
For $n=0$ we have two cases.

If $x \notin \Del$ then $\m(\mX_\Del[x/u]) = \mX_\Del \paralp{\lm} \mX_\Del = \m(\mX_\Del[x/u'])$.

If $x \in \Del$ then $\m(\mX_\Del[x/u]) = \mX_\Del[x/\m(u)] \paralp{\lm}  \mX_\Del[x/\m(u')] = \m(\mX_\Del[x/u'])$.

\item[--] For $n>0$ we consider the following cases.

If $x \notin \fv(\mX_\Del [x_1/u_1]\ldots[x_n/u_n])$ then $x \notin \fv(\mX_\Del [x_1/u'_1]\ldots[x_n/u'_n])$ and thus
$$
\begin{array}{ccc}
\m(\mX_\Del [x_1/u_1]\ldots[x_n/u_n][x/u])\\
=\\
\mX_\Del [x_1/u_1]\ldots[x_n/u_n] & \paralp{\lm} & \mX_\Del [x_1/u'_1]\ldots[x_n/u'_n]\\
&&=\\
&&\m(\mX_\Del [x_1/u'_1]\ldots[x_n/u'_n][x/u'])
\end{array}
$$

If $x \in \fv(\mX_\Del [x_1/u_1]\ldots[x_n/u_n])$ then
$$
\begin{array}{ll}
\m(\mX_\Del [x_1/u_1]\ldots[x_n/u_n][x/u]) & = \mbox{\cref{l:es-application-lambda-composition} $n$ times}\\
\m(\mX_\Del [x/u] [x_1/u_1[x/u]]\ldots[x_n/u_n[x/u]]) & = \\
\m(\mX_\Del [x/u]) [x_1/\m(u_1[x/u])]\ldots[x_n/\m(u_n[x/u])] & \paralp{\lm} \mbox{\ih\  $n$ times}\\
\m(\mX_\Del [x/u']) [x_1/\m(u'_1[x/u'])]\ldots[x_n/\m(u'_n[x/u' ])] & = \\
\m(\mX_\Del [x/u'] [x_1/u'_1[x/u']]\ldots[x_n/u'_n[x/u']]) & = \mbox{\cref{l:es-application-lambda-composition} $n$ times}\\
\m(\mX_\Del [x_1/u'_1]\ldots[x_n/u'_n][x/u'])& 
\end{array}
$$
\end{description}

\end{itemize}
\end{proof}

\begin{lemma}[Projecting $\Rew{\lm}$ into $\paralp{\lm}$]
  \label{l:rew-in-paralp}
  Let $s,s'$ be $\mterms$-metaterms. 
If $s \Rew{\lm} s'$ then $\m(s) \paralp{\lm} \m(s')$. 
\end{lemma}

\begin{proof}
  If $s \Rew{\m}  s'$, then $\m(s) =  \m(s')$  holds  by~\cref{l:unique-es-nf}.      Thus,  $\m(s)
  \paralp{} \m(s')$ and $\m(s)
  \paralp{\lm} \m(s')$ holds by definition.
  If $s=_{\Es} s'$, then $\m(s) =_{\Es}  \m(s')$
  by~\cref{l:unique-es-nf}.
  Then $\m(s) =_{\Es}  \m(s')\paralp{} \m(s')$
  implies $\m(s)
  \paralp{\lm} \m(s')$.
  It remains to show that  $s \Rew{\B} s'$ implies
  $\m(s) \paralp{\m(s)} \m(s')$.
  We reason by induction on $s$.
\begin{itemize}
\item If $s = (\l x. t)\ u \Rew{\B} t[x/u] = s'$, then
      $\m(s) = (\l x. \m(t))\ \m(u) \paralp{}
  \m(\m(t)[x/\m(u)]) \EqEs \m(t[x/u])$.
  Then $\m(s) \paralp{\m(s)} \m(s')$.
\item If $s = t\ u \Rew{\B} t'\ u=s'$, where $t \Rew{\B} t'$, 
      then by the \ih\  we get $\m(t\ u) =_{\mbox{\cref{l:s-distributes}}} \m(t)\ \m(u) \paralp{\lm}  \m(t')\ \m(u) =_{\mbox{\cref{l:s-distributes}}} 
      \m(s')$.
\item If $s = t\ u \Rew{\B} t\ u'$, where $u \Rew{\B} u'$, 
      then by the \ih\  we get 
      $\m(t\ u) =_{\mbox{\cref{l:s-distributes}}} \m(t)\ \m(u) \paralp{\lm}  \m(t')\ \m(u)
      =_{\mbox{\cref{l:s-distributes}}} 
      \m(s')$.
\item If $s = t\ u \Rew{\B} t\ u'$, where $u \Rew{\B} u'$, 
      then by the \ih\  we get 
       $\m(t\ u) =_{\mbox{\cref{l:s-distributes}}} \m(t)\ \m(u) \paralp{\lm} \m(t)\ \m(u') =_{\mbox{\cref{l:s-distributes}} }
      \m(s')$.
\item If $s = \l x. t \Rew{\B} \l x. t'$, where $t \Rew{\B} t'$, then
      by the \ih\  we get 
       $\m( \l x. t) =_{\mbox{\cref{l:s-distributes}}} \l x. \m(t) \paralp{\lm} \l x. \m(t') =_{\mbox{\cref{l:s-distributes}} }
      \m(s')$.
\item If $s = t[x/u] \Rew{\B} t'[x/u]$, where $t \Rew{\B} t'$, then 
      $\m(t) \paralp{\lm} \m(t')$ by the \ih\ and 
  $\m(u) \paralp{\lm} \m(u)$ by definition.
  By~\cref{l:stable-paralp}
      $\m(s) \EqEs \m(\m(t)[x/\m(u)]) \paralp{\lm} 
      \m(\m(t')[x/\m(u)])  \EqEs \m(s')$. 
\item If $s = t[x/u] \Rew{\B} t[x/u']$, where $u \Rew{\B} u'$, then 
      $\m(u) \paralp{\lm}  \m(u')$ by the \ih\ and 
  $\m(t) \paralp{\lm} \m(t)$ by definition.
 By~\cref{l:stable-paralp} we conclude 
      $\m(s)  \EqEs  \m(\m(t)[x/\m(u)]) \paralp{\lm} 
      \m(\m(t)[x/\m(u')])  \EqEs  \m(s')$.
\end{itemize}

Finally, one concludes that $s \Rew{\lm}  s'$ implies 
$\m(s) \paralp{ \lm } \m(s')$. 
\end{proof}

From~\cref{l:rew-in-paralp}, we
conclude  that   $\m$    projects   $\Rew{\lm}$  into
$\paralp{\lm}$.

\begin{lemma}
\label{l:paralp-diamond}
The relation $\paralp{\lm}$ has the diamond property,  \ie\
if $t_1 \lparalp{\lm} t \paralp{\lm} t_2$, then there is
$t_3$ such that $t_1 \paralp{\lm} t_3  \lparalp{\lm} t_2$.
\end{lemma}

\begin{proof}
By  induction on the definition of $\paralp{\lm}$. We organize the proof as follows.

\begin{enumerate}
\item We first prove that $t \lparalp{} u =_{\Es} u'$ 
implies $t =_{\Es} t' \lparalp{} u'$ by induction on the
definition of  $t \lparalp{} u$.

\item \label{pdos} We then conclude  that $t \lparalp{\lm} u =_{\Es} u'$ 
implies $t =_{\Es} t' \lparalp{} u'$ using the previous point.

\item \label{ptres} We now prove that $t_1 \lparalp{} t \paralp{} t_2$
implies $t_1 \paralp{\lm} t \lparalp{\lm} t_2$.

\begin{itemize}

\item Let us consider 
\[ (\l x. t_1)\  u_1 \lparalp{} (\l x. t)\ u \paralp{} \m(t_2[x/u_2]) \]

where $t \paralp{} t_1$ and $t \paralp{} t_2$ and
 $u \paralp{} u_1$ and $u \paralp{} u_2$. By the \ih\  we
know there are $t_3$ and $u_3$ such that
$t_1 \paralp{\lm} t_3$ and $t_2 \paralp{\lm} t_3$ and
 $u_1 \paralp{\lm} u_3$ and $u_2 \paralp{\lm} u_3$ so that
in particular 
$t_1 =_{\Es} w_1 \paralp{} w_3 =_{\Es} t_3$ and 
$u_1 =_{\Es} w'_1 \paralp{} w'_3 =_{\Es} u_3$. 
We have 

\[ (\l x. t_1)\  u_1 =_{\Es}
   (\l x. w_1)\  w'_1 \paralp{}
   \m(w_3[x/w'_3 ])  =_{\Es}
   \m(t_3[x/u_3]) \]
 
and~\cref{l:stable-paralp} gives 

\[ \m(t_2[x/u_2]) \paralp{\lm} 
   \m(t_3[x/u_3]) \]

\item Let us consider 
 
\[ \m(t_1[x/u_1]) \lparalp{} (\l x. t)\ u \paralp{} \m(t_2[x/u_2]) \] 

where $t \paralp{} t_1$ and $t \paralp{} t_2$ and
 $u \paralp{} u_1$ and $u \paralp{} u_2$. By the \ih\  we
know there are $t_3$ and $u_3$ such that
$t_1 \paralp{\lm} t_3$ and $t_2 \paralp{\lm} t_3$ and
 $u_1 \paralp{\lm} u_3$ and $u_2 \paralp{\lm} u_3$.
Then,~\cref{l:stable-paralp} gives 

\[ \m(t_1[x/u_1]) \paralp{\lm} \m(t_3[x/u_3]) \lparalp{\lm} \m(t_2[x/u_2])  \] 

\item All the other cases are straightforward
using~\cref{l:simple-properties-es}.

\end{itemize}

\item We finally prove the diamond property as follows. 
Let $t_1 \lparalp{\lm} t =_{\Es} u \paralp{} u' =_{\Es}  t_2$.
      By point (\ref{pdos}) there is $u_1$ such that
      $t_1 =_{\Es} u_1 \lparalp{}  u$ and by  point (\ref{ptres}) there is
      $t_3$ such that $u_1 \paralp{\lm}  t_3 \lparalp{\lm} u'$.
      We conclude
      $t_1 \paralp{\lm}  t_3 \lparalp{\lm} t_2$.
\end{enumerate}

\end{proof}

\begin{corollary}
The $\lm$-reduction relation  is confluent on $\mterms$-terms and $\mterms$-metaterms.
\end{corollary}

\begin{proof}
  Take any  $\mterms$-metaterms $t, t_1, t_2$
  such that $t \Rewn{\lm} t_i$ for $i=1,2$.
\Cref{l:rew-in-paralp}
gives  $\m(t) \paralnp{} \m(t_i)$.
Since   the  diamond  property  implies 
  confluence~\cite{Baader-Nipkow}, then~\cref{l:paralp-diamond}
implies   confluence of $\paralp{}$.
Therefore, there is 
a metaterm $s$ s.t. $\m(t_i) \paralnp{} s$.
By~\cref{l:unique-es-nf} the (unique) $\m$-normal forms
of $t_i$ exist, so that 
$t_i \Rewn{\m} \m(t_i)$. 
We can then  close the
diagram by $t_i \Rewn{\m} \m(t_i) \Rewn{\lm} s$ using~\cref{l:standard}
to obtain $\m(t_i) \Rewn{\lm} s$ from $\m(t_i) \paralnp{} s$. 
\end{proof}


\section{Relating Partial Substitutions to Definitions}
\label{s:partial-definitions}

Partial substitution can be  related to calculi with definitions.
A definition can be understood as an abbreviation given
  by a name for a larger term which can be used several times 
in a program or a proof. A definition mechanism 
is essential for practical use; current implementations of 
proof assistants provide such a facility.

We  consider two  calculi,  the  first  one, which  we  call
$\lpar$,  appears  in~\cite{deBruijn87}  and  uses  a  notion  of
partial  substitution on  $\l$-terms, while  the second  one, which  we call
$\ldef$,  uses  partial substitutions  on $\mterms$-terms to  model definitions  and
combines  standard  $\beta$-reduction   with  the  rules  of the
substitution calculus $\m$. The general result of this section is
that normalisation  in $\lpar$ and  $\lm$ are equivalent  on \linebreak 
$\l$-terms and normalisation in  $\ldef$ and $\lm$ are equivalent
on $\mterms$-terms. More precisely,  for every  $\l$-term $t$, $t
\in \SN{\lpar}$ if  and only if $t \in \SN{\lm}$  and for every 
$\mterms$-term  $t$,  $t  \in  \SN{\ldef}$  if and  only  if  $t  \in
\SN{\lm}$.   Thus,  the $\lm$-calculus  can  be  understood as  a
concise  and simple  language implementing  partial  and ordinary
substitution,  both in implicit  and explicit  style at  the same
time.

\subsection{The partial $\lpar$-calculus}
\label{s:lpar}

Terms of the partial $\lpar$-calculus are $\l$-terms. The operational
semantics of the $\lpar$-calculus is given by the following rules:

\[ \begin{array}{|llll|}
\hline
(\l x. C\cwc{x}_\phi)\ u    &\Rew{\bp}  & (\l x. C\cwc{u}_\phi)\ u & \mbox{ if } \{x\} \cup \fv(u) \subseteq \phi\\   
(\l x. t)\ u    &\Rew{\BGc}   & t & \mbox{ if } x \notin \fv(t) \\
\hline
\end{array} \]

We consider the following translation from  $\l$-terms to
$\mterms$-terms: 


\[ \begin{array}{lclll}
   \begin{array}{lll}
     \parlm{x} & := & x \\  
     \parlm{\l x. t} & := & \l x. \parlm{t}
   \end{array}
   &\sep&
   \parlm{t\ u} & := & \left \{ \begin{array}{ll}
                               \parlm{t}\ \parlm{u} & \mbox{ if } t
                               \mbox{ is not a $\l$-abstraction} \\
                               \parlm{v}[x/\parlm{u}] & \mbox{ if } t=
                               \l x. v  \\
                               \end{array} \right . 
   \end{array} \]

\begin{lemma}
\label{l:from-lpar-to-lm}
Let $t,t'$ be $\l$-terms.
If $t \Rew{\lpar} t'$, then $\parlm{t} \Rewplus{\lm} \parlm{t'}$.  
\end{lemma}

\begin{proof}
By induction on $\Rew{\lpar}$. 
\end{proof}

\begin{corollary}
\label{c:lm-lpar}
Let $t$ be a  $\l$-term. 
If $t \in \SN{\lm}$, then $t \in \SN{\lpar}$.
\end{corollary}

\begin{proof}
Let  $t \in \SN{\lm}$ and suppose $t \notin \SN{\lpar}$.
Then, from an infinite $\lpar$-reduction sequence starting
at $t$ we can construct, by~\cref{l:from-lpar-to-lm},
an infinite $\lm$-reduction sequence starting at $\parlm{t}$.
Since $t \Rewn{\lm} \parlm{t}$, then $t \notin \SN{\lm}$,  which
leads to a contradiction. We thus conclude $t \in \SN{\lpar}$.
\end{proof}

The  converse reasoning also works. Define a 
translation from $\mterms$-terms  to $\l$-terms:

\[ \begin{array}{lll@{\hspace{2cm}}lll}
   \lmpar{x} & := & x & 
   \lmpar{\l x. t} & := & \l x. \lmpar{t}\\
   \lmpar{t\ u} & := & \lmpar{t}\ \lmpar{u} & 
   \lmpar{t[x/u]} & := & (\l x. \lmpar{t})\ \lmpar{u} \\ 
    \end{array} \]

Remark that $\lmpar{t}\isubs{x}{\lmpar{u}} =
\lmpar{t\isubs{x}{u}}$.

\begin{lemma} 
\label{l:from-lm-to-lpar}
Let $t,t'$ be $\mterms$-terms
such that $t \Rew{\lm} t'$.
If $t \Rew{\B} t'$, then $\lmpar{t} = \lmpar{t'}$.
If  $t \Rew{\m} t'$, then $\lmpar{t} \Rewplus{\lpar} \lmpar{t'}$. 
\end{lemma}

\begin{proof}
By induction on $\Rew{\lm}$. 
\end{proof}

\begin{corollary}
\label{c:lpar-lm}
Let $t$ be a  $\l$-term.
Then $t \in \SN{\lpar}$ if and only if $t \in \SN{\lm}$.
\end{corollary}

\begin{proof}
  If   $t \in \SN{\lm}$, then $t \in \SN{\lpar}$
  by~\cref{c:lm-lpar}.
For the converse, let  $t \in \SN{\lpar}$ and suppose $t \notin \SN{\lm}$.
Consider an infinite $\lm$-reduction sequence starting
at $t$. Since $\Rew{\B}$ is terminating,  such 
infinite reduction sequence must contain an infinite
number of  $\Rew{\m}$ steps. By~\cref{l:from-lm-to-lpar}
this gives  an infinite $\lpar$-reduction sequence starting at
$\lmpar{t}$. Since $t$ is a $\l$-term, then $\lmpar{t} = t$,
thus $t \notin \SN{\lpar}$ which leads to a contradiciton.
We conclude $t \in \SN{\lm}$.
\end{proof}

\subsection{The $\l$-calculus with definitions $\ldef$}
\label{s:calculus-with-definitions}

The syntax         of          the         $\l$-calculus         with
definitions $\ldef$~\cite{severi94definitions},  is isomorphic to
that of the $\lm$-calculus, where the use of a definition $x:=u$ in a term $v$,
denoted ${\tt let}\ x:=u\ {\tt in}\ v$, can be thought
as the term $v[x/u]$ in          $\lm$.          The         original
presentation~\cite{severi94definitions}    of   the   operational
semantics of $\ldef$ is given  by a reduction system which is not
a (higher-order) term rewriting system.   This is due to the fact
that given a  definition $x:=u$, the term $x$  can be reduced to
the term $u$, so that  reduction creates new free variables since
$\fv(u)$  does  not necessarily  belong  to  $\set{x}$. Here,  we
present $\ldef$ by  a set of reduction rules  which preserve free
variables  of  terms.   Moreover,  we  consider  a  more  general
reduction   system  where   any  $\beta$-redex   can   be  either
$\beta$-reduced  or  transformed   to  a  definition,  while  the
calculus  appearing in~\cite{severi94definitions} does  not allow
dynamic creation of definitions.

\[ \begin{array}{|llll|}
\hline
(\l x. t)\ u    &\Rew{\beta}  & t \isubs{x}{u} & \\   
(\l x. t)\ u    &\Rew{\B}   & t[x/u] & \\
t[x/u]          &\Rew{\Gc}  & t      & \mbox{ if } x \notin \fv(t) \\ 
C\cwc{x}_\phi[x/u]   &\Rew{\us}   & C\cwc{u}_\phi[x/u] & \mbox{ if } \{x\} \cup \fv(u) \subseteq \phi\\
\hline
\end{array} \]

The following relations between $\ldef$ and $\lm$ holds:

\begin{lemma} \mbox{}
\label{l:ldef-from/to-lm}
Let $t,t'$ be $\mterms$-terms.
\begin{itemize}
\item If $t \Rew{\ldef} t'$, then $t \Rewplus{\lm} t'$.
  \item If $t \Rew{\lm} t'$, then $t \Rewplus{\ldef} t'$.
\end{itemize}
\end{lemma}

\begin{proof}
The first point can be shown by induction on $\Rew{\ldef}$ using the fact that
any $\beta$ step can be simulated by $\B$ followed
by several $\us$ steps and one $\Gc$ step.
The second point is straightforward.
\end{proof}

We can then conclude that normalisation for
the $\ldef$-calculus and the  $\lm$-calculus are  equivalent:

\begin{corollary}
  \label{c:sn-lm-sn-ldef}
  Let $t$ be a $\mterms$-term.
Then $t \in \SN{\lm}$ if and only if  $t \in \SN{\ldef}$. 
\end{corollary}

\section{Relating Partial to Explicit Substitutions}
\label{s:PSN}

We now relate $\lm$ to a calculus with  explicit (local)
substitutions called $\les$~\cite{K07}, summarised below. We then give
a (dynamic) translation from $\lm$ to $\les$,
showing that each
$\lm$-reduction step can be simulated by a non-empty reduction
sequence in $\les$. This result  will
be used in particular in the forthcoming~\cref{s:charact-les-lm}
and~\cref{s:graphical}
to obtain different normalisation theorems. 

Terms of the $\les$-calculus are   $\mterms$-terms.
Besides $\alpha$-conversion, we consider the equations and reduction
rules in~\cref{f:les-rules}. 
Remark that  working modulo $\alpha$-conversion  allows us
  to assume implicitly some conditions to avoid capture of
    variables such as  for example $x \neq y$ and $y \notin
  \fv(v)$ in the reduction rule $\Lamb$.

\begin{figure}[htp]
\[ \begin{array}{|lcll|}
\hline 
\mbox{{\bf Equations}}: &&&\\
t[x/u][y/v]  & =_{\Com} &   t[y/v][x/u] &  \mbox{ if } y \notin \fv(u)\ \&\
x \notin \fv(v)\\ 
&&&\\
\mbox{{\bf Reduction Rules}}: &&&\\
(\l x. t)\ u    &\Rew{\B}   & t[x/u] & \\
x[x/u]          &\Rew{\Var}& u & \\
t[x/u]          & \Rew{\Gc} & t & \mbox{ if } x \notin \fv(t) \\ 
(t\ u)[x/v]     &\Rew{\Appi} &(t[x/v]\ u[x/v]) & \mbox{ if } x \in \fv(t)\ \&\  x \in \fv(u)\\
(t\ u)[x/v]     &\Rew{\Appii} &(t\ u[x/v])& \mbox{ if } x \notin \fv(t)\ \&\  x \in \fv(u)\\
(t\ u)[x/v]     &\Rew{\Appiii} &(t[x/v]\ u)& \mbox{ if } x \in \fv(t)\ \&\  x \notin \fv(u)\\
(\l y.t)[x/v]   &\Rew{\Lamb} &  \l y . t[x/v] & \\
t[x/u][y/v]   &\Rew{\Compi}& t[y/v][x/u[y/v]] & \mbox{ if } y \in
\fv(u)\ \&\ y \in \fv(t)\\ 
t[x/u][y/v]   &\Rew{\Compii}& t[x/u[y/v]] & \mbox{ if } y \in \fv(u)\ \&\ y
\notin \fv(t)\\
\hline 
\end{array} \]
\caption{Equations and reduction rules for $\les$}
\label{f:les-rules} 
\end{figure}

We  consider the equivalence  relation $\Es$
  generated   by  $\alpha$   and  $\Com$.
  The  \emph{rewriting   system}  containing 
  all the reduction rules except $\B$
   is denoted
  by $\s$. We write $\B\s$
  for $\B  \cup \s$.    
  We note $\Rew{\ALC}$, the \emph{reduction
    relation}  generated by the rules  $\{ \Appi,$
  $\Appii,$ $\Appiii,$ $\Lamb,$ $\Compi,$ $\Compii \}$,
  closed by \emph{all} contexts, and taken 
  \emph{modulo the equivalence relation}  $\Es$.
 The  \defn{reduction
    relations of the $\les$-calculus} are  generated by  $\s$ (resp.
  $\B\s$) \emph{modulo} $\Es$, and  is denoted
  by  $\Rew{\es}$  (resp.  $\Rew{\les}$),  where  {\tt   e}  means
  equational and {\tt  s} means substitution.
As expected, reduction  preserves free variables. 

\begin{lemma}[Preservation of Free Variables]
\label{s:preservation-free-variables} 
Let $t,t'$ be $\mterms$-terms.
If $t \Rew{\les} t'$, then  $\fv(t') \subseteq \fv(t)$. More precisely, 
$t \Rew{\Gc} t'$ implies $\fv(t') \subseteq \fv(t)$, in
  all other cases $\fv(t') = \fv(t)$.
\end{lemma}

We write $t [\vec{x}/\vec{u}]$ for $t [x_1/u_1] \ldots [x_n/u_n]$
where $\vec{x} = x_1, \ldots, x_n$, $\vec{u} = u_1, \ldots, u_n$,
and $x_i \notin \fv(u_j)$  where $i,j \in [1,\ldots,n]$. We write
$x_i$ to  denote an  arbitrary member of  $x_1, \ldots,  x_n$ and
similarly for  $u_i$. The concatenation of  two vectors $\vec{x}$
and   $\vec{y}$   is   written    as   $\vv{xy}$.   We   let   $t
[\vec{x}/\vec{u}]$ denote the  $=_\Com$-equivalence class which
arises by the reordering of the independent substitutions.

We extend the set  of variables with marked variables $\fresh{x},
\fresh{y}$, \emph{etc.}   This will be used to  denote binders of
certain garbage substitutions which  will be fresh \emph{i.e.} if
$t = v[\fresh{x}/u]$ then $\fresh{x} \notin \fv(v)$.
 Remark that $\Rew{\ALC}$ may only propagate 
garbage substitutions  through
abstractions  and  not through  applications  or inside  explicit
substitutions.

\begin{restatable}{lemma}{alclocallyconfluent}
The reduction relation $\Rew{\ALC}$ is locally confluent and locally coherent.
  \label{l:ALC-lc-lc}
\end{restatable}

\begin{proof} See~\cref{app:relation-les}.
\end{proof}

\begin{lemma}
\label{l:unique-ALC-nf}
The $\ALC$-normal forms of terms exist and are unique modulo $\Es$. 
\end{lemma}

\begin{proof}
  The system $\Rew{\es}$ is terminating~\cite{K07} and so
  in particular $\Rew{\ALC}$
  turns out to be  terminating.
  By~\cref{l:ALC-lc-lc}, $\Rew{\ALC}$ is locally confluent and locally coherent.
  Therefore by~\cite{JK86}, $\Rew{\ALC}$ is confluent and hence $\ALC$-normal
  forms are unique modulo $\Es$-equivalence.
\end{proof}

From now on, we can assume the existence
of a \defn{function $\ALC$} computing
the (unique) $\ALC$-normal form
of a $\mterms$-term, modulo $\Es$.
This allows us in particular
to define   the following     \defn{translation $\nombretra$}  from  $\mterms$-terms  to
$\ALC$-normal forms:  
$$
\begin{array}{lcll}
\tra{x} & :=& x\\
\tra{\l x. t} & :=& \l x. \tra{t}\\
\tra{t\ u} & :=& (\tra{t}\ \tra{u})[\fresh{y}/\tra{u}] & \mbox{ where } \fresh{y} \mbox{ is fresh}\\
\tra{t[y/u]} & :=& \ALC(\tra{t}[y/\tra{u}]) & \mbox{ if $y \notin \fv(t)$}\\
\tra{t[y/u]} & :=& \ALC(\tra{t}[y/\tra{u}][\fresh{y}/\tra{u}]) & \mbox {
  if $y \in \fv(t)$ where $\fresh{y}$ is fresh}
\end{array}
$$

Remark that the translation $\tra{\_}$ preserves free variables.  Remark also
that the translation of a closure $\tra{t[y/u]}$ with $y \in \fv(t)$
introduces a garbage substitution with binder $\fresh{y}$. The
translation of $\tra{t\ u}$ similarly introduces garbage. Intuitively,
this should not interfere with the PSN property as: i) the body of the
garbage is strongly normalising exactly when the body of the regular
substitution is; and ii) the garbage substitution can only interact
with substitutions above it as can the regular substitution so that
any resulting infinite sequences can occur in the regular substitution
as well.  We call these garbage substitutions which are introduced by
the translation \emph{idle substitutions} whilst the other
substitutions are called \emph{mobile substitutions}.


\report{The  translation  $\tra{\_}$  duplicates  all non-garbage  substitutions  and
function arguments, creating idle copies as garbage substitution.
The reasoning is as follows.

To    simulate     the   partial   substitution    
$\Rew{\us}$ of the  $\lm$-calculus, we use the  (local) subcalculus $\ALC$.
$\ALC$-normal forms are required in order
to prove a simulation (otherwise the reduction $((x\ x)\ x)[x/y] \Rew{\us} 
((y\ x)\ x)[x/y]$ cannot be simulated).

Now  consider the (partial)  reduction $x[x/z]  \Rew{\us} z[x/z]$
versus  the  (local) reduction  $x[x/z]  \Rew{\Var} z$.   Partial
reduction does  not remove the explicit  substitution $[x/z]$ but
local reduction  will correctly do  so as no free  occurrences of
the bound variable  $x$ lie beneath the body  of the substitution
$z$.   Therefore,   we   cannot  immediately   simulate   partial
substitution.   A   naive  solution  would  be   to  compose  the
translation with a reduction  to the garbage-free normal form but
this clearly  fails in the  general case; we should  instead keep
the garbage.

Finally,    consider    the     sequence    $(x\    x\    x)[x/y]
\Rewn{\lm} (y\ y\  y)[x/y]$. From the discussion above,
the  translation needs  to  both push  the explicit  substitution
inside  the  term  and  also  keep it  outside  awaiting  garbage
collection. Our solution is  to simply duplicate the substitution
in such a way as to solve this problem.

In  short, the  mobile substitutions  allow us  to  simulate most
$\Rew{\us}$   reductions  but   idle   substitutions  are
required to  simulate $\Rew{\us}$  reductions where the  last free
occurrence of the bound variable is replaced.}



The reader should notice that 
if $t$ is pure then $t = \Gc(\tra{t})$.
Also, if $t$ is in $\Rew{\lm}$-normal form, then $\Gc(\tra{t})$
is in $\Rew{\les}$-normal form so that $\tra{t}$ turns out to be
in $\SN{\les}$.

\ignore{
\begin{lemma}
\label{l:tra-pure-is-pure}
If $t$ is pure then:
\begin{enumerate}
\item $t = \Gc(\tra{t})$;
\item If $t$ is in $\Rew{\lm}$-normal form then $\Gc(\tra{t})$ is in $\Rew{\les}$-normal form; and
\item If $t$ is in $\Rew{\lm}$-normal form then $\tra{t} \in \SN{\les}$.
\end{enumerate}
\end{lemma}
\begin{proof}
\begin{enumerate}
\ShortLongProof
{
\item Induction on the structure of $t$.
\item Prove that $\tra{t}$ only has $\Gc$-redexes by induction on the structure of $t$. As $\Gc$-reductions do not create new redexes, any $\Rewn{\Gc}$-reduct of $\tra{t}$ has only $\Gc$-redexes. The result follows by $\Rew{\Gc}$ SN.
\item Follows from (2).
}
{
\item We prove by induction on the structure of $t$.
\begin{itemize}
\item $t = x$. Trivial.
\item $t = \l x.t'$. 

By \ih, $t' = \Gc(\tra{t'})$. Therefore, $t = \l x.t' = \l x.\Gc(\tra{t'}) = \Gc(\tra{t})$.
\item $t = (u\ v)$. 

By \ih, $u = \Gc(\tra{u})$ and $v = \Gc(\tra{v})$.  Therefore, $t = (u\ v) = (\Gc(\tra{u})\  \Gc(\tra{v})) = \Gc(\tra{u}\  \tra{v}) = \Gc((\tra{u}\  \tra{v})[z/\tra{u}]) = \Gc(\tra{t})$.
\end{itemize}
\item We first show that $\tra{t}$ has no $\B$-redexes by inducting over the structure of $t$.
\begin{itemize}
\item $t = x$. Trivial.
\item $t = \l x.t'$. $\tra{t} = \l x.\tra{t'}$. By \ih, $\tra{t'}$ has no $\B$-redexes and so neither has $\tra{t}$.
\item $t = (u\ v)$. $\tra{t} = (\tra{t}\ \tra{u})[z/\tra{u}]$.  By \ih, $\tra{u}$ and $\tra{v}$ have no $\B$-redexes. Since $t = (u\ v)$ has no $\B$-redexes, $u \neq \l x.u'$ so $\tra{u} \neq \l x.\tra{u'}$. Therefore, $\tra{t}$ has no $\B$-redexes.
\end{itemize}
We now argue about what other redexes $\tra{t}$ does have. All substitutions are garbage so there are no $\Var$ redexes. As the substitutions also lie above applications, there are no $\ALC$ redexes. Therefore $\tra{t}$ can only have $\Gc$ redexes and the same holds for any $\Rewn{\Gc}$-reduct of $\tra{t}$. The result follows by $\Rew{\Gc}$ SN.
\item
Follows from (2).
}
\end{enumerate}
\end{proof}
}

\ignore{
  \begin{lemma}
\label{l:sim-replace-var-with-subterm}
Let $t,t',v$ be $\mterms$-terms.
If $t \Rew{\les} t'$ then $t\isubs{z}{v} \Rew{\les} t'\isubs{z}{v}$.
\end{lemma}

\begin{proof}
By induction on the definition of $t  \Rew{\les}  t'$.
\end{proof}
}

\ignore{
\begin{proof}

By induction on the recursive  definition of $t
  \Rew{\les}  t'$.  By $\alpha$-equivalence,  we can  assume that
  $\fv(v)$  and $\bv(t)$  are distinct  sets and  that  $z \notin
  \bv(t)$.  We  first  consider  the  base  case  of  the
  induction  where reductions  $t \Rew{\les}  t'$ occur  at the root.

\begin{itemize}
\item $t = (\l x. u_1)\ u_2 \Rew{\les} u_1[x/u_2] = t'$.

$t\isubs{z}{v} = (\l x. u_1\isubs{z}{v})\ u_2\isubs{z}{v} \Rew{\les} u_1\isubs{z}{v}[x/u_2\isubs{z}{v}]$.

\item $t = x[x/u] \Rew{\les} u = t'$.

$t\isubs{z}{v} = (x[x/u])\isubs{z}{v} = x[x/u\isubs{z}{v}] \Rew{\les} u\isubs{z}{v}$.

\item $t = u_1[x/u_2] \Rew{\les} u_1 = t'$.

$t\isubs{z}{v} = (u_1[x/u_2])\isubs{z}{v} = u_1\isubs{z}{v}[x/u_2\isubs{z}{v}] \Rew{\les} u_1\isubs{z}{v}$.

\item $t = (u_1\ u_2)[x/u_3] \Rew{\les} (u_1[x/u_3]\ u_2[x/u_3]) = t'$.

$t\isubs{z}{v} = (u_1\isubs{z}{v}\
u_2\isubs{z}{v})[x/u_3\isubs{z}{v}]$\\  
$\Rew{\les} (u_1\isubs{z}{v}[x/u_3\isubs{z}{v}]\ u_2\isubs{z}{v}[x/u_3\isubs{z}{v}])$.

\item $t = (u_1\ u_2)[x/u_3] \Rew{\les} (u_1[x/u_3]\ u_2) = t'$. Similar.

\item $t = (u_1\ u_2)[x/u_3] \Rew{\les} (u_1\ u_2[x/u_3]) = t'$. Similar.

\item $t = (\l y.u_1)[x/u_2] \Rew{\les} \l y.u_1[x/u_2] = t'$.

$t\isubs{z}{v} = (\l y.u_1\isubs{z}{v})[x/u_2\isubs{z}{v}] \Rew{\les} (\l y.u_1\isubs{z}{v}[x/u_2\isubs{z}{v}])$.

\item $t = u_1[x/u_2][y/u_3] \Rew{\les} u_1[y/u_3][x/u_2[y/u_3]] = t'$.

$t\isubs{z}{v} = u_1\isubs{z}{v}[x/u_2\isubs{z}{v}][y/u_3\isubs{z}{v}]$\\
$\Rew{\les} u_1\isubs{z}{v}[y/u_3\isubs{z}{v}][x/u_2\isubs{z}{v}[y/u_3\isubs{z}{v}]]$.

\item $t = u_1[x/u_2][y/u_3] \Rew{\les} u_1[x/u_2[y/u_3]] = t'$. Similar.
\end{itemize}

The inductive cases are the following:

\begin{itemize}
\item $t = (u_1\ u_2) \Rew{\les} (u'_1\ u_2) = t'$.

$t\isubs{z}{v} = (u_1\isubs{z}{v}\ u_2\isubs{z}{v}) \Rew{\les}$ (\ih)  $(u'_1\isubs{z}{v}\ u_2\isubs{z}{v})$.

\item $t = (u_1\ u_2) \Rew{\les} (u_1\ u'_2) = t'$. Similar.

\item $t = \l y.u_1 \Rew{\les} \l y.u'_1 = t'$. Similar.

\item $t = u_1[x/u_2] \Rew{\les} u'_1[x/u_2] = t'$. Similar.


\item $t = u_1[x/u_2] \Rew{\les} u_1[x/u_2'] = t'$. Similar. 
\end{itemize}
\end{proof}

}

\ignore{
  \begin{corollary}
Let $t,t',v$ be $\mterms$-terms.
If $t \Rew{\les}^n t'$ then $t\isubs{z}{v} \Rew{\les}^n
t'\isubs{z}{v}$, for $n \geq 0$.
\end{corollary}
}
\ignore{\begin{proof}
Remark that if $t = t_1[x/u_1][y/u_2] =_{\Es} t_2[y/u_2][x/u_1] = t'$ then $t\isubs{z}{v}  =_{\Es} t'\isubs{z}{v}$ and use~\cref{l:sim-replace-var-with-subterm} to conclude.
\end{proof}
}

\ignore{
  As a particular case of the previous corollary, we have 
$t\isubs{z}{v} \Rewn{\ALC} \ALC(t)\isubs{z}{v}$.

\begin{lemma}
\label{l:isubs-ALCnf-correspondence}
Let $t,v$ be $\mterms$-terms.
Then $\ALC(t)\isubs{z}{\fresh{z}[\fresh{z}/v]}  = \ALC(t[z/v])$  if $z
\in \fv(t)$ and $v$ in $\ALC$-normal form.
\end{lemma}

\begin{proof}
By induction on  the structure of $t$.
\end{proof}
}
\ignore{
\begin{proof}
By induction on  the structure of $t$.
By $\alpha$-equivalence, we can assume that $\fv(v)$ and $\bv(t)$ are distinct sets.
\begin{itemize}

\item $t = z$. $\ALC(t)\isubs{z}{\fresh{z}[\fresh{z}/v]} = \fresh{z}[\fresh{z}/v] = \ALC(z[z/v])$.

\item $t = t_1\ t_2$. There are three cases depending on whether $z \in \fv(t_1)$, $z \in \fv(t_2)$, or both. We treat the case where $z \in \fv(t_1)$ and $z \notin \fv(t_2)$.

$\ALC(t)\isubs{z}{\fresh{z}[\fresh{z}/v]} = (\ALC(t_1)\isubs{z}{\fresh{z}[\fresh{z}/v]}\ \ALC(t_2)\isubs{z}{\fresh{z}[\fresh{z}/v]})$\\ $=$ (\ih) $\:(\ALC(t_1[z/v])\ \ALC(t_2)) = \ALC((t_1\ t_2)[z/v])$.

\item $t = \l y.t_1$. 
$\ALC(t)\isubs{z}{\fresh{z}[\fresh{z}/v]} = \l y.\ALC(t_1)\isubs{z}{\fresh{z}[\fresh{z}/v]}$\\
$=$ (\ih) $\:\l y. \ALC(t_1[z/v]) = \ALC((\l y. t_1)[z/v])$.

\item $t = t_1[x/t_2]$. There are three cases depending on whether $z \in \fv(t_1)$, $z \in \fv(t_2)$, or both. We treat the case where $z \in \fv(t_1)$ and $z \notin \fv(t_2)$. Note that $\ALC$ preserves free names.
  
  $\ALC(t[z/v])  =    \ALC(t_1[x/t_2][z/v])  = \ALC(t_1[z/v][x/t_2])\\ 
  =\ALC(\ALC(t_1[z/v])[x/t_2])         =         \mbox{(\ih)         }
  \:\ALC(\ALC(t_1)\isubs{z}{\fresh{z}[\fresh{z}/v]}[x/t_2])\\ 
  =\ALC(t_1\isubs{z}{\fresh{z}[\fresh{z}/v]}[x/t_2])                 =
  \ALC(t_1[x/t_2])\isubs{z}{\fresh{z}[\fresh{z}/v]}.$
\end{itemize}

\end{proof}
}

\ignore{
  \begin{lemma}
\label{l:adding-independent-subsi}
Let $t,t',v$ be $\mterms$-terms in $\ALC$-normal form.
If $t \Rewplus{\les} t'$  then:
\begin{enumerate}
\item \label{l:adding-independent-subsia} if $z \notin \fv(t)$ then $\ALC(t[z/v]) \Rewplus{\les} \ALC(t'[z/v])$;
\item \label{l:adding-independent-subsib} if $z \in \fv(t)$ and $z \in \fv(t')$ then $\ALC(t[z/v]) \Rewplus{\les} \ALC(t'[z/v])$;
\item \label{l:adding-independent-subsic} if $z \in \fv(t)$ and $z \notin \fv(t')$ then $\ALC(t[z/v]) \Rewplus{\les} t'$.
\end{enumerate}
\end{lemma}
\begin{proof} Note that by our conventions, $z \notin \bv(t)$.

\begin{enumerate}
\item Let $z  \notin \fv(t)$. Write $t$ as  $\l \vec{y}.s$, where
  $s$ is not a $\l$-abstraction.  As $t \Rewplus{\les} t'$ and $t
  = \l \vec{y}.s$, it must be that $s \Rewplus{\les} \l
  \vec{x}.s'$ with  $t' = \l \vec{y}. \l  \vec{x}.s'$, where $s'$
  is not  a $\l$-abstraction and  $s,s'$ in $\ALC$-normal
  form.     Therefore,   $\ALC(t[z/v])   =    \l   \vec{y}.s[z/v]
  \Rewplus{\les} \l   \vec{y}.(\l   \vec{x}.s')[z/v]
  \Rewn{\les}   \l   \vec{y}.\l   \vec{x}.s'[z/v]   =
  \ALC(t'[z/v])$.

\item[2/3.]

  We define $\U(p)  = p\isubs{z}{\fresh{z}[\fresh{z}/v]}$ for any
  $p$ in the sequence  $t \Rewplus{\les} t'$.
  Then, from the sequence
  $t = t_1 \Rew{\les} t_2 \Rew{\les} \ldots  \Rew{\les} t'$
  we can build by~\cref{l:sim-replace-var-with-subterm}
  a sequence of the form 
  $\U(t_1) \Rew{\les} \U(t_2) \Rew{\les}\ldots  \Rew{\les} \U(t')$. 

  We have now two cases: \\
  
If $z \in \fv(t')$ then 
$\U(t_1) = t\isubs{z}{\fresh{z}[\fresh{z}/v]} = \ALC(t[z/v]) \Rewplus{\les} \U(t') = t'\isubs{z}{\fresh{z}[\fresh{z}/v]} = \ALC(t'[z/v])$, using~\cref{l:isubs-ALCnf-correspondence} twice.

If $z \notin \fv(t')$ then $\U(t_1) = t\isubs{z}{\fresh{z}[\fresh{z}/v]} = \ALC(t[z/v]) \Rewplus{\les} \U(t') = t'\isubs{z}{\fresh{z}[\fresh{z}/v]} = t'$ using~\cref{l:isubs-ALCnf-correspondence} once.
\end{enumerate}
\end{proof}
}

\ignore{

\begin{restatable}{lemma}{technicallemmaalc}
Let $t,u,u'$ be $\mterms$-terms.
Let $u$ be in $\ALC$-normal form.
If $u \Rew{\les} u'$  then
$\ALC(t[z/u]) \Rewplus{\les} \ALC(t[z/u'])$.
\label{l:reduce-copies-of-subs}
\end{restatable}

\begin{proof}
Details in~\cref{app:relation-les}.
\end{proof}

\begin{corollary}
\label{c:reduce-copies-of-subs}
Let $t,u,u'$ be $\mterms$-terms. Let $u$ be in $\ALC$-normal form.
If $u \Rewplus{\les} u'$   then
$$\ALC(t[z/u]) \Rewplus{\les} \ALC(t[z/u']).$$
\end{corollary}

}

\ignore{
\begin{lemma}
\label{l:non-interfering-garbage}
Let $t,u,u'$ be $\mterms$-terms.
If $\fresh{x}$ is fresh for $t$ and $t'$ and 
 $\ALC(t[\fresh{x}/u]) \Rewplus{\les} \ALC(t'[\fresh{x}/u])$ 
 then  $\ALC(t) \Rewplus{\les} \ALC(t')$. \delia{pareceria ser *}
\end{lemma}
\begin{proof}
  The garbage substitution $[\fresh{x}/u]$ may not be copied during the sequence\linebreak $\ALC(t[\fresh{x}/u]) \Rewplus{\les} \ALC(t'[\fresh{x}/u])$.
  Moreover, $\ALC(t'[\fresh{x}/u])$
  does not erase any substitution,
  then this same sequence cannot discard $[\fresh{x}/u]$. It occurs exactly once in each term of the sequence. For any such term $p$, we define $\Er(p)$ as the term obtained by dropping the substitution $[\fresh{x}/u]$.

The proof  follows from the diagram  below where the  top line is
the       sequence      $t_1= \ALC(t[\fresh{x}/u])      \Rewplus{\les}
\ALC(t'[\fresh{x}/u])=t_n$.  The  squares   in  the  diagram  can  be
completed as either $\mbox{i) } t_i = \l \vec{z}.v_i[\fresh{x}/u]
\Rew{\les} \l \vec{z}.v_{i+1}[\fresh{x}/u]  = t_{i+1} \mbox{ with
}  v_i \Rew{\les}  v_{i+1}$ so  that $\Er(t_i)  =  \l \vec{z}.v_i
\Rew{\les}  \l  \vec{z}.v_{i+1}$;  or   $\mbox{ii)  }  t_i  =  \l
\vec{z}.(\l  w.  v_i)[\fresh{x}/u]  \Rew{\les} \l  \vec{z}.\l  w.
v_i[\fresh{x}/u] = t_{i+1}$   so  that $\Er(t_i) = \l \vec{z}.\l w.
v_i = \Er(t_{i+1})$.


\[ \begin{array}{lllllllllllllll}
  t_1 & \Rew{\les} & t_2 & \Rew{} & \ldots & \Rew{} & t_n\\
  \\
  \Er(t_1) & \Rewn{\les} & \Er(t_2) & \Rewn{\les} & \ldots & \Rewn{\les} & \Er(t_n)\\
  \end{array} \]


\end{proof}
}

\begin{lemma}
\label{l:tra-preserves-es}
Let $t, t'$ be $\mterms$-terms. 
If $t =_{\Es} t'$ then $\tra{t} =_{\Es} \tra{t'}$.
\end{lemma}
\begin{proof}
By induction on the definition of  $t =_{\Es} t'$. 
\ignore{
\begin{itemize}
\item $t = t_1[x/t_3][y/t_2] =_{\Es} t_1[y/t_2][x/t_3] = t'$ with $y \notin \fv(t_3), x \notin \fv(t_2)$. We treat the case where $x,y \in \fv(t_1)$. We have
$$
\begin{array}{ll}
&\tra{t_1[x/t_3][y/t_2]}\\
= & \ALC(\tra{t_1}[x/\tra{t_3}][\fresh{x}/\tra{t_3}][y/\tra{t_2}][\fresh{y}/\tra{t_2}]).
\end{array}
$$
As $\ALC$ has unique normal forms up to $_{\Es}$, we have 
$$
\begin{array}{ll}
& \ALC(\tra{t_1}[x/\tra{t_3}][\fresh{x}/\tra{t_3}][y/\tra{t_2}][\fresh{y}/\tra{t_2}])\\
 =_{\cref{l:same-variables-tra}} & \ALC(\tra{t_1}[y/\tra{t_2}][\fresh{y}/\tra{t_2}][x/\tra{t_3}][\fresh{x}/\tra{t_3}])\\
 = & \tra{t_1[y/t_2][x/t_3]}.
\end{array}
$$
\item $t = x =_{\Es} x = t'$.
\item $t = \l x.t_1 =_{\Es} \l x.t'_1  = t'$ with $t_1 =_{\Es} t'_1$. Then $\tra{\l x.t_1} = \l x. \tra{t_1} =_{\ih} \l x. \tra{t'_1} = \tra{\l x.t'_1}$.
\item $t = t_1\ t_2 =_{\Es} t'_1\ t'_2 = t'$ with $t_1 =_{\Es} t'_1$ and $t_2 =_{\Es} t'_2$. Then $\tra{t_1\ t_2} = \tra{t_1}\ \tra{t_2} =_{\ih} \tra{t'_1}\ \tra{t'_2} = \tra{t'_1\ t'_2}$.
\item $t = t_1[y/t_2] =_{\Es} t'_1[y/t'_2] = t'$ with $t_1 =_{\Es} t'_1$ and $t_2 =_{\Es} t'_2$. We will consider the case where $y \in \fv(t_1)$. We have $\tra{t_1[y/t_2]} = \ALC(\tra{t_1}[y/\tra{t_2}][\fresh{y}/\tra{t_2}])$. As $\ALC$ has unique normal forms up to $_{\Es}$, we have $\ALC(\tra{t_1}[y/\tra{t_2}][\fresh{y}/\tra{t_2}]) =_{\ih} \ALC(\tra{t'_1}[y/\tra{t'_2}][\fresh{y}/\tra{t'_2}]) = \tra{t'_1[y/t'_2]}$.
\end{itemize}
}
\end{proof}

\begin{restatable}[$\les$ simulates $\lm$]{proposition}{lessimulateslm}
Let $t, t'$ be $\mterms$-terms. 
If $t \Rew{\lm} t'$ then $\tra{t} \Rewplus{\les} \tra{t'}$.
\label{p:les-sim-lm}
\end{restatable}

\begin{proof} By induction on  $t \Rew{\lm} t'$.
Details can be found in ~\cite{OConchuir:2006a}. 
  \end{proof}

This property will be used in~\cref{s:sn-implies-intersection-typable}
to give an alternative
proof of $\lm$-strong normalisation
of simply typable terms, and in~\cref{s:graphical} to relate simply
typable $\mterms$-terms to MELL proof-nets.

\ignore{SACAR
  \begin{corollary}[PSN for $\lm$]
  Let $t$ be a $\mterms$-term.
  If $t \in \SN{\les}$,  then $t\in \SN{\lm}$.
  \end{corollary}
  \begin{proof}
    Let $t \in \SN{\les}$ and suppose $t\notin \SN{\lm}$.
    Then, there is an infinite $\lm$-reduction sequence starting at
    $t$, which projects, by~\cref{p:les-sim-lm},
    to an infinite $\les$-reduction sequence
    starting at $\tra{t}$. RELACIONAR
    $t \in \SN{\les}$ con $\tra{t} \in \SN{\les}$.
\end{proof} }

\section{Normalisation Properties}
\label{s:charact-les-lm}

Intersection   type   disciplines~\cite{CDC78,CDC80}   are   more
flexible  than simple  type systems  in the  sense that  not only are 
typed  terms strongly  normalising,  but  the converse  also
holds,  thus giving  a characterisation  of the  set  of strongly
normalising     terms.     Intersection    types   
for calculi with explicit substitutions not enjoying 
full composition have been  
studied in~\cite{LLDDvB,Kikuchi07}. Here, we apply this
technique to   $\lm$, 
and  obtain  a characterisation  of the  set of
strongly-normalising terms by means of intersection types.

Moreover, we study PSN. The
PSN  property  received a  lot  of  attention  in calculi  with
  explicit            substitutions            (see           for
  example~\cite{ACCL91,Benaissa96,Bloo95}),   starting   from  an
  unexpected  result given  by  Melli\`es~\cite{Mellies1995a} who
  has shown that there are $\beta$-strongly normalisable terms in
  $\l$-calculus  that are  not strongly  normalisable  in calculi
  such as $\lambda\sigma$~\cite{ACCL91}.  Since then, there was a
  challenge to  define calculi with  explicit substitutions being
  confluent on metaterms and enjoying PSN at the same time.  Many
  formalisms     such    as for example
  $\lambda_{ws}$~\cite{guillaume01}
  and  $\les$~\cite{K07}  have been  shown to
  enjoy   both   properties: confluence on metaterms and PSN.    In   particular,   $\lm$   enjoys
  PSN~\cite{ShaneHOR06}.   However,  the   first  proof  of  this
  result~\cite{ShaneHOR06}    is    quite   involved.     Indeed,
  $\lm$-reduction is simulated by  another  calculus enjoying
  PSN,   called  $\lambda_{{\tt  blxr}}$,   which  is   a  slight
  modification of $\llxr$~\cite{KL05},  a formalism with explicit
  ressources  (weakening,  contraction,  substitution)  based  on
  proof-nets  for  the  multiplicative  exponential  fragment  of
  Linear  Logic~\cite{girard}.   The  proof  in~\cite{ShaneHOR06}
  consists  of two  main  steps: first  prove  that the  modified
  calculus $\lambda_{{\tt blxr}}$ has the PSN property (this is a
  long proof  although it is  made easier by  adapting Lengrand's
  techniques~\cite{Lengrand2005}),    then    prove   that    any
  $\lm$-reduction   step  can   be  simulated   by   a  non-empty
  $\lambda_{{\tt  blxr}}$-reduction   sequence. In this section we
  also give an alternative proof of PSN
  for the $\lm$-calculus.

  \subsection{Types}
\defn{Types} are built over a countable set of atomic symbols 
(base types) and the type constructors $\>$ (for functional types)
and $\cap$ (for intersection types).
An \defn{environment} $\Gam$  is a
partial function from variables to types.  We denote by
$\dom(\Gam)$ the \defn{domain} of $\Gam$.
Two
environments $\Gam$ and $\Del$ are said to be \defn{compatible} iff
for all
$x \in \dom(\Gam) \cap
\dom(\Del)$
we have $\Gam(x)  = \Del(x)$.  We denote the
\defn{union  of  compatible contexts}  by  $\Gam  \uplus  \Del$.  Thus  for
example $(x:A, y:B) \uplus (x:A, z:C) =$ $(x:A,$~$y:B,$~$z:C)$.

\defn{Typing  judgements} have the form  $\Gam \vdash
t:A$ where $t$ is  a term, $A$ is a type and  $\Gam$ is an
environment. \defn{Derivations} of typing  judgements
in a certain type discipline system are obtained by application of 
the typing rules of the system. We consider several systems.

The  \defn{additive  simply  type}  system  for 
  $\l$-terms  (resp. for  $\mterms$-terms), written $\saddl$
  (resp. $\saddls$),  is  given  by the
  rules \mathsmall{\aaxiom}, \mathsmall{\aapp}, and
  \mathsmall{\aabs}        (resp.        \mathsmall{\aaxiom},
  \mathsmall{\aapp},          \mathsmall{\aabs},          and
  \mathsmall{\asubstr}) in~\cref{f:additive-systems}.

  \begin{figure}[h!]
    \[ \begin{array}{|c@{\hspace{.5cm}}c|}
\hline
\irule{}
      {\mathsmall{\Gam, x:A \vd x:A}}  \mathsmall{(\aaxiom)}  & 
\irule{\mathsmall{\Gam \vd  t:A \> B} \sep
       \mathsmall{\Gam \vd u:A}}
      {\mathsmall{\Gam \vd  t\;u :B}}   \mathsmall{(\aapp)}  \\[.4cm]
\irule{\mathsmall{\Gam, x:A \vd t:B}}
      {\mathsmall{\Gam \vd \l x. t: A\> B}}   \mathsmall{(\aabs)}  &
\irule{\mathsmall{\Gam \vd u:B \sep \Gam, x:B \vd t:A }} 
      {\mathsmall{\Gam  \vd t[x/u]:A}}  \mathsmall{(\asubstr)} \\
\hline
\end{array} \]
\caption{System $\saddl$ for $\l$-Terms and System  $\saddls$ for $\mterms$-Terms}
\label{f:additive-systems}
  \end{figure}

The \defn{multiplicative simple type} system  for  
  $\l$-terms  (resp. for $\mterms$-terms), written $\smultl$ (resp. $\smultls$), is given by the 
  rules \mathsmall{\maxiom}, \mathsmall{\mapp}, and
  \mathsmall{\mabs} (resp. \mathsmall{\maxiom}, \mathsmall{\mapp}, 
  \mathsmall{\mabs}, and \mathsmall{\msubstr})
  in~\cref{f:multiplicative-systems}.

\begin{figure}[h!]
  \[ \begin{array}{|c@{\hspace{.5cm}}c|}
\hline
\irule{}
      {\mathsmall{x:A \vd x:A}}   \mathsmall{(\maxiom)}   & 
\irule{\mathsmall{\Gam \vd  t:A \> B \sep  \Del \vd u:A}}
      {\mathsmall{\Gam \uplus \Del \vd  t\;u:B}}
      \mathsmall{(\mapp)} \\[.4cm]
\irule{\mathsmall{\Gam  \vd t:B}}
      {\mathsmall{\Gam \setminus \{ x:A\}  \vd \l x. t: A\> B}}  \mathsmall{(\mabs)} \sep &
\irule{\mathsmall{\Gam \vd u:B \sep \Del \vd t:A}}
      {\mathsmall{\Gam \uplus ( \Del \setminus \{ x:B \} ) \vd t[x/u]:A}} \mathsmall{(\msubstr)}  \\
\hline
\end{array} \]
\caption{System $\smultl$ for $\l$-Terms and System  $\smultls$ for $\mterms$-Terms}
\label{f:multiplicative-systems}
  \end{figure}

\begin{lemma}
\label{l:typable-add-typable-mul-simple-types}
Let $t$ be a $\mterms$-term. Then 
$\Gam \vd_{\saddls} t:A$
iff  $\Gam \cap \fv(t) \vd_{\smultls} t:A$. Moreover, 
if $t$ is a  $\l$-term, then $\Gam \vd_{\saddl} t:A$
iff  $\Gam \cap \fv(t) \vd_{\smultl} t:A$.
\end{lemma}

For the intersection type systems,
  we also consider the additional rules in~\cref{f:solo-intersection}.

  \begin{figure}[h!]
    \[ \begin{array}{|cccc|}
\hline
\irule{\mathsmall{\Gam \vd t:A \sep \Gam \vd t:B}} 
      {\mathsmall{\Gam \vd t: A \cap  B}} & \mathsmall{(\ini)} & 
\irule{\mathsmall{\Gam \vd t: A_1 \cap  A_2}}
      {\mathsmall{\Gam \vd t:A_i}} & \mathsmall{(\ine)} \\
\hline
    \end{array} \]
    \caption{Additional Rules for Intersection Types}
    \label{f:solo-intersection}
\end{figure}

  The \defn{additive intersection type} system for $\l$-terms
  (resp. for $\mterms$-terms), written $\addl$ (resp. $\addls$)
  and  given in~\cref{f:typing-rules-intersection-additive-les}, is
  obtained by adding the rules $\ini$ and $\ine$ in~\cref{f:solo-intersection} to those of $\saddl$
  (resp. $\saddls$) in~\cref{f:additive-systems}. 
  
  \begin{figure}[h!]
\[ \begin{array}{|cccc|}
\hline
\irule{}
      {\Gam, x:A \vd x:A} &  (\aaxiom)  & 
\irule{\Gam \vd u:B \sep \Gam, x:B \vd t:A } 
      {\Gam  \vd t[x/u]:A} & (\asubstr)   \\ &&&\\
\irule{\Gam \vd  t:A \> B \sep
       \Gam \vd u:A}
      {\Gam \vd  t\;u:B} &  (\aapp) & 
\irule{\Gam, x:A \vd t:B}
      {\Gam \vd \l x. t: A\> B} &  (\aabs)  \\ &&&\\
\irule{\Gam \vd t:A \sep \Gam \vd t:B} 
      {\Gam \vd t: A \cap  B} & (\ini) & 
\irule{\Gam \vd t: A_1 \cap  A_2}
      {\Gam \vd t:A_i} & (\ine) \\
\hline
\end{array} \]
\caption{System $\addl$ for $\l$-Terms and
  System $\addls$ for  $\mterms$-terms}
\label{f:typing-rules-intersection-additive-les} 
\end{figure}


  The \defn{multiplicative intersection type} system for $\l$-terms
  (resp. for $\mterms$-terms), written $\multl$ (resp. $\multls$)
  and given in~\cref{f:typing-rules-l}, is
  obtained by adding the rules $\ini$ and $\ine$ in~\cref{f:solo-intersection} to those of $\saddl$
  (resp. $\saddls$) in~\cref{f:multiplicative-systems}.
  For technical reasons we specify  rule $\mabs$ (resp. $\msubstr$)
  by using two different
  instances  $\mabsi$ and $\mabsii$ (resp. $\msubsi$ and $\msubsii$).

\begin{figure}[h!]
\[ \begin{array}{|cccc|}
\hline
\irule{}
      {x:A \vd x:A} &  (\maxiom)  & 
\irule{\Gam \vd  t:A \> B \sep
       \Del \vd u:A}
      {\Gam \uplus \Del \vd  t\;u:B} &  (\mapp) \\ & & & \\
\irule{\Gam, x:A \vd t:B}
      {\Gam \vd \l x. t: A\> B} &  (\mabsi) & 
\irule{\Gam \vd t:B\ \mbox{ and } x \notin \Gam}
      {\Gam \vd \l x. t: A\> B} & (\mabsii) \\  & & & \\
\irule{\Gam \vd u:B \sep 
       \Del, x:B \vd t:A } 
      {\Gam \uplus \Del \vd t[x/u]:A} &(\msubsi)  & 
\irule{\Gam \vd u:B \sep 
       \Del \vd t:A \mbox{ and }\sep x \notin \Del} 
      {\Gam \uplus \Del \vd t[x/u]:A}  & (\msubsii)  \\& & & \\
\irule{\Gam \vd t:A \sep \Gam \vd t:B}
      {\Gam \vd t: A \cap  B} & (\ini) & 
\irule{\Gam \vd t: A_1 \cap  A_2}
      {\Gam \vd t:A_i} & (\ine) \\
\hline
\end{array} \]
\caption{System $\multl$ for  $\l$-terms
  and System $\multls$  for   $\mterms$-terms}
\label{f:typing-rules-l} 
\end{figure}

A term $t$ is said to be \defn{typable}  in system $\TST$,
written  $\Gam \vd_\TST t:A$ iff
there is $\Gam$ and $A$ s.t. the judgement $\Gam \vd t:A$ is derivable
from the set of typing rules of system $\TST$. 

Remark that for any  $\l$-term  $t$ 
we have $\Gam \vd_{\addl} t:A$  iff 
$\Gam \vd_{\addls} t:A$
and $\Gam \vd_{\multl} t:A$  iff 
$\Gam \vd_{\multls} t:A$.

  \begin{definition}
The relation $\ll$ on types is defined by the following axioms
and rules
\begin{enumerate}
\item $A \ll A$
\item $A \cap B \ll A$
\item $A \cap  B \ll B$
\item $A \ll B\ \&\  B \ll C$ implies $A \ll C$
\item $A \ll B\ \&\  A \ll C$ implies $A \ll B \cap  C$
\end{enumerate}
\end{definition}

We use  $\un{n}$ for  $\set{1 \ldots n}$  and $\capp{n}  A_i$ for
$A_1 \cap \ldots  \cap A_n$. The following property  can be shown
by induction on the definition of $\ll$.

\begin{lemma}
\label{l:relation-between-ll-types}
Let $\capp{n} A_i \ll \capp{m} B_j$, where none of the
$A_i$ and $B_j$ is an intersection type.  Then
for each $B_j$ there is $A_i$ s.t. $B_j = A_i$. 
\end{lemma}

\begin{proof}
  By an induction on the definition of $\capp{n} A_i \ll \capp{m}
  B_j$. Let $\capp{p} C_k$ be some
  type where no $C_k$ is an intersection type.
\begin{description}
\item[Case $\capp{n} A_i \ll \capp{n} A_i$.] Trivial.
\item[Case $\capp{n} A_i \cap \capp{p} C_k \ll \capp{n} A_i$.] Trivial.
\item[Case $\capp{p} C_k \cap \capp{n} A_i \ll \capp{n} A_i$.] Trivial.
\item[Case  $\capp{n} A_i  \ll  \capp{p} C_k,  \capp{p} C_k  \ll
  \capp{m}  B_j$.] Applying  \ih\  once,  we have  for  each $B_j$
  there is $C_k$  s.t. $B_j = C_k$. Applying  \ih\  again, we have
  for each $C_k$ there is $A_i$ s.t. $C_k = A_i$.
\item[Case $\capp{n} A_i \ll  B_1 \cap \ldots \cap B_k, \capp{n}
  A_i \ll  B_{k+1} \cap \ldots  \cap B_m$.] 
  Applying the
  \ih\  to $\capp{n} A_i \ll  B_1 \cap \ldots \cap B_k$
  and $\capp{n}
  A_i \ll  B_{k+1} \cap \ldots  \cap B_m$
  we have for each $B_j, 1 \leq j \leq k$ there is $A_i$ s.t. $B_j =
  A_i$ and for each $B_j, k+1 \leq j \leq
  m$ there is $A_i$ s.t. $B_j = A_i$.
 \end{description}
\end{proof}

\subsection{Basic Properties of the Type Systems}
\label{s:basic-properties-type-systems}

We show some basic properties of the type systems.

\begin{lemma}
\label{l:ll-and-typing}
If $\Gam \vd_\TST  t:A$ and $A \ll B$, then 
$\Gam \vd_\TST t:B$ for all $\TST \in \{ \addl, \addls,
  \multl, \multls\}$.
\end{lemma}

\begin{proof} Let $\Gam \vd_\TST  t:A$. We 
reason by induction on the definition of $A \ll B$.
\begin{description}
\item[Case $A=B, A \ll A$.] Trivial.
\item[Case $A=B \cap C \ll B$.] Use $\ine$.
\item[Case $A=C \cap B \ll B$.] Use $\ine$.
\item[Case $A \ll C, C \ll B$.] Use \ih\ once to get $\Gam \vd_\TST  t:C$ and a second time to get $\Gam \vd_\TST  t:B$.
\item[Case $A \ll  B_1, A \ll B_2, B = B_1  \cap B_2$.] Use \ih\ 
  twice to  get $\Gam  \vd_\TST  t:B_1$ and  $\Gam \vd_\TST  t:B_2$  and then
  apply $\ini$.
\end{description}
\end{proof}

The proofs of the following lemmas can be found in~\cref{app:types}.

\begin{restatable}[Environments are Stable by $\ll$]
  {lemma}{environmentstableorderoftypes}
If $\Gam, x:B \vd_\TST  t:A$ and $C \ll B$, then $ \Gam, x:C
  \vd_\TST t:A$ for all $\TST \in \{ \addl, \addls,
  \multl, \multls\}$.
\label{l:ll-preserves-typing}
\end{restatable}

\begin{restatable}[Weakening]{lemma}{weakening}
If $\Del \cap \fv(t) = \emptyset$,
then 
$\Gam \vd_{\addls} t:A$ iff $\Gam, \Del \vd_{\addls} t:A$. 
\label{l:weakening}
\end{restatable}

\begin{restatable}[Additive  Generation Lemma]{lemma}{additivegenerationlemma} 
  Let $\TST$ be an additive system. Then 
\begin{enumerate}
\item $\Gam \vd_\TST x:A$ iff there is $x:B \in \Gam$ and $B \ll A$.
\item $\Gam \vd_\TST t\ u:A$ iff  there exist $A_i, B_i, i \in \un{n}$ s.t.
$\capp{n} A_i \ll A$ and $\Gam \vd_\TST t: B_i \> A_i$ and
$\Gam \vd_\TST u:B_i$.
\item $\Gam \vd_\TST t[x/u]:A$ iff there exist $A_i, B_i, i \in \un{n}$ s.t.
$\capp{n} A_i \ll A$ and  $\forall i \in
\un{n}$ $\Gam \vd_\TST u:B_i$ and $\Gam, x:B_i \vd_\TST t: A_i$.
\item\label{lambda-case-add}  $\Gam \vd_\TST\l x. t:A$ iff   there exist $A_i, B_i, i \in
  \un{n}$ s.t. $\capp{n} (A_i \> B_i) \ll A$
and   $\forall i \in
\un{n}$   $\Gam, x:A_i \vd_\TST t:B_i$. 
\item \label{last-add} $\Gam \vd_\TST\l x. t:B \> C$ iff   $\Gam, x:B \vd_\TST t:C$. 
\end{enumerate}
\label{l:additive-generation-lemma}
\end{restatable}

\begin{restatable}[Multiplicative Generation Lemma]{lemma}{multiplicativegenerationlemma} 
  Let $\TST$ be a multiplicative system. Then 
\begin{enumerate}
\item \label{lax} $\Gam \vd_\TST x:A$ iff $\Gam = x:B$ and $B \ll A$.
\item \label{lapp} $\Gam \vd_\TST t\ u:A$ iff $\Gam = \Gam_1 \uplus \Gam_2$,
where $\Gam_1 = \fv(t)$ and $\Gam_2 = \fv(u)$ 
and there exist $A_i, B_i, i \in \un{n}$ s.t.
$\capp{n} A_i \ll A$ and  $\forall i \in \un{n}$, $\Gam_1 \vd_\TST t: B_i \> A_i$ and
$\Gam_2 \vd_\TST u:B_i$.
\item \label{lsub} $\Gam \vd_\TST t[x/u]:A$ iff $\Gam = \Gam_1 \uplus \Gam_2$,
where $\Gam_1 = \fv(t) \setminus \set{x}$ and $\Gam_2 = \fv(u)$
and there exist $A_i, B_i, i \in \un{n}$ s.t.
$\capp{n} A_i \ll A$ and  $\forall i \in
\un{n}$, $\Gam_2 \vd_\TST u:B_i$ and 
either $x \notin \fv(t)\ \&\ \Gam_1 \vd_\TST t: A_i$ or 
$x \in \fv(t)\ \&\ \Gam_1, x:B_i \vd_\TST t: A_i$.
\item \label{lambda-case}
$\Gam \vd_\TST \l x. t:A$ iff  $\Gam = \fv(\l x.t)$ and there exist $A_i, B_i, i \in
  \un{n}$ s.t. $\capp{n} (A_i \> B_i) \ll A$
and l $\forall i \in
\un{n}$, either $x \notin \fv(t)\ \&\ \Gam \vd_\TST t:B_i$ or 
$x \in \fv(t)\ \&\ \Gam, x:A_i \vd_\TST t:B_i$.
\item \label{last} $\Gam \vd_\TST \l x. t:B \> C$ iff $\Gam = \fv(\l x.t)$ and  $\Gam, x:B \vd_\TST
  t:C$ or $\Gam \vd_\TST t:C$.
\end{enumerate}
  \label{l:multiplicative-generation-lemma}
\end{restatable}

We can now state  a 
correspondence between the multiplicative 
and additive systems with intersection types.

\begin{lemma}
\label{l:typable-add-typable-mul}
Let $t$ be a $\mterms$-term. Then 
$\Gam \vd_{\addls} t:A$
iff  $\Gam \cap \fv(t) \vd_{\multls} t:A$. Moreover, 
if $t$ is a  $\l$-term, then $\Gam \vd_{\addl} t:A$
iff  $\Gam \cap \fv(t) \vd_{\multl} t:A$.
\end{lemma}

\begin{proof} 
  The right to left implication is by induction on $t$ using both
  generation  lemmas and~\cref{l:weakening}.   The left  to
  right implication  is by induction on $t$  using the generation
  lemmas.
\end{proof}

Since  systems $\saddl$ and $\smultl$ (resp. $\saddls$  and $\smultls$)
type the same sets of $\l$-terms (resp. $\mterms$-terms) (\cref{l:typable-add-typable-mul-simple-types}), and
systems $\addl$ and $\multl$ (resp. $\addls$  and $\multls$) type
the same sets of  $\l$-terms (resp. $\mterms$-terms) (\cref{l:typable-add-typable-mul}), then, from now on,
 \defn{simply typable  $\l$-term} means typable in
$\saddl$ or  $\smultl$, \defn{intersection  typable  $\l$-term} means typable in
$\addl$ or  $\multl$,   
 \defn{simply typable  $\mterms$-term} means typable in
 $\saddls$ or $\smultls$ and 
 \defn{intersection  typable $\mterms$-term} means typable in
 $\addls$ or  $\multls$.

\subsection{Simply typable  $\mterms$-terms are $\lm$-strongly normalising}
\label{s:simply-typed-implies-sn}

The goal of this section is to show that simply typable $\mterms$-terms
are $\lm$-strongly normalising. This result turns out to be a
consequence of strong normalisation of simply typable $\l$-terms
in the partial  $\lpar$-calculus; a result 
which can be shown using a simple  arithmetical proof~\cite{vanDaalen,David}.
This proof is constructive as it only uses induction and
intuitionistic reasoning. 

\begin{lemma}
\label{l:substitution-of-sn-lpar}
Let  $t,u$ be  a simply typable  $\l$-terms. 
If $t,u\in \SN{\lpar}$, then $t\isubs{x}{u}
\in \SN{\lpar}$.
\end{lemma}

\begin{proof}
By  induction on  $\langle \type(u),   \eta_{\lpar}(t),  |t|
\rangle$, where $|t|$ is the number of constructors in $t$. \paper{We treat the interesting cases.} 

\begin{itemize}
\report{
\item  If  $t =   \lambda  y.  v$,  then  $v\isubs{x}{u}  \in
         \SN{\lpar}$  by  the  \ih\   and   thus $t\isubs{x}{u}    = \l
         x.      v\isubs{x}{u}    \in   \SN{\lpar}$.

\item $t  = y  \ov{v_n}$
         with $x \neq  y$. The \ih\ gives
$v_i\isubs{x}{u} \in \SN{\lpar}$ 
since $\eta_{\lpar}(v_i)$ decreases  and
         $|v_i|$ strictly decreases. Then we conclude
         straightforward. 

\item $t =x $. Then $x \isubs{x}{u} =u\in \SN{\lpar}$ by the hypothesis. 
}

\item $t  =  x  v   \ov{v_n}$. The \ih\ 
         gives  $V = v \isubs{x}{u}$ and  $V_i =
         v_i\isubs{x}{u}$ in $\SN{\lpar}$. 
         To show $t\isubs{x}{u}=u V \ov{V_n} \in \SN{\lpar}$
         it is sufficient to show that all its reducts are
         in $\SN{\lpar}$. We reason by induction on
         $\eta_{\lpar}(u) + \eta_{\lpar}(V) + \Sigma_{i \in 1
           \ldots n}\ \eta_{\lpar}(V_i)$. 

         \begin{itemize} 

        \item  If  the reduction takes place  in $u$,
         $V$ or $V_i$,  then the   property holds  by  the \ih\  \item
         Suppose $u = \l y . U$ and $(\l y  . U)\ V\ \ov{V_n} \Rew{\BGc}
         U\   \ov{V_n}$.   We  write     $U\    \ov{V_n}$  as      $(z\
         \ov{V_n})\isubs{z}{U}$, where $z$ is a fresh variable.  Since
         every  $V_i   \in  \SN{\lpar}$,   then  $z\   \ov{V_n}    \in
         \SN{\lpar}$.  Also, $u    \in  \SN{\lpar}$   implies $U   \in
         \SN{\lpar}$.  Thus,   $\type(U)   < \type(u)$   implies   $(z\
         \ov{V_n}) \isubs{z}{U} \in \SN{\lpar}$ by the \ih
        
         \item 
         Suppose  $u = \l y . C\cwc{y}$ and $(\l y . C\cwc{y})\ V\ \ov{V_n} 
        \Rew{\bp} (\l y. C\cwc{V})\ V\  \ov{V_n}$. 
         We write $\l y. C\cwc{V}$ as $(\l y. C\cwc{z})\isubs{z}{V}$, where
         $z$ is a fresh variable. Since $u \in \SN{\lpar}$, then
         $C\cwc{y} \in \SN{\lpar}$. The change of free occurrences of 
         variables preserve normalisation so that
         $C\cwc{z} \in \SN{\lpar}$ and thus $\l y. C\cwc{z} \in \SN{\lpar}$.
         We also  have $\type(V) = \type(v)  
          < \type(u)$     so     that we get 
         $(\l y. C\cwc{z})\isubs{z}{V} \in  \SN{\lpar}$ by the  \ih\ 
         \end{itemize}

\item $t =
         (\l y. s) v  \ov{v_n}$.  The  \ih\ gives $S = s\isubs{x}{u}$
         and    $V = v    \isubs{x}{u}$   and $V_i =  v_i
         \isubs{x}{u}$ are in $\SN{\lpar}$. These terms  are  also  typable.
         To show 
         $t\isubs{x}{u}        =   (\l y.       S)  V  \ov{V_n} \in \SN{\lpar}$
         it is sufficient to show that all its reducts are in $\SN{\lpar}$.  
         We reason by induction on
         $\eta_{\lpar}(S) + \eta_{\lpar}(V) + \Sigma_{i \in 1
           \ldots n}\ \eta_{\lpar}(V_i)$.

         \begin{itemize}
         \item  If  the reduction takes place  in $S$,
         $V$ or $V_i$,  then the   property holds  by  the \ih\  

         \item
         Suppose $(\l y  . S)\ V\ \ov{V_n} \Rew{\BGc}
         S\   \ov{V_n}$.   We  write     $S\    \ov{V_n}$  as      $(s\
         \ov{v_n})\isubs{x}{u}$.  Since
         $(\l y. s)\ v\ \ov{v_i} \Rew{\lpar} s\ \ov{v_n}$, then
         $\eta_{\lpar}(s\ \ov{v_n}) < \eta_{\lpar}((\l y. s)\ v\ \ov{v_n})$ 
         and thus we conclude 
          $S\   \ov{V_n}    \in
         \SN{\lpar}$ by the \ih
        
         \item 
         Suppose  $u = \l y . C\cwc{y}$ and $(\l y . C\cwc{y})\ V\ \ov{V_n} 
        \Rew{\bp} (\l y. C\cwc{V})\ V\  \ov{V_n}$. 
         We write $\l y. C\cwc{V}$ as $(\l y. C\cwc{v})\isubs{x}{u}$. 
         Since $(\l y . C\cwc{y})\ v\ \ov{v_n} 
        \Rew{\bp} (\l y. C\cwc{v})\ v\  \ov{v_n}$, then 
         $\eta_{\lpar}((\l y. C\cwc{v})\ v\  \ov{v_n}) < 
          \eta_{\lpar}((\l y. C\cwc{v})\ v\  \ov{v_n})$ and thus
         we conclude  $(\l y. C\cwc{V})\ V\  \ov{V_n} \in \SN{\lpar}$
         by the \ih

          \end{itemize}

\end{itemize} 
\end{proof}

\begin{theorem}[SN for $\lpar$]
  \label{c:lpar-sn}
  Let $t$ be a $\l$-term.
If $t$ is  simply typable,  then $t\in \SN{\lpar}$.
\end{theorem}

\begin{proof}
By induction  on the structure  of $t$. The cases $t=x$  and $t=\l x. u$
are  straightforward.   If $t  = u v$,   then write   $t=(z\
v)\isubs{z}{u}$. By  the \ih\ $u,v \in  \SN{\lpar}$  and  thus~\cref{l:substitution-of-sn-lpar}   gives  $t  \in  \SN{\lpar}$.
\end{proof}

\begin{corollary}[SN for $\lm$ (i)]
  \label{c:sn-lm}
  Let $t$ be  a $\mterms$-term.
If $t$ is simply typable,   then $t\in \SN{\lm}$.
\end{corollary}

\begin{proof}
Take  $t$ typable  in  $\saddls$. Then,   $\lmpar{t}$ (defined in~\cref{s:lpar}) is a 
$\l$-term. One shows by induction on $t$ that
$\lmpar{t}$ is typable in $\saddl$,   and that $\lmpar{t} \Rewplus{\B} t$.
Since $\lmpar{t}$ is a simply typable  $\l$-term, then
by~\cref{c:lpar-sn}
$\lmpar{t} \in \SN{\lpar}$ and by~\cref{c:lpar-lm}
$\lmpar{t} \in \SN{\lm}$ . Thus $t$ is also in $\SN{\lm}$. 
\end{proof}

 This same result admits
  an alternativa proof. 
  \begin{corollary}[SN for $\lm$ (ii)]
  Let $t$ be  a $\mterms$-term.
  If $t$ is simply typable,   then $t\in \SN{\lm}$.
\end{corollary}
\begin{proof}
Let $\Gam \vd_{\smultls}  t:A$. 
It is not difficult
to  show that  $\tra{t}$ is also typable in
$\smultls$, by induction on $t$. Then $\tra{t} \in \SN{\les}$ by~\cite{K07}. 
Now, suppose  $t \notin \SN{\lm}$. Then given an infinite
$\lm$-reduction sequence starting at $t$
we can construct, by~\cref{p:les-sim-lm},
an infinite $\les$-reduction sequence
starting at $\tra{t}$. This leads to a contradiction with $\tra{t} \in \SN{\les}$.
Thus  $t \in \SN{\lm}$.
\end{proof}

\subsection{Intersection Typable  $\mterms$-terms are $\lm$-strongly normalising}
\label{s:intersection-typed-are-sn}

The goal of this section is to show that intersection
typable $\mterms$-terms
are $\lm$-strongly normalising.
We make use of the functions $\lmpar{\_}$
and $\tra{\_}$, respectively 
defined in~\cref{s:lpar} and~\cref{s:PSN}.

\begin{lemma}
    \label{t:t-typable-V-typable}
  Let $t$ be  a $\mterms$-term.
  Then $t$ is typable in $\multls$ if and only if
  $\lmpar{t}$ is typable in $\multl$.
\end{lemma}
\begin{proof}
  By induction on the typing derivation of
  $t$.
\end{proof}

\begin{theorem}
  \label{t:typable-implies-sn-les}
  Let $t$ be  a $\mterms$-term.
If $t$ is intersection  typable,  then $t\in \SN{\les}$.
\end{theorem}
\begin{proof}
  Let $t$ be intersection  typable,
  so that in particular
  $t$ is typable in $\multls$.
 \Cref{t:t-typable-V-typable} gives
  $\lmpar{t} $ typable in $\multl$
  and~\cref{l:typable-add-typable-mul} gives
  $\lmpar{t} $ typable in $\addl$.
  Thus $\lmpar{t} $ is $\beta$-strongly normalising~\cite{Pottinger80}.
  As a consequence, $\lmpar{t}  \in \SN{\les}$
  by PSN~\cite{K07}.
  Since $\lmpar{t} \Rewn{\les} t$,
  then we  conclude also $t \in \SN{\les}$.
 \end{proof}

\begin{lemma}
    \label{t:t-typable-T-typable}
  Let $t$ be  a $\mterms$-term.
  Then $t$ is typable in $\multls$ if and only if
  $\tra{t}$ is typable in $\multls$.
\end{lemma}
\begin{proof}
  By induction on the typing derivation of
  $t$.
\end{proof}

\begin{theorem}
  \label{t:typable-implies-sn-milner}
Let $t$ be  a $\mterms$-term.
If $t$ is intersection  typable,  then $t\in \SN{\lm}$.
\end{theorem}
\begin{proof}
Let $t$ be intersection  typable,
  so that in particular 
  $t$ is typable in $\multls$.
  By~\cref{t:t-typable-T-typable} also 
  $\tra{t}$ is  typable in $\multls$.
  Thus, $\tra{t} \in \SN{\les}$ by~\cref{t:typable-implies-sn-les}.
  Suppose $t \notin \SN{\lm}$, so that there is
   an infinite
   $\lm$-reduction sequence starting at $t$,
   which projects, 
  by~\cref{p:les-sim-lm}, to an infinite $\les$-reduction
  sequence starting at $\tra{t}$. This leads to a contradiction
  with $\tra{t} \in \SN{\les}$.
Thus we conclude $t \in \SN{\lm}$ as required.
 \end{proof}

\subsection{$\lm$-strongly normalising terms are intersection typable  $\mterms$-terms}
\label{s:sn-implies-intersection-typable}

We now complete the picture  by showing that the intersection
type discipline for $\mterms$-terms gives a characterisation of 
$\lm$-strongly normalising terms. To do this, we use the
translation $\lmpar{\_}$ introduced in~\cref{s:lpar}
to relate $\mterms$-terms to $\l$-terms. 



\begin{lemma}
  \label{l:compositionaly-back}
  Let $t,u$ be  $\mterms$-terms. 
Then $\lmpar{t}\isubs{x}{\lmpar{u}} = \lmpar{t\isubs{x}{u}}$.
\end{lemma} 
\begin{proof} By induction on $t$.

\begin{itemize}
\item If $t = y$ and $y = x$ then $x\isubs{x}{\lmpar{u}} = \lmpar{u} = \lmpar{x\isubs{x}{u}}$. If $y \neq x$ then $y\isubs{x}{\lmpar{u}} = y = \lmpar{y\isubs{x}{u}}$. 
\item If $t = \l y.t'$ then $\lmpar{t}\isubs{x}{\lmpar{u}} = (\l y.\lmpar{t'}\isubs{x}{\lmpar{u}}) = (\l y.\lmpar{t'\isubs{x}{u}})$ by the \ih. Then, $(\l y.\lmpar{t'\isubs{x}{u}}) = \lmpar{\l y.t'\isubs{x}{u}} = \lmpar{(\l y.t')\isubs{x}{u}}$.
\item If $t = u\ v$ then the proof is similar.
\item  If $t  = t_1[y/t_2]$  then  $\lmpar{t}\isubs{x}{\lmpar{u}} =
  ((\l  y.\lmpar{t_1})\lmpar{t_2})\isubs{x}{\lmpar{u}}  =$\linebreak
  $(\l                           y.\lmpar{t_1}\isubs{x}{\lmpar{u}})\
  (\lmpar{t_2}\isubs{x}{\lmpar{u}})$. Applying \ih\  twice we get
\begin{eqnarray*}
&&(\l y.\lmpar{t_1}\isubs{x}{\lmpar{u}})\ (\lmpar{t_2}\isubs{x}{\lmpar{u}})\\
&=&(\l y.\lmpar{t_1\isubs{x}{u}})\ \lmpar{t_2\isubs{x}{u}}\\
&=&\lmpar{t_1\isubs{x}{u}[y/t_2\isubs{x}{u}]}\\
&=&\lmpar{(t_1[y/t_2])\isubs{x}{u}}
\end{eqnarray*}
\end{itemize} 
\end{proof}

\begin{lemma}
\label{l:beta-les-via-back}
Let $t$ be a $\mterms$-term. 
If $\lmpar{t} \Rew{\beta} t'$, then  $\exists\ u$ s.t.  
$t  \Rewplus{\lm} u$ and $t' = \lmpar{u}$.
\end{lemma}

\begin{proof}
By induction on the reduction step $\lmpar{t}
\Rew{\beta} t'$.

If the step is external, then we have two possibilites.

\begin{itemize}
\item If $\lmpar{(\l x. t_1)\ t_2} = (\l x. \lmpar{t_1})\ \lmpar{t_2}
  \Rew{\beta} \lmpar{t_1}\isubs{x}{\lmpar{t_2}}$, then 
 $(\l x. t_1)\ t_2 \Rew{\B} t_1[x/t_2] \Rewplus{\lm} 
t_1\isubs{x}{t_2}$ by~\cref{l:full-composition}.
We conclude by~\cref{l:compositionaly-back}. 

\item If $\lmpar{t_1[x/t_2]} =  (\l x. \lmpar{t_1})\ \lmpar{t_2}
  \Rew{\beta}\lmpar{t_1}\isubs{x}{\lmpar{t_2}}$, then
$t_1[x/t_2] \Rewplus{\lm} 
t_1\isubs{x}{t_2}$ by~\cref{l:full-composition}.
We conclude again by~\cref{l:compositionaly-back}. 
\end{itemize}

If the step is internal, then we reason by cases. 

\begin{itemize}
\item If $\lmpar{t_1\ t_2} = \lmpar{t_1}\ \lmpar{t_2} \Rew{\beta}
  t'_1\ \lmpar{t_2}$, then $t_1 \Rewplus{\lm} u_1$ and $t'_1=
  \lmpar{u_1}$ by the \ih\  so that
  $ t_1\ t_2 \Rewplus{\lm} u_1\ t_2$
  and $t'_1\ \lmpar{t_2} = \lmpar{u_1\ t_2}$.

\item If $\lmpar{t_1\ t_2} = \lmpar{t_1}\ \lmpar{t_2} \Rew{\beta}
  \lmpar{t'_1}\ t'_2$, then this case is similar to the previous one.

\item If $\lmpar{\l x. t_1} = \l x. \lmpar{t_1}\Rew{\beta}\l
  x. t'_1$, then $t_1 \Rewplus{\lm} u_1$ and $t'_1=
  \lmpar{u_1}$ by the \ih\  so that
  $\l x. t_1 \Rewplus{\lm} \l x. u_1$
  and $\l x. t'_1 = \lmpar{\l x. u_1}$.

\item If $\lmpar{t_1[x/t_2]} =  (\l x. \lmpar{t_1})\ \lmpar{t_2}
  \Rew{\beta} (\l x. t'_1)\ \lmpar{t_2}$, then 
   $t_1 \Rewplus{\lm} u_1$ and $t'_1=
  \lmpar{u_1}$ by the \ih\  so that
  $t_1[x/t_2]  \Rewplus{\lm} u_1[x/t_2]$
  and $(\l x. t'_1)\ \lmpar{t_2} = \lmpar{u_1[x/t_2]}$.

\item If $\lmpar{t_1[x/t_2]} =  (\l x. \lmpar{t_1})\ \lmpar{t_2}
  \Rew{\beta} (\l x. \lmpar{t_1})\ t'_2$, then this case is
  similar to the previous one. 
\end{itemize}
\end{proof}

\begin{theorem} 
\label{t:sn-implies-typable}
Let $t$ be a $\mterms$-term. 
If $t \in \SN{\lm}$, then $t$ is an intersection typable  $\mterms$-term. 
\end{theorem}

\begin{proof}
Let $t \in  \SN{\lm}$.
Suppose $\lmpar{t} \notin \SN{\beta}$.  Then, there is an infinite
$\beta$-reduction  sequence starting  at $\lmpar{t}$,  
which can be projected, by~\cref{l:beta-les-via-back}, to an infinite
$\lm$-reduction sequence starting at $t$.  Thus $t
\notin  \SN{\lm}$, which leads  to a
contradiction.

Therefore $\lmpar{t}\in \SN{\beta}$, so that $\lmpar{t}$ is 
typable in $\addl$ by~\cite{Pottinger80}.  By~\cref{l:typable-add-typable-mul}
$\lmpar{t}$ is also typable in $\multls$,
and by~\cref{t:t-typable-V-typable} $t$ is typable in
$\multls$.
\end{proof}

\subsection{PSN}

We now  show  the PSN property
stating that  $\lm$-reduction preserves  $\beta$-strong
normalisation. A proof of this result
already exists~\cite{ShaneHOR06}. We reprove this property in a more simple way. 

\begin{corollary}[PSN for $\lm$]
  \label{c:psn}
  Let $t$ be a $\l$-term.
  If $t \in \SN{\beta}$,  then $t\in \SN{\lm}$.
\end{corollary}
\begin{proof}
If $t \in \SN{\beta}$, then $t$ is typable 
in $\addl$ by~\cite{Pottinger80}, so that
$t$ is also typable 
in $\addls$ (which contains $\addl$).
We conclude
$t \in \SN{\lm}$ by~\cref{t:typable-implies-sn-milner}. 
\end{proof}

We finally conclude with the following equivalences:

\begin{corollary}
Let $t$ be a $\mterms$-term. Then  $t$ is typable in $\addls$ iff
$t$ is typable in $\multls$ iff
$t \in \SN{\lm}$ iff $t \in \SN{\ldef}$. Furthermore,
let $t$ be a  $\l$-term. Then  $t$ is typable in $\addl$ iff
$t$ is typable in $\multl$ iff
$t \in \SN{\lm}$ iff $t \in \SN{\lpar}$ iff $t \in \SN{\beta}$.
\end{corollary} 
\begin{proof}
The statement $t$ is typable in $\addls$ iff $t$ is typable in
$\multls$ holds by~\cref{l:typable-add-typable-mul}.  The statement
$t$ intersection typable iff $t\in \SN{\lm}$ holds
by~\cref{t:typable-implies-sn-milner} and~\cref{t:sn-implies-typable}.
The statement $t \in \SN{\lm}$ iff $t \in \SN{\ldef}$ holds
by~\cref{c:sn-lm-sn-ldef}.  The statement $t$ is typable in $\addl$
iff $t$ is typable in $\multl$ holds
by~\cref{l:typable-add-typable-mul}.  The statement $t$ is typable in
$\addl$ iff $t \in \SN{\beta}$ holds by~\cite{Pottinger80}.  The
statement $t \in \SN{\lm}$ iff $t \in \SN{\lpar}$ holds
by~\cref{c:lpar-lm}.  \Cref{c:psn} gives $t \in \SN{\beta}$ implies $t
\in \SN{\lm}$. And $t \in \SN{\lm}$ implies $t \in \SN{\beta}$ is a
consequence of~\cref{l:one-step-beta-simulation}.

\end{proof}



\section{Relating Partial Substitutions to Graphical Formalisms}
\label{s:graphical}

\subsection{MELL Proof-nets}
\label{l:proof-nets}

Calculi  with explicit substitutions
enjoy a nice relation with the multiplicative exponential fragment
of linear logic (MELL). This is done by interpreting terms into
\emph{proof-nets},  
a graphical formalism which represent 
MELL proofs in natural deduction style.
In order  to obtain this interpretation,  one first  defines a
(simply)  typed version  of the  term calculus.    The translation
from  $\mterms$-terms to  proof-nets gives  a simulation  of  the reduction
rules   for  explicit  substitutions   via  cut   elimination  in
proof-nets.    As  an   immediate  consequence   of  this
simulation, one  proves that a  simply typed version of  the term
calculus is strongly normalizing.  Also, an important property of
the  simulation is  that each  step in  the calculus  with  ES is
simulated  by a  \emph{constant} number  of steps  in proof-nets:
this  shows that  the two  systems  are very  close, unlike  what
happens  when simulating  the  $\l$-calculus. This  gives also  a
powerful    tool   to    reason   about    the    complexity   of
$\beta$-reduction.

We apply this idea to the $\lm$-calculus  by
using previous work based on an interpretation of $\les$-terms into
MELL proof-nets~\cite{K07} and our translation in~\cref{s:PSN}. We thus
obtain:

Let $t$ be a $\mterms$-term which is simply typable.
Then $t$ is in particular typable in  the 
multiplicative simple typed system $\smultls$
given in~\cref{f:multiplicative-systems}. Then the  translation
of $t$ into a MELL proof-net can be  given by
$\pntra{t} = \pnlestra{\tra{t}}$, where $\tra{\_}$ is the
translation from $\lm$ to $\les$ introduced in~\cref{s:PSN},
while  $\pnlestra{\_}$ is the translation from $\les$ to MELL
proof-nets given in~\cite{K07}.
Call $R/E$ the strongly normalising reduction relation on MELL proof-nets. Then: 

\begin{proposition}
\label{p:lm-proof-nets}
Let $t$ be a $\mterms$-term.
If $t$ is typable in $\smultls$
and $t \Rew{\lm} t'$, then $\pntra{t} \Rewplus{R/E} {\cal C}[\pntra{t'}]$, where
${\cal C}[\pntra{t'}]$ denotes a proof-net 
containing $\pntra{t'}$ as a sub proof-net. 
\end{proposition}

\begin{proof}
  Let $t \Rew{\lm} t'$.
  \Cref{p:les-sim-lm} gives $\tra{t} \Rewplus{\les}
\tra{t'}$.
Moreover, by a simple inspection of the proof of this proposition we 
know that there is  at least one $\{\B, \Var, \Gc\}$-step in the reduction
sequence $\tra{t} \Rewplus{\les}
\tra{t'}$. This together with Theorem 8.2 in~\cite{K07} gives 
$\pntra{t} = \pnlestra{\tra{t}} \Rewplus{R/E} {\cal
  C}[\pnlestra{\tra{t'}}]= {\cal C}[\pntra{t'}]$. 
\end{proof}

\begin{corollary}[SN for $\lm$ (iii)]
If $t$ is typable in $\smultls$, then $t\in \SN{\lm}$.
\end{corollary}

\begin{proof}
As $R/E$ is strongly normalising, we conclude $t \in \SN{\lm}$ using~\cref{p:lm-proof-nets}.
\end{proof}

\subsection{Local bigraphs}
\label{l:local-bigraphs}

Milner,    Leifer,     and    Jensen's    bigraphical    reactive
systems~\cite{Milner2001,LeiferMilner2000,JensenMilner2004} have been proposed  as a framework for modelling  the mobility of
distributed  agents able  to  manipulate their  own linkages  and
nested  locations. Milner  has  presented  an
encoding  of  $\lm$ as  a  bigraphical
reactive system $\lamBig$ as a means to study confluence in bigraphs~\cite{Milner2006}.  
This  encoding  may  also  be  understood  as  a  formalism  with
\emph{partial}  substitutions. 

The $\lm$-calculus   is  close   to  $\lamBig$  both
statically  and dynamically;  $\alpha$-equivalent terms  have the
same  encoding  and   one-step reduction  in  the  former matches
one-step reaction in  the latter. Thus,  any properties proved for
$\lm$ hold for the image of the encoding in $\lamBig$.

There is a close operational correspondence between $\lm$ and $\lamBig$:

\begin{proposition}[\cite{Milner2006}]
\label{p:lm-lambig-operational-correspondence}
Let $t$ be a $\mterms$-term.  Then $t \Rew{\lm} t'$ 
iff the
encoding of $t$ in $\lamBig$ can react in one step to  the encoding of $t'$ in $\lamBig$.
\end{proposition}

Thus, the image $\lamBig^e$ of the encoding is closed under reaction. 
We  can  reason  about   reaction  in  $\lamBig^e$  by  considering
reduction of $\lm$ terms without   metavariables:

\begin{corollary}[Confluence, PSN, SN]
$\lamBig^e$ is confluent and satisfies PSN. Encodings of intersection typed terms are strongly normalising.
\end{corollary}

\section{Conclusions}
\label{s:conclusions}

We  answer some  fundamental remaining questions  concerning the  adequacy of
Milner's $\l$-calculus with partial substitutions. In particular,
we  prove  that the  $\lm$-calculus  is  confluent  on terms  and
metaterms, that it  enjoys PSN, and that it allows a characterisation of $\lm$-strongly normalising
terms by using intersection type disciplines.

We relate $\lm$ to the  calculi with definitions $\lpar$ and $\ldef$,
thus obtaining a certain number of interesting results concerning
normalisation.   We also  relate the  $\lm$-calculus  to  
classical calculi  with explicit  substitutions.  Thus,  the $\lm$-calculus
can be  understood as a concise and  simple language implementing
partial and ordinary substitution,  both in implicit and explicit
style at the same time.

Last but not least, we establish a clear
connection   between  simply   typed   $\lm$-calculus  and   MELL
proof-nets,  thus  injecting  again  a  graph  representation  to
$\mterms$-terms  which   were  inspired  from   bigraphical  reactive
systems.

In  related    work,     Bundgaard    and
Hildebrandt~\cite{BundgaardHildebrandt2006}  use  partial   substitution similar to   $\lm$   in  their extension of Higher-Order  Mobile  Embedded  Resources (Homer), a higher-order  process calculus. Partial substitution
is also used  in different frameworks such as  for example Ariola
and  Felleisen's~\cite{AriolaFelleisen1997}  call-by-need  lambda
calculus and   Ariola    and   Klop's~\cite{AriolaKlop1997}   cyclic
$\lambda$-calculus.

  Grohmann  and  Miculan   have  modelled  the  call-by-name  and
  call-by-value       $\lambda$-calculi       with      bigraphs
  \cite{GrohmannMiculan2007}  by adapting  Milner's  model. While
  they concentrate on encodings  of $\l$-terms, the model is
  still   based  on  $\lm$ and   reduction  matches   reaction
  (\Cref{p:lm-lambig-operational-correspondence}). 
  Therefore,   our  results   can   be  used   to  reason   about
  normalisation properties  of encodings of  $\mterms$-terms in their
  models.



\paragraph{Acknowledgements}
We are grateful to  V. van  Oostrom who
pointed out to us references to calculi with 
partial notions of substitutions such
as $\lpar$ and $\ldef$. 

{\scriptsize
\paper{\bibliographystyle{abbrv}}}
\report{\bibliographystyle{alpha}}
\bibliography{Biblio}

\newpage
\appendix

\section{Proofs of~\cref{s:metaterms}}
\label{app:confluence}

We define a measure $\size(t)$ for $\mterms$-metaterm $t$ as follows: 

\[ \begin{array}{lll}
   \size(x)      &:= & 1\\
   \size(\mX_\Del) &:= & |\Del|\\
   \size(t\ u)   &:= & \size(t) + \size(u) \\
   \size(\l x. t)&:= & \size(t) \\
   \size(t[x/u]) &:= & \size(t) + \size(u) + \mul{x}{t} \cdot \size(u) \\
   &&\mbox{ unless } t = \mX_\Del [x_1/u_1]\ldots[x_n/u_n], x \in \Del \\
   \size(t[x/u]) &:= & \size(t) - 1 + \mul{x}{t} \cdot \size(u) \\
   &&\mbox{if } t = \mX_\Del [x_1/u_1]\ldots[x_n/u_n], x \in \Del
   \end{array} \] 

where  \[ \begin{array}{lll}
   \mul{x}{t}       &:= & 0 \hfill \mbox{ if } x \notin \fv(t)\\
   \mul{x}{x}       &:= & 1\\
   \mul{x}{\mX_\Del} &:= & 1 \hfill \mbox{ if } x \in \Del \\
   \mul{x}{t\ u}    &:= & \mul{x}{t} + \mul{x}{u}\\
   \mul{x}{\l y. t} &:= & \mul{x}{t}  \\
   \mul{x}{t[y/u]}  &:= & \mul{x}{t} + \mul{x}{u} + \mul{y}{t} \cdot \mul{x}{u}\\
    &&\mbox{ unless } t = \mX_\Del [x_1/u_1]\ldots[x_n/u_n], y \in \Del \\
   \mul{x}{t[y/u]}  &:= & \mul{x}{t}+ \mul{y}{t} \cdot \mul{x}{u}\\
    &&\mbox{ if } t = \mX_\Del [x_1/u_1]\ldots[x_n/u_n], y \in \Del 
   \end{array} \] 

Observe that $\size(t) \geq 1$ and $\mul{x}{t} \geq 0$.

The measure $\mul{x}{\_}$ places an upper bound on the number of free occurrences of $x$ in $\m$-reducts of $t$. The 
 last definition for $\mul{x}{t}$ comes from the intuition that
the number of free occurrences  of $x$  which give rise to redexes
in $t[y/u]$ is $\mul{x}{t} + \mul{x}{u} + (\mul{y}{t}-1) \cdot \mul{x}{u}$ as
the free occurrence of $y$ in $\mX_\Del$  does not create a redex by the
definition of  $\Rew{\usm}$.

The last  definition for $\size(t[x/u])$  is similar; the  $x$ in
the  metavariable is  counted by  $\mul{x}{t} \cdot  \size(u)$ so
since  we  do not  have  a  $\Rew{\usm}$-redex,  we subtract  one
$\size(u)$ and  since the $x$  in the metavariable is  useless in
terms of  reduction, we subtract  $1$ \ie the expression  for the
last case is  $\size(t) + \size(u) + \mul{x}{t}  \cdot \size(u) -
\size(u) - 1$. For example, the term $\mX_{\{x\}}[x/u]$ has size $\size(u)$.


\begin{lemma}
\label{l:size-mul-invariant-up-to-Com}
Let $t,t'$ be $\mterms$-metaterms. If $t =_\Com t'$, then:
\begin{enumerate}
\item $\mul{z}{t} = \mul{z}{t'}$ 
\item $\size(t) = \size(t')$;
\end{enumerate}
\end{lemma}

\begin{proof} By induction on  $t =_\Com t'$. 
We first show the four interesting cases at
  the root having the form 
$t= t_1[x/u][y/v] =_\Com t_1[y/v][x/u]=t'$, in
  all of them we have $y \notin \fv(u)$ and $x \notin \fv(v)$.

\begin{itemize} 
\item $t_1 = \mX_\Del [x_1/u_1]\ldots[x_n/u_n]$, $x, y \in \Del$. We have 
\begin{enumerate}
\item 
\[ \begin{array}{ll}
  & \mul{z}{t_1[x/u][y/v]}\\
= & \mul{z}{t_1[x/u]} + \mul{y}{t_1[x/u]} \cdot \mul{z}{v} \\ 
= & \mul{z}{t_1} + \mul{x}{t_1} \cdot \mul{z}{u} + (\mul{y}{t_1} + \mul{x}{t_1} \cdot \mul{y}{u}) \cdot \mul{z}{v} \\ 
= & \mul{z}{t_1} + \mul{x}{t_1} \cdot \mul{z}{u} + \mul{y}{t_1} \cdot \mul{z}{v} \\ 
= & \mul{z}{t_1} + \mul{y}{t_1} \cdot \mul{z}{v} + (\mul{x}{t_1} + \mul{y}{t_1} \cdot \mul{x}{v}) \cdot \mul{z}{u} \\ 
= & \mul{z}{t_1[y/v]} + \mul{x}{t_1[y/v]} \cdot \mul{z}{u} \\ 
& \mul{z}{t_1[y/v][x/u]}
\end{array} \]

\item \[ \begin{array}{ll}
  & \size(t_1[x/u][y/v])\\
= & \size(t_1[x/u]) - 1 + \mul{y}{(t_1[x/u]} \cdot \size(v) \\
= & \size(t_1) - 1 + \mul{x}{t_1} \cdot \size(u) - 1 + \mul{y}{t_1}\cdot \size(v) +  \mul{x}{t_1} \cdot \mul{y}{u} \cdot \size(v) \\
= & \size(t_1) - 1 + \mul{y}{t_1} \cdot \size(v) - 1 + \mul{x}{t_1} \cdot \size(u) + 0 \\
= & \size(t_1) - 1 + \mul{y}{t_1} \cdot \size(v) - 1 + \mul{x}{t_1} \cdot \size(u) + \mul{y}{t_1} \cdot \mul{x}{v} \cdot \size(u) \\
= & \size(t_1[y/v]) - 1 + \mul{x}{(t_1[y/v]} \cdot \size(u) \\
& \size(t_1[y/v][x/u])
\end{array} \]
\end{enumerate}
\item $t_1 = \mX_\Del [x_1/u_1]\ldots[x_n/u_n]$, $x \in \Del$ $y \notin \Del$. We have 

\begin{enumerate}
\item 

\[ \begin{array}{ll}
  & \mul{z}{t_1[x/u][y/v]}\\
= & \mul{z}{t_1[x/u]} + \mul{z}{v} + \mul{y}{t_1[x/u]} \cdot \mul{z}{v} \\ 
= & \mul{z}{t_1} + \mul{x}{t_1} \cdot \mul{z}{u} + \mul{z}{v} + (\mul{y}{t_1} + \mul{x}{t_1} \cdot \mul{y}{u}) \cdot \mul{z}{v} \\ 
= & \mul{z}{t_1} + \mul{x}{t_1} \cdot \mul{z}{u} + \mul{z}{v} + \mul{y}{t_1} \cdot \mul{z}{v} \\ 
= & \mul{z}{t_1} + \mul{z}{v} + \mul{y}{t_1} \cdot \mul{z}{v} + (\mul{x}{t_1} + \mul{x}{v} + \mul{y}{t_1} \cdot \mul{x}{v}) \cdot \mul{z}{u} \\ 
= & \mul{z}{t_1[y/v]} + \mul{x}{t_1[y/v]} \cdot \mul{z}{u} \\ 
= & \mul{z}{t_1[y/v][x/u]}
\end{array} \]

\item \[ \begin{array}{ll}
  & \size(t_1[x/u][y/v])\\
= & \size(t_1[x/u]) + \size(v) + \mul{y}{(t_1[x/u]} \cdot \size(v) \\
= & \size(t_1) - 1 + \mul{x}{t_1} \cdot \size(u) + \size(v) + \mul{y}{t_1}\cdot \size(v) +  \mul{x}{t_1} \cdot \mul{y}{u} \cdot \size(v) \\
= & \size(t_1) - 1 + \mul{y}{t_1} \cdot \size(v) + \size(v) + \mul{x}{t_1} \cdot \size(u) + 0 \\
= & \size(t_1) + \size(v) + \mul{y}{t_1} \cdot \size(v) - 1 + \mul{x}{t_1} \cdot \size(u) \\
	= & \size(t_1) + \size(v) + \mul{y}{t_1} \cdot \size(v) - 1 + (\mul{x}{t_1} + \mul{x}{v} + \mul{y}{t_1} \cdot \mul{x}{v}) \cdot \size(u) \\
= & \size(t_1[y/v]) - 1 + \mul{x}{t_1[y/v]} \cdot \size(u) \\
& \size(t_1[y/v][x/u])
\end{array} \]
\end{enumerate}

\item $t_1 = \mX_\Del [x_1/u_1]\ldots[x_n/u_n]$, $x \notin \Del$ $y \in \Del$. Similar to the previous case.

\item $t_1 \neq \mX_\Del [x_1/u_1]\ldots[x_n/u_n]$ where $x,y \in \Del$. We have

\begin{enumerate}
\item 
\[ \begin{array}{ll}
  & \mul{z}{t_1[x/u][y/v]}\\
= & \mul{z}{t_1[x/u]} + \mul{z}{v} + \mul{y}{t_1[x/u]} \cdot \mul{z}{v} \\ 
= & \mul{z}{t_1} + \mul{z}{u} + \mul{x}{t_1} \cdot \mul{z}{u} + \mul{z}{v} + (\mul{y}{t_1} + \mul{y}{u} + \mul{x}{t_1} \cdot \mul{y}{u}) \cdot \mul{z}{v} \\ 
= & \mul{z}{t_1} + \mul{z}{u} + \mul{x}{t_1} \cdot \mul{z}{u} + \mul{z}{v} + \mul{y}{t_1} \cdot \mul{z}{v} \\ 
= & \mul{z}{t_1} + \mul{z}{v} + \mul{y}{t_1} \cdot \mul{z}{v} + \mul{z}{u} + (\mul{x}{t_1} + \mul{x}{v} + \mul{y}{t_1} \cdot \mul{x}{v}) \cdot \mul{z}{u} \\ 
= & \mul{z}{t_1[y/v]} + \mul{z}{u} + \mul{x}{t_1[y/v]} \cdot \mul{z}{u} \\ 
= & \mul{z}{t_1[y/v][x/u]}
\end{array} \]

\item
\[ \begin{array}{ll}
  & \size(t_1[x/u][y/v])\\
= & \size(t_1[x/u]) + \size(v) + \mul{y}{(t_1[x/u]} \cdot \size(v) \\
= & \size(t_1) + \size(u) + \mul{x}{t_1} \cdot \size(u) + \size(v) + \mul{y}{t_1}\cdot \size(v) + \mul{y}{u}\cdot \size(v) + \mul{x}{t_1} \cdot \mul{y}{u} \cdot \size(v) \\
= & \size(t_1) + \size(u) + \mul{x}{t_1} \cdot \size(u) + \size(v) + \mul{y}{t_1}\cdot \size(v) \\
= & \size(t_1) + \size(v) + \mul{y}{t_1} \cdot \size(v) + \size(u) + \mul{x}{t_1}\cdot \size(u) + \mul{x}{v}\cdot \size(u) + \mul{y}{t_1} \cdot \mul{x}{v} \cdot \size(u) \\
= & \size(t_1[y/v]) + \size(u) + \mul{x}{(t_1[y/v]} \cdot \size(u) \\
& \size(t_1[y/v][x/u])
\end{array} \]
\end{enumerate}
\end{itemize}

We now consider the inductive cases.

\begin{itemize}
\item $t = t_1\ t_2 =_\Com t_1\ t'_2 = t'$ with $t_2 =_\Com t'_2$.

\begin{enumerate}
\item $\mul{z}{t} = \mul{z}{t_1\ t_2} = \mul{z}{t_1} + \mul{z}{t_2} =_{\ih} \mul{z}{t_1} + \mul{z}{t'_2} = \mul{z}{t_1\ t'_2} = \mul{z}{t'}$.

\item 
$\size(t) = \size(t_1\ t_2) = \size(t_1) + \size(t_2) =_{\ih} \size(t_1) + \size(t'_2) = \size(t_1\ t'_2) = \size(t')$.
\end{enumerate}

\item $t = \l y. t_1 =_\Com \l y. t'_1 = t'$ and $t = t_1\ u =_\Com t'_1\ u = t'$ with $t_1 =_\Com t'_1$. Similar to the previous case.

\item $t = t_1[y/t_2] =_\Com t'_1[y/t_2] = t'$ with $t_1 =_\Com t'_1$.

  Consider $t_1  = \mX_\Del[x_1/u_1]\ldots[x_n/u_n]$ with  $y \in
  \Del$.   Then necessarily  $t'_1 =
  \mX_\Del[y_1/v_1]\ldots[y_n/v_n]$, where $v_i = u_{\sigma(i)}$,
  for some permutation $\sigma$. Moreover,  
  we can assume that $y \notin \bv(t_1)$ by $\alpha$-equivalence.

  Thus,

\begin{enumerate}
\item $\mul{z}{t}   =  \mul{z}{t_1[y/t_2]}  =   \mul{z}{t_1}  +
  \mul{y}{t_1}   \cdot  \mul{z}{t_2}   =_{\ih}   \mul{z}{t'_1}  +
  \mul{y}{t'_1}  \cdot   \mul{z}{t_2}  =  \mul{z}{t'_1[y/t_2]}  =
  \mul{z}{t'}$

\item $\size(t) = \size(t_1[y/t_2]) = \size(t_1) - 1
  +  \mul{y}{t_1}  \cdot \size(t_2)  =_{\ih}  \size(t'_1)  - 1  +
  \mul{y}{t'_1}   \cdot   \size(t_2)   =   \size(t'_1[y/t_2])   =
  \size(t')$.
\end{enumerate}

Otherwise,

\begin{enumerate}
\item $\mul{z}{t}  = \mul{z}{t_1[y/t_2]} =  \mul{z}{t_1} +
  \mul{z}{t_2}   +   \mul{y}{t_1}   \cdot  \mul{z}{t_2}   =_{\ih}
  \mul{z}{t'_1} + \mul{z}{t_2} + \mul{y}{t'_1} \cdot \mul{z}{t_2}
  =   \mul{z}{t'_1[y/t_2]}  =   \mul{z}{t'}$

\item  $\size(t)   =
  \size(t_1[y/t_2])  =  \size(t_1)  + \size(t_2)  +  \mul{y}{t_1}
  \cdot   \size(t_2)   =_{\ih}   \size(t'_1)   +   \size(t_2)   +
  \mul{y}{t'_1}   \cdot   \size(t_2)   =   \size(t'_1[y/t_2])   =
  \size(t')$.
\end{enumerate}

\item $t = t_1[y/t_2] =_\Com t_1[y/t'_2] = t'$ with $t_2 =_\Com t'_2$.

  Consider $t_1  = \mX_\Del[x_1/u_1]\ldots[x_n/u_n]$ with  $y \in
  \Del$. Then,

\begin{enumerate}
\item $\mul{z}{t}  = \mul{z}{t_1[y/t_2]} = \mul{z}{t_1} +
  \mul{y}{t_1}   \cdot   \mul{z}{t_2}   =_{\ih}  \mul{z}{t_1}   +
  \mul{y}{t_1}  \cdot   \mul{z}{t'_2}  =  \mul{z}{t_1[y/t'_2]}  =
  \mul{z}{t'}$
  \item  $\size(t) = \size(t_1[y/t_2]) = \size(t_1) - 1
  +  \mul{y}{t_1}  \cdot  \size(t_2)  =_{\ih} \size(t_1)  -  1  +
  \mul{y}{t_1}   \cdot   \size(t'_2)   =   \size(t_1[y/t'_2])   =
  \size(t')$.
\end{enumerate}

  Otherwise,

\begin{enumerate}
\item 
$\mul{z}{t}  = \mul{z}{t_1[y/t_2]} =  \mul{z}{t_1} +
  \mul{z}{t_2}   +   \mul{y}{t_1}   \cdot  \mul{z}{t_2}   =_{\ih}
  \mul{z}{t_1} + \mul{z}{t'_2} + \mul{y}{t_1} \cdot \mul{z}{t'_2}
  =   \mul{z}{t_1[y/t'_2]}  =   \mul{z}{t'}$

\item $\size(t)   =
  \size(t_1[y/t_2])  =  \size(t_1)  + \size(t_2)  +  \mul{y}{t_1}
  \cdot   \size(t_2)   =_{\ih}   \size(t_1)   +   \size(t'_2)   +
  \mul{y}{t_1}   \cdot   \size(t'_2)   =   \size(t_1[y/t'_2])   =
  \size(t')$.
\end{enumerate}
\end{itemize}
\end{proof}

\ignore{
For  Delia  rather  than  the  report:  I  think  the  definition
$\size(\mX_\Del) =  |\Del|$ is more faithful to  the intention of
$\size{}$  than  $\size(\mX_\Del)  =   1$  which  I  had  defined
previously; consider $\size(\mX_{\{x,y\}}[y/v][x/u])$:
\[
\begin{array}{llll}
& \mbox{Let } \size(\mX_\Del) = 1: \\
 & \size(\mX_{\{x,y\}}[y/v][x/u]) \\
= & \size(\mX_{\{x,y\}}[y/v]) - 1 + \mul{x}{\mX_{\{x,y\}}[y/v]} \cdot \size(u)\\
= & \size(\mX_{\{x,y\}}) - 1 + \size(v) - 1 + \mul{x}{\mX_{\{x,y\}}[y/v]} \cdot \size(u)\\
= & \size(v) - 1 + \mul{x}{\mX_{\{x,y\}}[y/v]} \cdot \size(u)\\
= & \size(v) - 1 + \mul{x}{\mX_{\{x,y\}}}\cdot \size(u) + \mul{y}{\mX_{\{x,y\}}} \cdot \mul{x}{v} \cdot \size(u)\\
= & \size(v) - 1 + \size(u) + \mul{x}{v} \cdot \size(u)\\
\end{array}
\]
whereas
\[
\begin{array}{llll}
& \mbox{Let } \size(\mX_\Del) = |\Del|: \\
 & \size(\mX_{\{x,y\}}[y/v][x/u]) \\
= & \size(\mX_{\{x,y\}}[y/v]) - 1 + \mul{x}{\mX_{\{x,y\}}[y/v]} \cdot \size(u)\\
= & \size(\mX_{\{x,y\}}) - 1 + \size(v) - 1 + \mul{x}{\mX_{\{x,y\}}[y/v]} \cdot \size(u) \\
= & \size(v) + \mul{x}{\mX_{\{x,y\}}[y/v]} \cdot \size(u) \\
= & \size(v) + \mul{x}{\mX_{\{x,y\}}}\cdot \size(u) + \mul{y}{\mX_{\{x,y\}}} \cdot \mul{x}{v} \cdot \size(u)\\
= & \size(v) + \size(u) + \mul{x}{v} \cdot \size(u)\\
\end{array}
\]
This last equation seems what we want.
}

We now extend the previous  measures to contexts by adding
 $\mul{x}{\Box} := 0$ and $\size(\Box) := 0$. 

\begin{lemma} \label{l:for-decreasing00}
Let $v$ be a $\mterms$-metaterm
such that $x,y \notin \fv(v)$ and
$x \neq  y$. Let $x \in \Del$. Then, 
\begin{enumerate}
\item\label{l:for-decreasing00-1}  $\mul{x}{C\cwc{x}} > \mul{x}{C\cwc{v}}$. 
\item\label{l:for-decreasing00-2} $\mul{y}{C\cwc{x}} \geq \mul{y}{C\cwc{v}}$. 
\item\label{l:for-decreasing00-3} $\mul{x}{C\cwc{\mX_\Del}_{x,\fv(v)}} >    \mul{x}{C\cwc{\mX_\Del[x/v]}_{x,\fv(v)}}$. 
\item\label{l:for-decreasing00-4} $\mul{y}{C\cwc{\mX_\Del}_{x,\fv(v)}} \geq \mul{y}{C\cwc{\mX_\Del[x/v]}_{x,\fv(v)}}$. 
\end{enumerate}
\end{lemma}

\begin{proof} By induction on $C$. Let $\phi = \{x\} \cup \fv(v)$.
\begin{itemize}

\item 

$C = \Box$.

\begin{enumerate}
\item $\mul{x}{x} = 1 > \mul{x}{v}= 0$.  

\item $\mul{y}{x}
  = 0 \geq \mul{y}{v}=
  0$. 

\item  $\mul{x}{\mX_\Del} = 1 > 0 = \mul{x}{\mX_\Del[x/v]}$, since by $\alpha$-conversion we can assume $x \notin \fv(\mX_\Del[x/v])$. 

\item  If $y \in \Del$,
 $\mul{y}{\mX_\Del}
 = \mul{y}{\mX_\Del} + 0 =
\mul{y}{\mX_\Del} + \mul{x}{\mX_\Del} \cdot \mul{y}{v} =
  \mul{y}{\mX_\Del[x/v]}.$

If $y \notin \Del$,
 $\mul{y}{\mX_\Del}
 = \mul{y}{\mX_\Del} + 0 =
\mul{y}{\mX_\Del} + \mul{y}{v} + \mul{x}{\mX_\Del} \cdot \mul{y}{v} =
  \mul{y}{\mX_\Del[x/v]}.$  
 \end{enumerate}

\item $C = D\ t$. 

\begin{enumerate}
\item       $\mul{x}{(D\ t)\cwc{x}} = 
       \mul{x}{D\cwc{x}\ t} =
       \mul{x}{D\cwc{x}} + \mul{x}{t} >_{\ih\ \ref{l:for-decreasing00-1}}
       \mul{x}{D\cwc{v}} + \mul{x}{t} = 
       \mul{x}{(D\ t)\cwc{v}}$

\item        $\mul{y}{(D\ t)\cwc{x}} = 
       \mul{y}{D\cwc{x}\ t} =
       \mul{y}{D\cwc{x}} + \mul{y}{t} \geq_{\ih\ 2}
       \mul{y}{D\cwc{v}} + \mul{y}{t} = 
       \mul{y}{(D\ t)\cwc{v}}$

 \item  			$\mul{x}{(D\ t)\cwc{\mX_\Del}_{\phi}} = 
       \mul{x}{D\cwc{\mX_\Del}_{\phi}\ t} =
       \mul{x}{D\cwc{\mX_\Del}_{\phi}} + \mul{x}{t} >_{\ih\ 3}
       \mul{x}{D\cwc{\mX_\Del[x/v]}_{\phi}} + \mul{x}{t} = 
       \mul{x}{(D\ t)\cwc{\mX_\Del[x/v]}_{\phi}}$

\item        $\mul{y}{(D\ t)\cwc{\mX_\Del}_{\phi}} = 
       \mul{y}{D\cwc{\mX_\Del}_{\phi}\ t} =
       \mul{y}{D\cwc{\mX_\Del}_{\phi}} + \mul{y}{t} \geq_{\ih\ 4}
       \mul{y}{D\cwc{\mX_\Del[x/v]}_{\phi}} + \mul{y}{t} = 
       \mul{y}{(D\ t)\cwc{\mX_\Del[x/v]}_{\phi}}$
\end{enumerate}

\item $C= t\ D$ and $C= \l z. D$ are similar.

\item $C=u[z/D]$.
By the  \ih\ \ref{l:for-decreasing00-1}, \ih\ \ref{l:for-decreasing00-2},
\ih\ \ref{l:for-decreasing00-3}, and \ih\ \ref{l:for-decreasing00-4}  we have $\mul{x}{D\cwc{x}} > \mul{x}{D\cwc{v}}$,
$\mul{y}{D\cwc{x}} \geq \mul{y}{D\cwc{v}}$,
$\mul{x}{D\cwc{\mX_\Del}_\phi} > \mul{x}{D\cwc{\mX_\Del[x/v]}_\phi}$, and
$\mul{y}{D\cwc{\mX_\Del}_\phi} \geq \mul{y}{D\cwc{\mX_\Del[x/v]}_\phi}$ respectively. 
 Also, by $\alpha$-conversion we can
  assume $z \notin \fv(v)$ so that 
  $\mul{z}{D\cwc{x}} \geq \mul{z}{D\cwc{v}}$ by the \ih\ \ref{l:for-decreasing00-2} and $\mul{z}{D\cwc{\mX_\Del}_\phi} \geq \mul{z}{D\cwc{\mX_\Del[x/v]}_\phi}$ by the \ih\ \ref{l:for-decreasing00-4}. 
  
\begin{enumerate}
\item       $\mul{x}{u[z/D]\cwc{x}} = 
       \mul{x}{u[z/D\cwc{x}]} =
       \mul{x}{u} + \mul{x}{D\cwc{x}} + \mul{z}{u} \cdot \mul{x}{D\cwc{x}} >
       \mul{x}{u} + \mul{x}{D\cwc{v}} + \mul{z}{u} \cdot \mul{x}{D\cwc{v}} = 
       \mul{x}{u[z/D\cwc{v}]}=
       \mul{x}{u[z/D]\cwc{v}}$

\item       $\mul{y}{u[z/D]\cwc{x}} = 
       \mul{y}{u[z/D\cwc{x}]} =
       \mul{y}{u} + \mul{y}{D\cwc{x}} + \mul{z}{u} \cdot \mul{y}{D\cwc{x}} \geq 
       \mul{y}{u} + \mul{y}{D\cwc{v}} + \mul{z}{u} \cdot \mul{y}{D\cwc{v}} =
       \mul{y}{u[z/D\cwc{v}]}=
       \mul{y}{u[z/D]\cwc{v}}$

\item Then, 
\[
\begin{array}{ll}
& \mul{x}{u[z/D]\cwc{\mX_\Del}_{\phi}}\\
 = &
\mul{x}{u[z/D\cwc{\mX_\Del}_{\phi}]}\\
 = &
\mul{x}{u} + m + \mul{z}{u} \cdot \mul{x}{D\cwc{\mX_\Del}_{\phi}} \\
 >&
\mul{x}{u} + m' + \mul{z}{u} \cdot \mul{x}{D\cwc{\mX_\Del[x/v]}_{\phi}} \\
 = &
\mul{x}{u[z/D\cwc{\mX_\Del[x/v]}_{\phi}]}\\
= &
\mul{x}{u[z/D]\cwc{\mX_\Del[x/v]}_{\phi}}
\end{array}
\]
where $m, m' = 0$ if $u = \mX_{\Del'} [x_1/u_1]\ldots[x_n/u_n], z
\in  \Del'$   (in  which  case   $\mul{z}{u}  >  0$)  and   $m  =
\mul{x}{D\cwc{\mX_\Del}_{\phi}}$,              $m'              =
\mul{x}{D\cwc{\mX_\Del[x/v]}_{\phi}}$  otherwise (hence $m > m'$ by \ih\ 3). 
\item Then, 
\[
\begin{array}{ll}
& \mul{y}{u[z/D]\cwc{\mX_\Del}_{\phi}}\\
 = &
\mul{y}{u[z/D\cwc{\mX_\Del}_{\phi}]}\\
 = &
\mul{y}{u} + m + \mul{z}{u} \cdot \mul{y}{D\cwc{\mX_\Del}_{\phi}} \\
 >&
\mul{y}{u} + m' + \mul{z}{u} \cdot \mul{y}{D\cwc{\mX_\Del[x/v]}_{\phi}} \\
 = &
\mul{y}{u[z/D\cwc{\mX_\Del[x/v]}_{\phi}]}\\
= &
\mul{y}{u[z/D]\cwc{\mX_\Del[x/v]}_{\phi}}
\end{array}
\]
where $m, m' = 0$ if $u = \mX_{\Del'} [x_1/u_1]\ldots[x_n/u_n], z
\in \Del'$ (in which case $\mul{z}{u} > 0$) and $m =
\mul{y}{D\cwc{\mX_\Del}_{\phi}}$, $m' =
\mul{y}{D\cwc{\mX_\Del[x/v]}_{\phi}}$ otherwise (hence $m > m'$ by \ih\ 4).
\end{enumerate}

\item $C=D[z/u]$. 
By
  the \ih\ \ref{l:for-decreasing00-1}, \ih\ \ref{l:for-decreasing00-2}, \ih\
  \ref{l:for-decreasing00-3}, and \ih\ \ref{l:for-decreasing00-4} we have $\mul{x}{D\cwc{x}} > \mul{x}{D\cwc{v}}$,
  $\mul{y}{D\cwc{x}} \geq \mul{y}{D\cwc{v}}$,
  $\mul{x}{D\cwc{\mX_\Del}_\phi}
  > \mul{x}{D\cwc{\mX_\Del[x/v]}_\phi}$, and
  $\mul{y}{D\cwc{\mX_\Del}_\phi} \geq \mul{y}{D\cwc{\mX_\Del[x/v]}_\phi}$
  respectively.  Also, by $\alpha$-conversion we can assume $z \neq x$
  and by definition of our notation we can assume $z \notin \fv(v)$.
  Thus, $\mul{z}{D\cwc{x}} \geq \mul{z}{D\cwc{v}}$ by the \ih\ \ref{l:for-decreasing00-2} and
  $\mul{z}{D\cwc{\mX_\Del}_\phi} \geq \mul{z}{D\cwc{\mX_\Del[x/v]}_\phi}$
  by the \ih\ \ref{l:for-decreasing00-4}.

\begin{enumerate}
\item 
      $\mul{x}{D[z/u]\cwc{x}} = 
       \mul{x}{D\cwc{x}[z/u]} =
       \mul{x}{D\cwc{x}} + \mul{x}{u} + \mul{z}{D\cwc{x}} \cdot \mul{x}{u} >
       \mul{x}{D\cwc{v}} + \mul{x}{u} + \mul{z}{D\cwc{v}} \cdot \mul{x}{u} = 
       \mul{x}{D\cwc{v}[z/u]}=
       \mul{x}{D[z/u]\cwc{v}}$.

\item 
      $\mul{y}{D[z/u]\cwc{x}} = 
       \mul{y}{D\cwc{x}[z/u]} =
       \mul{y}{D\cwc{x}} + \mul{y}{u} + \mul{z}{D\cwc{x}} \cdot \mul{y}{u} \geq
       \mul{y}{D\cwc{v}} + \mul{y}{u} + \mul{z}{D\cwc{v}} \cdot \mul{y}{u} = 
       \mul{y}{D\cwc{v}[z/u]}=
       \mul{y}{D[z/u]\cwc{v}}$
       

\item Then, 
\[
\begin{array}{ll}
& \mul{x}{D[z/u]\cwc{\mX_\Del}_{\phi}}\\
 = &
\mul{x}{D\cwc{\mX_\Del}_{\phi}[z/u]}\\
 = &
\mul{x}{D\cwc{\mX_\Del}_{\phi}} + m + \mul{z}{D\cwc{\mX_\Del}_{\phi}} \cdot \mul{x}{u} \\
 >&
\mul{x}{D\cwc{\mX_\Del[x/v]}_{\phi}} + m + \mul{z}{D\cwc{\mX_\Del[x/v]}_{\phi}} \cdot \mul{x}{u} \\
 = &
\mul{x}{D\cwc{\mX_\Del[x/v]}_{\phi}[z/u]}\\
= &
\mul{x}{D[z/u]\cwc{\mX_\Del[x/v]}_{\phi}}
\end{array}
\]
as $z \neq x$ where $m = 0$ if $D\cwc{\mX_\Del}_{\phi} = \mX_{\Del'} [x_1/u_1]\ldots[x_n/u_n], y \in \Del'$ and $\mul{x}{u}$ otherwise.

\item Then, 
\[
\begin{array}{ll}
& \mul{y}{D[z/u]\cwc{\mX_\Del}_{\phi}}\\
 = &
\mul{y}{D\cwc{\mX_\Del}_{\phi}[z/u]}\\
 = &
\mul{y}{D\cwc{\mX_\Del}_{\phi}} + m + \mul{z}{D\cwc{\mX_\Del}_{\phi}} \cdot \mul{y}{u} \\
 \geq&
\mul{y}{D\cwc{\mX_\Del[x/v]}_{\phi}} + m + \mul{z}{D\cwc{\mX_\Del[x/v]}_{\phi}} \cdot \mul{y}{u} \\
 = &
\mul{y}{D\cwc{\mX_\Del[x/v]}_{\phi}[z/u]}\\
= &
\mul{y}{D[z/u]\cwc{\mX_\Del[x/v]}_{\phi}}
\end{array}
\]
as $z \neq x$ where $m = 0$ if $D\cwc{\mX_\Del}_{\phi} = \mX_{\Del'} [x_1/u_1]\ldots[x_n/u_n], y \in \Del$ and $\mul{y}{u}$ otherwise.
\end{enumerate}         
\end{itemize}
\end{proof}

\begin{lemma} \label{l:for-decreasing0}
Let $v$ be  a $\mterms$-metaterm such that $y \notin \fv(v)$.
Let $x \neq y$. Then,  
$\mul{x}{C\cwc{y}} + \mul{y}{C\cwc{y}} \cdot \mul{x}{v} \geq  
 \mul{x}{C\cwc{v}} + \mul{y}{C\cwc{v}} \cdot \mul{x}{v}$.
\end{lemma}

\begin{proof} By induction on $C$. Remark that $\mul{y}{v} = 0$ by
  $\alpha$-conversion. 

\begin{itemize}
\item $C = \Box$. 

     \[ \begin{array}{ll} 
      \mul{x}{y} + \mul{y}{y} \cdot \mul{x}{v} & = \\ 
      0 + \mul{x}{v} & = \\  
      \mul{x}{v} +  0  & = \\
      \mul{x}{v} + \mul{y}{v} \cdot \mul{x}{v}
      \end{array} \] 

\item $C = D\ t$. 

    \[ \begin{array}{ll} 
       \mul{x}{D\cwc{y}\ t} + \mul{y}{D\cwc{y}\ t} \cdot \mul{x}{v} & =       \\
       \mul{x}{D\cwc{y}} + \mul{x}{t} + \mul{y}{D\cwc{y}} \cdot \mul{x}{v} + \mul{y}{t} \cdot \mul{x}{v}
        & \geq_{\ih}\\
        \mul{x}{D\cwc{v}} + \mul{x}{t} + \mul{y}{D\cwc{v}}\cdot \mul{x}{v}+  \mul{y}{t}
      \cdot \mul{x}{v}& = \\
      \mul{x}{D\cwc{v}\ t} + \mul{y}{D\cwc{v}\ t} \cdot \mul{x}{v}
      \end{array} \]

\item $C= t\ D$ and $C= \l z. D$ are similar.

\item $C=D[z/u]$. W.l.o.g. we can assume $z \notin \fv(v)$
  and $z \neq x$.

\[ \begin{array}{ll}
\mul{x}{D\cwc{y}[z/u]} + \mul{y}{D\cwc{y}[z/u]} \cdot \mul{x}{v} & = \\
\mul{x}{D\cwc{y}} + \mul{x}{u} + \mul{z}{D\cwc{y}} \cdot \mul{x}{u} + 
\mul{y}{D\cwc{y}} \cdot \mul{x}{v} + & \\
\mul{y}{u}\cdot \mul{x}{v} + \mul{z}{D\cwc{y}} \cdot \mul{y}{u}
\cdot \mul{x}{v} & \geq_{\ih}  \\
\mul{x}{D\cwc{v}} + \mul{x}{u} + \mul{z}{D\cwc{y}} \cdot \mul{x}{u} + 
\mul{y}{D\cwc{v}} \cdot \mul{x}{v} + & \\
\mul{y}{u}\cdot \mul{x}{v} + \mul{z}{D\cwc{y}} \cdot \mul{y}{u}
\cdot \mul{x}{v} & \geq_{\mbox{\cref{l:for-decreasing00}}} \\
\mul{x}{D\cwc{v}} + \mul{x}{u} + \mul{z}{D\cwc{v}} \cdot \mul{x}{u} + 
\mul{y}{D\cwc{v}} \cdot \mul{x}{v} + & \\
 \mul{y}{u} \cdot \mul{x}{v} + \mul{z}{D\cwc{v}} \cdot \mul{y}{u}
\cdot \mul{x}{v} & = \\
\mul{x}{D\cwc{v}[z/u]} + \mul{y}{D\cwc{v}[z/u]} \cdot \mul{x}{v} 
\end{array} \]

\item $C=u[z/D]$. 

\[ \begin{array}{ll}
\mul{x}{u[z/D\cwc{y}]} + \mul{y}{u[z/D\cwc{y}]} \cdot \mul{x}{v} & = \\
\mul{x}{u} + \mul{x}{D\cwc{y}} + \mul{z}{u} \cdot \mul{x}{D\cwc{y}} +
\mul{y}{u} \cdot \mul{x}{v}  + & \\
\mul{y}{D\cwc{y}} \cdot \mul{x}{v}  +
\mul{z}{u} \cdot \mul{y}{D\cwc{y}} \cdot \mul{x}{v}  &  = \\ 
\mul{x}{u} + (1+ \mul{z}{u}) \cdot (\mul{x}{D\cwc{y}} +\mul{y}{D\cwc{y}}
\cdot \mul{x}{v}) + \mul{y}{u} \cdot \mul{x}{v} 
& \geq_{\ih}\\ 
\mul{x}{u} + (1+ \mul{z}{u}) \cdot (\mul{x}{D\cwc{v}} +\mul{y}{D\cwc{v}} \cdot \mul{x}{v})+ \mul{y}{u} \cdot \mul{x}{v} 
& = \\ 
\mul{x}{u} + \mul{x}{D\cwc{v}} + \mul{z}{u} \cdot \mul{x}{D\cwc{v}} + 
\mul{y}{u}\cdot \mul{x}{v} + & \\
\mul{y}{D\cwc{v}}\cdot \mul{x}{v} + \mul{z}{u} \cdot \mul{y}{D\cwc{v}} \cdot \mul{x}{v} & = \\
\mul{x}{u[z/D\cwc{v}]} + \mul{y}{u[z/D\cwc{v}]} \cdot \mul{x}{v} \\
\end{array} \]

\end{itemize}
\end{proof}

\begin{lemma} \label{l:for-decreasing0usm}
Let $v$ be a $\mterms$-metaterm such that $y \notin \fv(v)$.
Let $\phi = \{y\} \cup \fv(v)$,
$x \neq y$ and $y \in \Del$. Then, \\
$\mul{x}{C\cwc{\mX_\Del}_{\phi}} + \mul{y}{C\cwc{\mX_\Del}_{\phi}} \cdot \mul{x}{v} \geq  
 \mul{x}{C\cwc{\mX_\Del[y/v]}_{\phi}} + \mul{y}{C\cwc{\mX_\Del[y/v]}_{\phi}} \cdot \mul{x}{v}$.
\end{lemma}

\begin{proof} By induction on $C$. 
\begin{itemize}
\item $C = \Box$. 
     \[ \begin{array}{ll} 
      \mul{x}{\mX_\Del} + \mul{y}{\mX_\Del} \cdot \mul{x}{v} & = \\ 
      \mul{x}{\mX_\Del[y/v]} + \mul{y}{\mX_\Del[y/v]} \cdot \mul{x}{v}
      \end{array} \] 

\item $C = D\ t$. 

    \[ \begin{array}{ll} 
       \mul{x}{D\cwc{\mX_\Del}_{\phi}\ t} + \mul{y}{D\cwc{\mX_\Del}_{\phi}\ t} \cdot \mul{x}{v} & =       \\
       \mul{x}{D\cwc{\mX_\Del}_{\phi}} + \mul{x}{t} + \mul{y}{D\cwc{\mX_\Del}_{\phi}} \cdot \mul{x}{v} + \mul{y}{t} \cdot \mul{x}{v}
        & \geq_{\ih}\\
        \mul{x}{D\cwc{\mX_\Del[y/v]}_{\phi}} + \mul{x}{t} + \mul{y}{D\cwc{\mX_\Del[y/v]}_{\phi}}\cdot \mul{x}{v}+  \mul{y}{t}
      \cdot \mul{x}{v}& = \\
      \mul{x}{D\cwc{\mX_\Del[y/v]}_{\phi}\ t} + \mul{y}{D\cwc{\mX_\Del[y/v]}_{\phi}\ t} \cdot \mul{x}{v}
      \end{array} \]

\item $C= t\ D$ and $C= \l z. D$ are similar.

\item $C=D[z/u]$. W.l.o.g. we can assume $z \notin \fv(v)$.

\[ \begin{array}{ll}
\mul{x}{D\cwc{\mX_\Del}_{\phi}[z/u]} + \mul{y}{D\cwc{\mX_\Del}_{\phi}[z/u]} \cdot \mul{x}{v} & = \\
\mul{x}{D\cwc{\mX_\Del}_{\phi}} + m + \mul{z}{D\cwc{\mX_\Del}_{\phi}} \cdot \mul{x}{u} + 
\mul{y}{D\cwc{\mX_\Del}_{\phi}} \cdot \mul{x}{v}\\
+ n \cdot \mul{x}{v} + \mul{z}{D\cwc{\mX_\Del}_{\phi}} \cdot \mul{y}{u}
\cdot \mul{x}{v} & \geq_{\ih}  \\
\mul{x}{D\cwc{\mX_\Del[x/v]}_\phi} + m + \mul{z}{D\cwc{\mX_\Del}_\phi} \cdot \mul{x}{u} + 
\mul{y}{D\cwc{\mX_\Del[x/v]}_\phi} \cdot \mul{x}{v}\\
 + n \cdot \mul{x}{v} + \mul{z}{D\cwc{\mX_\Del}_\phi} \cdot \mul{y}{u}
\cdot \mul{x}{v} & \geq_{\mbox{\cref{l:for-decreasing00}}} \\
\mul{x}{D\cwc{\mX_\Del[x/v]}_\phi} + m + \mul{z}{D\cwc{\mX_\Del[x/v]}_\phi} \cdot \mul{x}{u} + 
\mul{y}{D\cwc{\mX_\Del[x/v]}_\phi} \cdot \mul{x}{v}\\
 + n \cdot \mul{x}{v} + \mul{z}{D\cwc{\mX_\Del[x/v]}_\phi} \cdot \mul{y}{u}
\cdot \mul{x}{v} & = \\
\mul{x}{D\cwc{\mX_\Del[x/v]}_\phi[z/u]} + \mul{y}{D\cwc{\mX_\Del[x/v]}_\phi[z/u]} \cdot \mul{x}{v} 
\end{array} \] 
as $z \neq x$ where $m, n = 0$ if $D\cwc{\mX_\Del}_{\phi} = \mX_{\Del'} [x_1/u_1]\ldots[x_n/u_n], y \in \Del'$ and $m = \mul{x}{u}, n = \mul{y}{u}$ otherwise.

\item $C=u[z/D]$. 

\[ \begin{array}{ll}
\mul{x}{u[z/D\cwc{\mX_\Del}_\phi]} + \mul{y}{u[z/D\cwc{\mX_\Del}_\phi]} \cdot \mul{x}{v} & = \\
\mul{x}{u} + m + \mul{z}{u} \cdot \mul{x}{D\cwc{\mX_\Del}_\phi} +
\mul{y}{u} \cdot \mul{x}{v}\\
  + n \cdot \mul{x}{v}  +
\mul{z}{u} \cdot \mul{y}{D\cwc{\mX_\Del}_\phi} \cdot \mul{x}{v}  &  = \\ 
\mul{x}{u} + (p+ \mul{z}{u}) \cdot (\mul{x}{D\cwc{\mX_\Del}_\phi} +\mul{y}{D\cwc{\mX_\Del}_\phi}
\cdot \mul{x}{v}) + \mul{y}{u} \cdot \mul{x}{v} 
& \geq_{\ih}\\ 
\mul{x}{u} + (p+ \mul{z}{u}) \cdot (\mul{x}{D\cwc{\mX_\Del[x/v]}_\phi} +\mul{y}{D\cwc{\mX_\Del[x/v]}_\phi} \cdot \mul{x}{v})+ \mul{y}{u} \cdot \mul{x}{v} 
& = \\ 
\mul{x}{u} + m' + \mul{z}{u} \cdot \mul{x}{D\cwc{\mX_\Del[x/v]}_\phi} + 
\mul{y}{u}\cdot \mul{x}{v}\\
 + n' \cdot \mul{x}{v} + \mul{z}{u} \cdot \mul{y}{D\cwc{\mX_\Del[x/v]}_\phi} \cdot \mul{x}{v} & = \\
\mul{x}{u[z/D\cwc{\mX_\Del[x/v]}_\phi]} + \mul{y}{u[z/D\cwc{\mX_\Del[x/v]}_\phi]} \cdot \mul{x}{v} \\
\end{array} \]
where $m, n, m', n', p = 0$ if $u = \mX_{\Del'} [x_1/u_1]\ldots[x_n/u_n], z \in \Del'$ (in which case $\mul{z}{u} > 0$) and $m = \mul{x}{D\cwc{\mX_\Del}_{\phi}}$, $n = \mul{y}{D\cwc{\mX_\Del}_{\phi}}$, $m' = \mul{x}{D\cwc{\mX_\Del[x/v]}_{\phi}}$, $n' = \mul{y}{D\cwc{\mX_\Del[x/v]}_{\phi}}$, $p = 1$ otherwise (hence $p \geq p'$ in both cases).

\end{itemize}
\end{proof}

\begin{lemma} \label{l:for-decreasing1}
Let $v$ be a $\mterms$-metaterm such that $x \notin \fv(v)$. 
Then, \\
$\size(C\cwc{x}) + \mul{x}{C\cwc{x}} \cdot \size(v) > \size(C\cwc{v}) +
\mul{x}{C\cwc{v}} \cdot \size(v)$.
\end{lemma}

\begin{proof} By induction on $C$. 
\begin{itemize}
\item $C = \Box$. 

      \[ \begin{array}{ll}
         \size(x) + \mul{x}{x} \cdot \size(v) & = \\ 
           1 + \size(v) & > \\
           \size(v) + 0 & = \\ 
           \size(v) +\mul{x}{v} \cdot \size(v) \\
           \end{array} \] 

\item $C = D\ t$. 

      \[ \begin{array}{ll}
       \size(D\cwc{x}\ t) + \mul{x}{D\cwc{x}\ t} \cdot \size(v) & = \\
       \size(D\cwc{x}) + \size(t) + (\mul{x}{D\cwc{x}} + \mul{x}{t})
       \cdot \size(v) & = \\
       \size(D\cwc{x}) + \size(t) + \mul{x}{D\cwc{x}}  \cdot \size(v) +
       \mul{x}{t} \cdot \size(v) & >_{\ih} \\
       \size(D\cwc{v}) + \size(t) + \mul{x}{D\cwc{v}}\cdot  \size(v) + \mul{x}{t} \cdot  \size(v) & = \\
      \size(D\cwc{v}\ t)  + \mul{x}{D\cwc{v}\ t} \cdot \size(v)
       \end{array} \]

\item $C= t\ D$ and $C= \l z. D$ are similar.

\item $C=D[z/u]$. W.l.o.g. we assume $z \notin \fv(v)$.

     \[ \begin{array}{ll}
      \size(D\cwc{x}[z/u]) + \mul{x}{D\cwc{x}[z/u]} \cdot \size(v) & = \\ 
      \size(D\cwc{x}) + \mul{z}{D\cwc{x}} \cdot \size(u) + \size(u) +  
      \mul{x}{D\cwc{x}}\cdot \size(v) + & \\
      \mul{x}{u}\cdot \size(v) + \mul{z}{D\cwc{x}} \cdot
      \mul{x}{u} \cdot \size(v) & >_{\ih} \\
      \size(D\cwc{v}) +\mul{z}{D\cwc{x}} \cdot \size(u) +  \size(u) + 
      \mul{x}{D\cwc{v}}\cdot \size(v) + & \\
      \mul{x}{u}\cdot \size(v) + \mul{z}{D\cwc{x}} \cdot
      \mul{x}{u} \cdot \size(v) & \geq_{\mbox{\cref{l:for-decreasing00}}}  \\
      \size(D\cwc{v}) +\mul{z}{D\cwc{v}} \cdot \size(u) +  \size(u) + 
      \mul{x}{D\cwc{v}}\cdot \size(v) + & \\
      \mul{x}{u}\cdot \size(v) + \mul{z}{D\cwc{v}} \cdot
      \mul{x}{u} \cdot \size(v) & = \\
      \size(D\cwc{v}[z/u]) + \mul{x}{D\cwc{v}[z/u]} \cdot \size(v)
      \end{array} \]

\item $C=u[z/D]$. 

\[ \begin{array}{ll}
\size(u[z/D\cwc{x}]) + \mul{x}{u[z/D\cwc{x}]} \cdot \size(v) & = \\ 
\size(u) + \size(D\cwc{x}) + \mul{z}{u} \cdot  \size(D\cwc{x}) + 
\mul{x}{u}\cdot \size(v) + & \\
\mul{x}{D\cwc{x}}\cdot \size(v) + \mul{z}{u} \cdot \mul{x}{D\cwc{x}}\cdot \size(v) & = \\
\size(u) + (\mul{z}{u} +1 )\cdot  (\size(D\cwc{x}) +
\mul{x}{D\cwc{x}}\cdot \size(v)) + 
\mul{x}{u}\cdot \size(v)   & >_{\ih}\\
\size(u) + (\mul{z}{u} +1 )\cdot  (\size(D\cwc{v}) +
\mul{x}{D\cwc{v}}\cdot \size(v)) + 
\mul{x}{u}\cdot \size(v)   & = \\
\size(u) + \size(D\cwc{v}) + \mul{z}{u} \cdot \size(D\cwc{v}) + 
\mul{x}{u} \cdot \size(v) + & \\
\mul{x}{D\cwc{v}} \cdot \size(v) + \mul{z}{u} \cdot \mul{x}{D\cwc{v}}\cdot \size(v)
 & = \\
\size(u[z/D\cwc{v}]) + \mul{x}{u[z/D\cwc{v}]} \cdot \size(v) 
\end{array} \]

\end{itemize}
\end{proof}

\begin{lemma} \label{l:for-decreasing1usm}
Let $v$ be a $\mterms$-metaterm
such that $x \notin \fv(v)$. Let $\phi = \{x\} \cup \fv(u)$
and $x \in \Del$. 
Then,\\
$\size(C\cwc{\mX_\Del}_\phi) + \mul{x}{C\cwc{\mX_\Del}_\phi} \cdot \size(v) > \size(C\cwc{\mX_\Del[x/v]}_\phi) +
\mul{x}{C\cwc{\mX_\Del[x/v]}_\phi} \cdot \size(v)$.
\end{lemma}

\begin{proof} By induction on $C$. 
\begin{itemize}
\item $C = \Box$. 

      \[ \begin{array}{ll}
         \size(\mX_\Del) + \mul{x}{\mX_\Del} \cdot \size(v) & = \\ 
           |\Del| + \size(v) & > \\
           |\Del| - 1 + \size(v) + 0 & = \\ 
           \size(\mX_\Del[x/v]) +\mul{x}{\mX_\Del[x/v]} \cdot \size(v) \\
           \end{array} \] 

\item $C = D\ t$. 

      \[ \begin{array}{ll}
       \size(D\cwc{\mX_\Del}_\phi\ t) + \mul{x}{D\cwc{\mX_\Del}_\phi\ t} \cdot \size(v) & = \\
       \size(D\cwc{\mX_\Del}_\phi) + \size(t) + (\mul{x}{D\cwc{\mX_\Del}_\phi} + \mul{x}{t})
       \cdot \size(v) & = \\
       \size(D\cwc{\mX_\Del}_\phi) + \size(t) + \mul{x}{D\cwc{\mX_\Del}_\phi}  \cdot \size(v) +
       \mul{x}{t} \cdot \size(v) & >_{\ih} \\
       \size(D\cwc{\mX_\Del[x/v]}_\phi) + \size(t) + \mul{x}{D\cwc{\mX_\Del[x/v]}_\phi}\cdot  \size(v) + \mul{x}{t} \cdot  \size(v) & = \\
      \size(D\cwc{\mX_\Del[x/v]}_\phi\ t)  + \mul{x}{D\cwc{\mX_\Del[x/v]}_\phi\ t} \cdot \size(v)
       \end{array} \]

\item $C= t\ D$ and $C= \l z. D$ are similar.

\item $C=D[z/u]$. W.l.o.g. we assume $z \notin \fv(v)$.

     \[ \begin{array}{ll}
      \size(D\cwc{\mX_\Del}_\phi[z/u]) + \mul{x}{D\cwc{\mX_\Del}_\phi[z/u]} \cdot \size(v) & = \\ 
      \size(D\cwc{\mX_\Del}_\phi) + \mul{z}{D\cwc{\mX_\Del}_\phi} \cdot \size(u) + m +  
      \mul{x}{D\cwc{\mX_\Del}_\phi}\cdot \size(v) \\
      + n + \mul{z}{D\cwc{\mX_\Del}_\phi} \cdot
      \mul{x}{u} \cdot \size(v) & >_{\ih} \\
      \size(D\cwc{\mX_\Del[x/v]}_\phi) +\mul{z}{D\cwc{\mX_\Del}_\phi} \cdot \size(u) +  m + 
      \mul{x}{D\cwc{\mX_\Del[x/v]}_\phi}\cdot \size(v) \\
      + n + \mul{z}{D\cwc{\mX_\Del}_\phi} \cdot
      \mul{x}{u} \cdot \size(v) & \geq_{\mbox{\cref{l:for-decreasing00}}}  \\
      \size(D\cwc{\mX_\Del[x/v]}_\phi) +\mul{z}{D\cwc{\mX_\Del[x/v]}_\phi} \cdot \size(u) +  m + 
      \mul{x}{D\cwc{\mX_\Del[x/v]}_\phi}\cdot \size(v) \\
      + n + \mul{z}{D\cwc{\mX_\Del[x/v]}_\phi} \cdot
      \mul{x}{u} \cdot \size(v) & = \\
      \size(D\cwc{\mX_\Del[x/v]}_\phi[z/u]) + \mul{x}{D\cwc{\mX_\Del[x/v]}_\phi[z/u]} \cdot \size(v)
      \end{array} \]
as $z \neq x$ where $m, n = -1, 0$ if $D\cwc{\mX_\Del}_\phi = \mX_{\Del'} [x_1/u_1]\ldots[x_n/u_n], y \in \Del'$ and $m = \size(u), n = \mul{x}{u}\cdot \size(v)$ otherwise.

\item $C=u[z/D]$. 

If $u = \mX_{\Del'} [x_1/u_1]\ldots[x_n/u_n], z \in \Del'$ then we have
\[ \begin{array}{ll}
\size(u[z/D\cwc{\mX_\Del}_\phi]) + \mul{x}{u[z/D\cwc{\mX_\Del}_\phi]} \cdot \size(v) & = \\ 
\size(u) -1 + \mul{z}{u} \cdot  \size(D\cwc{\mX_\Del}_\phi) + 
\mul{x}{u}\cdot \size(v)
 + \mul{z}{u} \cdot \mul{x}{D\cwc{\mX_\Del}_\phi}\cdot \size(v)
  & =\\
\size(u) -1 + \mul{z}{u} \cdot ( \size(D\cwc{\mX_\Del}_\phi) +  \mul{x}{D\cwc{\mX_\Del}_\phi}\cdot \size(v)) + 
\mul{x}{u}\cdot \size(v)
  & >_{\ih}\\
\size(u) -1 + \mul{z}{u} \cdot ( \size(D\cwc{\mX_\Del[x/v]}_\phi) +  \mul{x}{D\cwc{\mX_\Del[x/v]}_\phi}\cdot \size(v)) + 
\mul{x}{u}\cdot \size(v)
  & = \\
\size(u) -1 + \mul{z}{u} \cdot \size(D\cwc{\mX_\Del[x/v]}_\phi) + 
\mul{x}{u}\cdot \size(v) + \mul{z}{u} \cdot \mul{x}{D\cwc{\mX_\Del[x/v]}_\phi}\cdot \size(v)
  & = \\
\size(u[z/D\cwc{\mX_\Del[x/v]}_\phi]) + \mul{x}{u[z/D\cwc{\mX_\Del[x/v]}_\phi]} \cdot \size(v) 
\end{array} \] 

If $u \neq \mX_{\Del'} [x_1/u_1]\ldots[x_n/u_n], z \in \Del'$ then we have
\[ \begin{array}{ll}
\size(u[z/D\cwc{\mX_\Del}_\phi]) + \mul{x}{u[z/D\cwc{\mX_\Del}_\phi]} \cdot \size(v) & = \\ 
\size(u) + (\mul{z}{u} + 1) \cdot ( \size(D\cwc{\mX_\Del}_\phi) + \mul{x}{D\cwc{\mX_\Del}_\phi}\cdot \size(v))  + 
\mul{x}{u}\cdot \size(v)
  & >_{\ih}\\
\size(u) + (\mul{z}{u} + 1) \cdot ( \size(D\cwc{\mX_\Del[x/v]}_\phi) + \mul{x}{D\cwc{\mX_\Del[x/v]}_\phi}\cdot \size(v))  + 
\mul{x}{u}\cdot \size(v)
  & =\\
\size(u[z/D\cwc{\mX_\Del[x/v]}_\phi]) + \mul{x}{u[z/D\cwc{\mX_\Del[x/v]}_\phi]} \cdot \size(v) 
\end{array} \] 

\end{itemize}
\end{proof}

\begin{lemma}
\label{l:sub-terminating-metaterms}
The reduction relation  $\Rew{\m}$ is terminating on $\mterms$-metaterms.
\end{lemma}

\begin{proof} 
We first show that  $t \Rew{\m} t'$ implies $\mul{x}{t} \geq  \mul{x}{t'}$
for every variable $x$. This can be done  by induction on $t \Rew{\m}
t'$. If the reduction 
step is internal, then the \ih\ allows
to conclude. Otherwise, 
 
\begin{itemize}
\item $u[y/v]\Rew{\Gc} u$.       We have 

     $\mul{x}{u} + \mul{x}{v} + \mul{y}{u} \cdot \mul{x}{v} \geq 
      \mul{x}{u}$.

\item $C\cwc{y}[y/v]\Rew{\us} C\cwc{v}[y/v]$. Thus, $y \notin \fv(v)$ and
 $x \neq y$. We have

      \[ \begin{array}{l}
       \mul{x}{C\cwc{y}[y/v]} = \\
       \mul{x}{C\cwc{y}} + \mul{x}{v} + \mul{y}{C\cwc{y}} \cdot  \mul{x}{v} \geq_{\mbox{\cref{l:for-decreasing0}}} \\
       \mul{x}{C\cwc{v}} + \mul{x}{v} + \mul{y}{C\cwc{v}} \cdot  \mul{x}{v} = \\
       \mul{x}{C\cwc{v}[y/v]}\\
       \end{array} \]

\item $C\cwc{\mX_\Del}_{\phi} [y/v]\Rew{\usm} C\cwc{\mX_\Del[y/v]}_{\phi}[y/v]$ with $\phi = \{y\} \cup \fv(v)$. Thus, $y \in \Del$, $C \neq \Box[y_1/v_1]\ldots[y_n/v_n], n \geq 0$, $y \notin \fv(v)$ and
 $x \neq y$. We have

      \[ \begin{array}{l}
       \mul{x}{C\cwc{\mX_\Del}_\phi[y/v]} = \\
       \mul{x}{C\cwc{\mX_\Del}_\phi} + \mul{x}{v} + \mul{y}{C\cwc{\mX_\Del}_\phi} \cdot  \mul{x}{v}
       \geq_{\mbox{\cref{l:for-decreasing0usm}}} \\
       \mul{x}{C\cwc{\mX_\Del[y/v]}_\phi} + \mul{x}{v} + \mul{y}{C\cwc{\mX_\Del[y/v]}_\phi} \cdot  \mul{x}{v} = \\
       \mul{x}{C\cwc{\mX_\Del[y/v]}_\phi[y/v]}\\
       \end{array} \]
\end{itemize}

To show that $t \Rew{\m} t'$ implies $\size(t) > \size(t')$
we also  reason by induction on $t \Rew{\m} t'$. If the reduction
step is internal, then the previous property and the \ih\ allow
to conclude. Otherwise

\begin{itemize}
\item $u[x/v]\Rew{\Gc} u$.

      We have $\size(u) + \size(v) + \mul{x}{u} \cdot \size(v) >
      \size(u)$.

\item $C\cwc{x}[x/v]\Rew{\us} C\cwc{v}[x/v]$. We have

      \[ \begin{array}{l}
       \size(C\cwc{x}[x/v]) = \\
       \size(C\cwc{x}) + \size(v) + \mul{x}{C\cwc{x}} \cdot \size(v) >_{\mbox{\cref{l:for-decreasing1}}} \\
       \size(C\cwc{v}) + \size(v) + \mul{x}{C\cwc{v}} \cdot \size(v)\\
       \size(C\cwc{v}[x/v])\\
       \end{array} \]

\item $C\cwc{\mX_\Del}_{\phi} [x/v]\Rew{\usm} C\cwc{\mX_\Del[x/v]}_{\phi}[x/v]$ with $\phi = \{x\} \cup \fv(v)$. Thus, $x \in \Del$, $C \neq \Box[y_1/v_1]\ldots[y_n/v_n], n \geq 0$, $x \notin \fv(v)$. We have

      \[ \begin{array}{l}
       \size(C\cwc{\mX_\Del}_\phi[x/v]) = \\
       \size(C\cwc{\mX_\Del}_\phi) + \size(v) + \mul{x}{C\cwc{x}_\phi} \cdot \size(v) >_{\mbox{\cref{l:for-decreasing1usm}}} \\
       \size(C\cwc{\mX_\Del[x/v]}_\phi) + \size(v) + \mul{x}{C\cwc{\mX_\Del[x/v]}_\phi} \cdot \size(v)\\
       \size(C\cwc{\mX_\Del[x/v]}_\phi[x/v])\\
       \end{array} \]
\end{itemize}
\end{proof}

\begin{lemma} 
\label{l:lm-local-confluent-coherent}
The reduction relation $\Rew{\m}$  is
  locally confluent and locally coherent on metaterms.
\end{lemma}

\begin{proof} We check all the $\m$-critical pairs. 
  \begin{itemize}
  \item All the critical pairs coming from   two reduction rules can
    be closed  trivially. Indeed, if the critical pairs arises from
   reductions in parallel positions, then this case is trivial.
   If a substitution $[x/u]$ is propagated via $\us$ 
   on two different occurrences
   $x$ inside a term $t$, then the order of application of $\us$
   being irrelevant this case is  also trivial. The same for
   the  propagation of $[x/u]$ via $\usm$ w.r.t two metavariables.
   Remark that $\us$ and $\usm$ also commute. 
   The more delicate case is a critical pair between ($\us$ or
   $\usm$) and $\Gc$. This pair can also be closed as the
   following example shows:

   \[ y[x/u] \LRew{\Gc } y[z/x][x/u] \Rew{\us} y[z/u][x/u] \]  

   is closed by

   \[ y[x/u]  \LRew{\Gc} y[z/u][x/u] \]  
  \item The critical pairs between $\Gc$ and $\C$ can be closed
    as in the forthcoming~\cref{l:ALC-lc-lc}.

  \item The critical pairs between $\us$ and $\C$ can
    be closed as follows:

    \[ \begin{array}{clc}
       C\cwc{x}[x/u][y/v] & =_{\C} & C\cwc{x}[y/v][x/u]\\
       \downarrow_{\us}   &        &  \downarrow_{\us}  \\
       C\cwc{u}[x/u][y/v] & =_{\C} & C\cwc{u}[y/v][x/u]\\
       \end{array} \]

 \item The critical pairs between $\usm$ and $\C$ can
    be closed as follows:

\[ \begin{array}{clc}
       C\cwc{\mX_{\Del}}[x/u][y/v] & =_{\C} & C\cwc{\mX_{\Del}}[y/v][x/u]\\
       \downarrow_{\usm}   &        &  \downarrow_{\usm}  \\
       C\cwc{\mX_{\Del}[x/u]}[x/u][y/v] & =_{\C} & C\cwc{\mX_{\Del}[x/u]}[y/v][x/u]\\
       \end{array} \]

    Remark that $C$ being a good context, the same
    happens with $C[y/v]$ so that reduction $\usm$ 
    from the term $C\cwc{x}[y/v][x/u]$ is allowed.

  \end{itemize}
    
\end{proof}

\begin{lemma}
\label{l:snf-stable-substitution}
Let $t$ and $u$ be $\m$-normal forms. Then $t \isubs{x}{u}$ is a $\m$-normal form.
\end{lemma}

\begin{proof} The proof is by induction on $t$ using~\cref{l:forma-snf}.
\begin{itemize}
\item $t = y$. If $y = x$, $t \isubs{x}{u} = u$. If $y \neq x$,
  $t \isubs{x}{u} = y$
which is a $\m$-normal form. 
\item $t = t_1\ t_2$. $t \isubs{x}{u} = (t_1\ t_2) \isubs{x}{u} =
  t_1   \isubs{x}{u}\  t_2   \isubs{x}{u}$.   By  the   induction
  hypothesis,  $t_1 \isubs{x}{u}$ and  $t_2 \isubs{x}{u}$  are in
  $\m$-normal form.
\item $t  = \lambda y. t_1$.  $t \isubs{x}{u} =  (\lambda y. t_1)
  \isubs{x}{u} =  \lambda y. t_1 \isubs{x}{u}$.  By the induction
  hypothesis, $t_1 \isubs{x}{u}$ is in $\m$-normal form.
\item $t = \mX_\Del  [x_1/u_1]\ldots[x_n/u_n]$ with every $u_i$ a
  $\m$-normal  form. By  the \ih\   every $u_i\isubs{x}{u}$  is an
  $\m$-normal form and by $\alpha$-conversion we can suppose that
  $x_i   \notin   \fv(u)$.   Thus   $t\isubs{x}{u}   =   \mX_\Del
  \isubs{x}{u}   [x_1/u_1\isubs{x}{u}]\ldots[x_n/u_n\isubs{x}{u}]$
  is an $\m$-normal form by~\cref{l:forma-snf}.
\end{itemize}
\end{proof}



  \begin{lemma}
\label{l:snf-explicit-substitution}
Let $t$ and $u$ be metaterms.
Then $\m(t[x/u]) = \m(t)\isubs{x}{\m(u)}$.
\end{lemma}

\begin{proof}
The proof is by induction on $t$.
\begin{itemize}
\item $t = z$. 
 If $z = x$ then $\m(t[x/u]) = \m(u) = x\isubs{x}{\m(u)} = t\isubs{x}{\m(u)}$. If $z \neq x$ then $\m(t[x/u]) = \m(z) = z\isubs{x}{\m(u)} = t\isubs{x}{\m(u)}$. 
\item $t = \l z.t_1$. 
 Then \[ \begin{array}{llll } 
   \m(t[x/u]) & = & \\
   \m((\l z.t_1)[x/u])& =_{\mbox{\cref{l:es-application-lambda-composition}}} \\
   \m(\l z.t_1[x/u]) & =_{\mbox{\cref{l:s-distributes}}} \\
   \l z.\m(t_1[x/u]) &  =_{\ih}\\
   \l z. \m(t_1)\isubs{x}{\m(u)} & = \\
   (\l z. \m(t_1))\isubs{x}{\m(u)} & =_{\mbox{\cref{l:s-distributes}}} \\
   \m(\l z. t_1)\isubs{x}{\m(u)}&  = \m(t)\isubs{x}{\m(u)}\\
  \end{array} \]  
\item $t = t_1\ t_2$. 
 Then

\[ \begin{array}{llll } 
  \m(t[x/u]) & = \\
  \m((t_1\ t_2)[x/u]) & =_{\mbox{\cref{l:es-application-lambda-composition},~\cref{l:s-distributes}}} \\
  \m(t_1[x/u]) \m(t_2[x/u]) & =_{\ih} \\
   \m(t_1)\isubs{x}{\m(u)}
  \m(t_2)\isubs{x}{\m(u)} & = \\
  (\m(t_1) \m(t_2))\isubs{x}{\m(u)} &
  =_{\mbox{\cref{l:s-distributes}}} \\
  \m(t_1\ t_2)\isubs{x}{\m(u)} & =  
  \m(t)\isubs{x}{\m(u)} \\
  \end{array} \] 
\item $t = \mX_\Del$. 
 Then $\m(t[x/u]) = \m(\mX_\Del[x/u]) = \mX_\Del[x/\m(u)] = \mX_\Del\isubs{x}{\m(u)} = \m(t)\isubs{x}{\m(u)}$. 
\item $t = t_1[y/v]$. 
\begin{itemize}
\item $t_1 = z$. 
 If $z = y$ then $\m(t[x/u]) = \m(v[x/u]) =_{\ih} \m(v)\isubs{x}{\m(u)} = \m(y[y/v])\isubs{x}{\m(u)} = \m(t)\isubs{x}{\m(u)}$. 
 If $z \neq y$ then $\m(t[x/u]) = \m(z[x/u]) =_{\ih} \m(z)\isubs{x}{\m(u)} = \m(z[y/v])\isubs{x}{\m(u)} = \m(t)\isubs{x}{\m(u)}$. 
\item $t_1 = \l z.t_2$. 
Then 

\[ \begin{array}{llll}
  \m(t[x/u]) & = \\
  \m((\l z.t_2)[y/v][x/u]) &
=_{\mbox{\cref{l:es-application-lambda-composition}}}\\
\m(\l z.t_2[y/v][x/u]) & =_{\mbox{\cref{l:s-distributes}}} \\
\l
 z.\m(t_2[y/v][x/u]) & =_{\ih} \\
 \l z. \m(t_2[y/v])\isubs{x}{\m(u)} & = \\
 (\l
z. \m(t_2[y/v]))\isubs{x}{\m(u)}&  =_{\mbox{\cref{l:s-distributes}}} \\
\m(\l z. t_2[y/v])\isubs{x}{\m(u)} &
=_{\mbox{\cref{l:es-application-lambda-composition}}} \\
 \m((\l
z. t_2)[y/v])\isubs{x}{\m(u)} & = 
\m(t)\isubs{x}{\m(u)}
\end{array} \]  
\item $t_1 = t_2\ t_3$. 
 Then \[ \begin{array}{llll}
   \m(t[x/u]) & = \\
   \m((t_2\ t_3)[y/v][x/u]) & =_{\mbox{\cref{l:es-application-lambda-composition},~\cref{l:s-distributes}}} \\
   \m(t_2[y/v][x/u]) \m(t_3[y/v][x/u]) & =_{\ih}  \\
   \m(t_2[y/v])\isubs{x}{\m(u)} \m(t_3[y/v])\isubs{x}{\m(u)} & = \\
   (\m(t_2[y/v]) \m(t_3[y/v]))\isubs{x}{\m(u)} &  =_{\mbox{\cref{l:es-application-lambda-composition},~\cref{l:s-distributes}}} \\
   \m((t_2\ t_3)[y/v])\isubs{x}{\m(u)}&  =
   \m(t)\isubs{x}{\m(u)}
 \end{array} \]  

\item $t_1 = \mX_\Del$. Then
  \[ \begin{array}{llll}
    \m(t[x/u]) & = \\
    \m(\mX_\Del[y/v][x/u])  & =_{\mbox{\cref{l:es-application-lambda-composition}}} \\
    \m(\mX_\Del[x/u][y/v[x/u]])  & = \\
    \m(\mX_\Del[x/u])[y/\m(v[x/u])]  & = \\
    \mX_\Del\isubs{x}{\m(u)}[y/\m(v[x/u])]  & =_{\ih} \\
    \mX_\Del\isubs{x}{\m(u)}[y/\m(v)\isubs{x}{\m(u)}])  & = \\
    (\mX_\Del[y/\m(v)])\isubs{x}{\m(u)}  & = \\
    \m(\mX_\Del[y/v])\isubs{x}{\m(u)}  & = \m(t)\isubs{x}{\m(u)}
  \end{array} \]
%
\item $t_1 = t_2[y_1/v_1]$. Let $y=y_2, v=v_2$ for
  convenience. Then 

\[ \begin{array}{ll}
\m(t[x/u]) & = \\ 
\m(t_2[y_1/v_1][y_2/v_2][x/u]) & =_{\mbox{\cref{l:es-application-lambda-composition}}}\\
 \m(t_2[y_1/v_1][x/u][y_2/v_2[x/u]])&  =_{\ih} \\
 \m(\m(t_2[y_1/v_1])\isubs{x}{\m(u)}[y_2/\m(v_2)\isubs{x}{\m(u)}]) & = \\
\m(\m(t_2[y_1/v_1])[y_2/\m(v_2)]\isubs{x}{\m(u)}) & = \\
\m(t_2[y_1/v_1][y_2/v_2])\isubs{x}{\m(u)} & = \\
\m(t)\isubs{x}{\m(u)}
\end{array}\]
\end{itemize}
\end{itemize}
\end{proof}

\section{Proofs of~\cref{s:PSN}}
\label{app:relation-les}

\alclocallyconfluent*

\begin{proof}
To check local confluence of $\Rew{\ALC}$ we have to check that

\[ t_1 \LRew{\ALC}  t \Rew{A} t_2 \mbox{ implies }
  t_1 \Rewn{\ALC}  t_3 =_{\Es} t_4 \LRewn{\ALC} t_2 \]

where  $A = \{ \Appi, \Appii, \Appiii, \Lamb, \Compi, \Compii \}$

To check local coherence  of $\Rew{\ALC}$ we have to check that

\[ t_1 \LRew{\ALC}  t =_{\Es}  t_2 
\mbox{ implies } 
t_1 \Rewn{\ALC}  t_3 =_{\Es} t_4 \LRewn{\ALC} t_2 \] 

By~\cite{JK86} it is sufficient to show that for every critical
pair $\langle c_1, c_2 \rangle$ there exist two $\ALC$-normal forms $c_1!$ and $c_2!$ of $c_1$ and $c_2$ respectively such that $c_1! =_{\Es} c_2!$. We thus check this property, remarking that in some cases we find a common reduct
before arriving at an $\ALC$-normal form so that in particular
the desired property holds.

\begin{itemize}
\item $\Compii$ and $\Appi$:

We have $y \in \fv(u)\ \&\ y \notin \fv(t_1\ t_2)\ \&\ x \in
\fv(t_1)\ \&\ x \in \fv(t_2)$ so that 

$(t_1\ t_2)[x/u][y/v] \Rew{\Compii} (t_1\ t_2)[x/u[y/v]]$
        
$(t_1\ t_2)[x/u][y/v] \Rew{\Appi} (t_1[x/u]\ t_2[x/u])[y/v] $

can be closed by

$(t_1\ t_2)[x/u[y/v]] \Rew{\Appi} t_1[x/u[y/v]] \ t_2[x/u[y/v]]$

$(t_1[x/u]\ t_2[x/u])[y/v] \Rew{\Appi}
t_1[x/u][y/v]\ t_2[x/u][y/v] \Rewn{\Compii}t_1[x/u[y/v]] \
t_2[x/u[y/v]]$

\item $\Compii$ and $\Appii$: 

We have $y \in \fv(u)\ \&\ y \notin \fv(t_1\ t_2)\ \&\ x \notin
\fv(t_1)\ \&\ x \in \fv(t_2)$ so that

$(t_1\ t_2)[x/u][y/v] \Rew{\Compii} (t_1\ t_2)[x/u[y/v]]$
        
$(t_1\ t_2)[x/u][y/v] \Rew{\Appii} (t_1\ t_2[x/u])[y/v] $

can be closed by

$(t_1\ t_2)[x/u[y/v]] \Rew{\Appii} t_1\ t_2[x/u[y/v]]$

$(t_1\ t_2[x/u])[y/v] \Rew{\Appii}
t_1\ t_2[x/u][y/v] \Rewn{\Compii} t_1\ t_2[x/u[y/v]]$

\item $\Compii$ and $\Appiii$: Similar.

\item $\Compii$ and $\Lamb$: 

We have $y \in \fv(u)\ \&\ y \notin \fv(\l z. t)$ so that

$(\l z. t)[x/u][y/v] \Rew{\Compii} (\l z. t)[x/u[y/v]]$ 

$(\l z. t)[x/u][y/v] \Rew{\Lamb} (\l z. t[x/u])[y/v]$ 

can be closed by

$(\l z. t)[x/u[y/v]] \Rew{\Lamb} \l z. t[x/u[y/v]]$

$(\l z. t[x/u])[y/v] \Rew{\Lamb}\l z. t[x/u][y/v] \Rew{\Compii} \l
z. t[x/u[y/v]]$

\item $\Compi$ and $\Appi$:

We have $y \in \fv(u)\ \&\ y \in \fv(t_1\ t_2)\ \&\ x \in
\fv(t_1)\ \&\ x \in \fv(t_2)$ so that

$(t_1\ t_2)[x/u][y/v] \Rew{\Compi} (t_1\ t_2)[y/v][x/u[y/v]]$

$(t_1\ t_2)[x/u][y/v] \Rew{\Appi} (t_1[x/u]\ t_2[x/u])[y/v]$
 
If $y \in \fv(t_1)\ \&\ y \in \fv(t_2)$, then we close by

$(t_1\ t_2)[y/v][x/u[y/v]] \Rew{\Appi} 
 (t_1[y/v]\ t_2[y/v])[x/u[y/v]] \Rew{\Appi} 
t_1[y/v][x/u[y/v]]\ t_2[y/v][x/u[y/v]]$

$(t_1[x/u]\ t_2[x/u])[y/v] \Rew{\Appi} 
t_1[x/u][y/v]\ t_2[x/u][y/v]\Rewn{\Compi} 
t_1[y/v][x/u[y/v]]\ t_2[y/v][x/u[y/v]]$

If $y \notin \fv(t_1)\ \&\ y \in \fv(t_2)$, then we close by

$(t_1\ t_2)[y/v][x/u[y/v]] \Rew{\Appii} 
 (t_1\ t_2[y/v])[x/u[y/v]] \Rew{\Appi} 
t_1[x/u[y/v]]\ t_2[y/v][x/u[y/v]]$

$(t_1[x/u]\ t_2[x/u])[y/v] \Rew{\Appi} 
(t_1[x/u][y/v]\ t_2[x/u][y/v]\Rew{\Compi} 
t_1[x/u][y/v]\ t_2[y/v][x/u[y/v]] \\ \Rew{\Compii}
t_1[x/u[y/v]]\ t_2[y/v][x/u[y/v]]$

If $y \in \fv(t_1)\ \&\ y \notin \fv(t_2)$, then is similar to the
previous case. 

Remark that the case  $y \notin \fv(t_1)\ \&\ y \notin \fv(t_2)$
is not possible.

\item $\Compi$ and $\Appii$

We have $y \in \fv(u)\ \&\ y \in \fv(t_1\ t_2)\ \&\ x \notin
\fv(t_1)\ \&\ x \in \fv(t_2)$ so that

$(t_1\ t_2)[x/u][y/v] \Rew{\Compi} (t_1\ t_2)[y/v][x/u[y/v]]$

$(t_1\ t_2)[x/u][y/v] \Rew{\Appii} (t_1\ t_2[x/u])[y/v]$

If $y \in \fv(t_1)\ \&\ y \in \fv(t_2)$, then we close by

$(t_1\ t_2)[y/v][x/u[y/v]] \Rew{\Appi} 
 (t_1[y/v]\ t_2[y/v])[x/u[y/v]] \Rew{\Appii} 
t_1[y/v]\ t_2[y/v][x/u[y/v]]$

$(t_1\ t_2[x/u])[y/v] \Rew{\Appi} 
t_1[y/v]\ t_2[x/u][y/v]\Rewn{\Compi} 
t_1[y/v]\ t_2[y/v][x/u[y/v]]$

If $y \notin \fv(t_1)\ \&\ y \in \fv(t_2)$, then we close by

$(t_1\ t_2)[y/v][x/u[y/v]] \Rew{\Appii} 
 (t_1\ t_2[y/v])[x/u[y/v]] \Rew{\Appii} 
 t_1\ t_2[y/v][x/u[y/v]]$

$(t_1\ t_2[x/u])[y/v] \Rew{\Appii} 
 t_1\ t_2[x/u][y/v]\Rew{\Compi} 
t_1\ t_2[y/v][x/u[y/v]]$

If $y \in \fv(t_1)\ \&\ y \notin \fv(t_2)$, then is similar to the
previous case. 

Remark that the case  $y \notin \fv(t_1)\ \&\ y \notin \fv(t_2)$
is not possible.

\item $\Compi$ and $\Appiii$: this case is similar to the
  previous one.

\item $\Compi$ and $\Lamb$:

We have $y \in \fv(u)\ \&\ y \in \fv(\l z.t)$ so that

$(\l z. t)[x/u][y/v] \Rew{\Compi} (\l z. t)[y/v][x/u[y/v]]$ 

$(\l z. t)[x/u][y/v] \Rew{\Lamb} (\l z. t[x/u])[y/v]$ 

can be closed by

$(\l z. t)[y/v][x/u[y/v]] \Rewn{\Lamb} \l z. t[y/v][x/u[y/v]]$

$(\l z. t[x/u])[y/v] \Rew{\Lamb} \l z. t[x/u][y/v]\Rew{\Compi} \l z. t[y/v][x/u[y/v]]$

\item $\Compii$ and $\Compii$:

We have $z \in \fv(v)\ \&\ z \notin \fv(t[x/u])\ \&\ y \in
\fv(u)\ \&\ y \notin \fv(t)$ so that

$t[x/u][y/v][z/w] \Rew{\Compii} t[x/u][y/v[z/w]]$

$t[x/u][y/v][z/w] \Rew{\Compii} t[x/u[y/v]][z/w]$

can be closed by

$t[x/u][y/v[z/w]] \Rew{\Compii} t[x/u[y/v[z/w]]]$

$t[x/u[y/v]][z/w] \Rew{\Compii} t[x/u[y/v][z/w]] \Rew{\Compii}  t[x/u[y/v[z/w]]]$

\item $\Compi$ and $\Compii$:

We have $z \in \fv(t[x/u])\ \&\ z \in \fv(v)\ \&\ 
y \in \fv(u)\ \&\ y \notin \fv(t)$ so that 

$t[x/u][y/v][z/w] \Rew{\Compi} t[x/u][z/w][y/v[z/w]]$

$t[x/u][y/v][z/w] \Rew{\Compii} t[x/u[y/v]][z/w]$

If $z \in \fv(u)\ \&\ z \in \fv(t)$, we close by

$t[x/u][z/w][y/v[z/w]]  \Rew{\Compi} t[z/w][x/u[z/w]][y/v[z/w]]$

$t[x/u[y/v]][z/w] \Rew{\Compi} t[z/w][x/u[y/v][z/w]]
\Rew{\Compi}  t[z/w][x/u[z/w][y/v[z/w]]]$

If $z \in \fv(u)\ \&\ z \notin \fv(t)$, we close by

$t[x/u][z/w][y/v[z/w]]  \Rew{\Compii} t[x/u[z/w]][y/v[z/w]]$

$t[x/u[y/v]][z/w] \Rew{\Compii} t[x/u[y/v][z/w]]
\Rew{\Compi}  t[x/u[z/w][y/v[z/w]]]$

If $z \notin \fv(u)\ \&\ z \in \fv(t)$, we close by

$t[x/u][z/w][y/v[z/w]]  =_{\C}  t[z/w][x/u][y/v[z/w]]
\Rew{\Compii} t[z/w][x/u[y/v[z/w]]] $

$t[x/u[y/v]][z/w] \Rew{\Compi} t[z/w] [x/u[y/v][z/w] ]
\Rew{\Compii}  t[z/w] [x/u[y/v[z/w]]]$

\item $\Compii$ and $\Compi$:

We  have $z \in \fv(v)\ \&\ z \notin \fv(t[x/u])\ \&\
y \in \fv(u)\ \&\ y \in \fv(t)$ so that 

$t[x/u][y/v][z/w] \Rew{\Compii} t[x/u][y/v[z/w]]$

$t[x/u][y/v][z/w] \Rew{\Compi} t[y/v][x/u[y/v]][z/w]$

can be closed by

$t[x/u][y/v[z/w]] \Rew{\Compi} t[y/v[z/w]][x/u[y/v[z/w]]]$

$t[y/v][x/u[y/v]][z/w] \Rew{\Compi} t[y/v][z/w][x/u[y/v][z/w]]
\Rew{\Compii} t[y/v[z/w]][x/u[y/v][z/w]] \\
\Rew{\Compii} t[y/v[z/w]][x/u[y/v[z/w]]]$

\item $\Compi$ and $\Compi$

We  have $z \in \fv(v)\ \&\ z \in \fv(t[x/u])\ \&\
y \in \fv(u)\ \&\ y \in \fv(t)$ so that 

$t[x/u][y/v][z/w] \Rew{\Compi} t[x/u][z/w][y/v[z/w]]$

$t[x/u][y/v][z/w] \Rew{\Compi} t[y/v][x/u[y/v]][z/w]$

If $z \in \fv(u)\ \&\ z \notin \fv(t)$, then we close by

$t[x/u][z/w][y/v[z/w]] \Rew{\Compii} 
 t[x/u[z/w]][y/v[z/w]]\Rew{\Compi}
t[y/v[z/w]][x/u[z/w][y/v[z/w]]]$

$t[y/v][x/u[y/v]][z/w] \Rew{\Compi}
t[y/v][z/w][x/u[y/v][z/w]] \Rew{\Compi}
t[y/v][z/w][x/u[z/w][y/v[z/w]]]  \\ \Rew{\Compii}
t[y/v[z/w]][x/u[z/w][y/v[z/w]]]$

If $z \notin \fv(u)\ \&\ z \in \fv(t)$, then we close by

$t[x/u][z/w][y/v[z/w]] =_{\C} 
 t[z/w][x/u][y/v[z/w]] \Rew{\Compi}
t[z/w][y/v[z/w]][x/u[y/v[z/w]]]$

$t[y/v][x/u[y/v]][z/w] \Rew{\Compi}
t[y/v][z/w][x/u[y/v][z/w]]\Rew{\Compi}
t[z/w][y/v[z/w]][x/u[y/v][z/w]] \\ \Rew{\Compii}
t[z/w][y/v[z/w]][x/u[y/v[z/w]]]$

\item $\Appi$ and $\C$:

We have $y \notin \fv(u)\ \&\ x \notin \fv(v)\ \&\ x \in
\fv(t_1)\ \&\ x \in \fv(t_2)$ so that

$(t_1\ t_2)[x/u][y/v] =_{\C} (t_1\ t_2)[y/v] [x/u]$

$(t_1\ t_2)[x/u][y/v] \Rew{\Appi} (t_1[x/u]\ t_2[x/u])[y/v]$

If $y \in \fv(t_1)\ \&\ y \in \fv(t_2)$, we close by 

$(t_1\ t_2)[y/v] [x/u] \Rew{\Appi} (t_1[y/v] \ t_2[y/v] )[x/u]
 \Rew{\Appi} t_1[y/v][x/u] \ t_2[y/v][x/u]$

$(t_1[x/u]\ t_2[x/u])[y/v]\Rew{\Appi}
t_1[x/u][y/v]\ t_2[x/u][y/v] =_{\C} t_1[y/v][x/u]\ t_2[y/v][x/u]$

If $y \notin \fv(t_1)\ \&\ y \in \fv(t_2)$, we close by 

$(t_1\ t_2)[y/v] [x/u] \Rew{\Appii} (t_1 \ t_2[y/v] )[x/u]
 \Rew{\Appi} t_1[x/u] \ t_2[y/v][x/u]$

$(t_1[x/u]\ t_2[x/u])[y/v]\Rew{\Appii}
t_1[x/u]\ t_2[x/u][y/v] =_{\C} t_1[x/u]\ t_2[y/v][x/u]$

The case  $y \in \fv(t_1)\ \&\ y \notin \fv(t_2)$ is similar to
the previous one. 

If $y \notin \fv(t_1)\ \&\ y \notin \fv(t_2)$, we close by 

$(t_1\ t_2)[y/v] [x/u] \Rew{\Gc} (t_1 \ t_2)[x/u]
 \Rew{\Appi} t_1[x/u] \ t_2[x/u]$

$(t_1[x/u]\ t_2[x/u])[y/v]\Rew{\Gc}
t_1[x/u]\ t_2[x/u]$

\item $\Appii$ and $\C$:

We have $y \notin \fv(u)\ \&\ x \notin \fv(v)\ \&\ x \notin
\fv(t_1)\ \&\ x \in \fv(t_2)$ so that

$(t_1\ t_2)[x/u][y/v] =_{\C} (t_1\ t_2)[y/v] [x/u]$

$(t_1\ t_2)[x/u][y/v] \Rew{\Appii} (t_1\ t_2[x/u])[y/v]$

If $y \in \fv(t_1)\ \&\ y \in \fv(t_2)$, we close by 

$(t_1\ t_2)[y/v] [x/u] \Rew{\Appi} (t_1[y/v] \ t_2[y/v] )[x/u]
 \Rew{\Appii} t_1[y/v]\ t_2[y/v][x/u]$

$(t_1\ t_2[x/u])[y/v]\Rew{\Appi}
t_1[y/v]\ t_2[x/u][y/v] =_{\C} t_1[y/v]\ t_2[y/v][x/u]$

If $y \notin \fv(t_1)\ \&\ y \in \fv(t_2)$, we close by 

$(t_1\ t_2)[y/v] [x/u] \Rew{\Appii} (t_1 \ t_2[y/v] )[x/u]
 \Rew{\Appii} t_1\ t_2[y/v][x/u]$

$(t_1\ t_2[x/u])[y/v]\Rew{\Appii}
t_1\ t_2[x/u][y/v] =_{\C} t_1\ t_2[y/v][x/u]$

The case  $y \in \fv(t_1)\ \&\ y \notin \fv(t_2)$ is similar to
the previous one. 

If $y \notin \fv(t_1)\ \&\ y \notin \fv(t_2)$, we close by 

$(t_1\ t_2)[y/v] [x/u] \Rew{\Gc} (t_1 \ t_2)[x/u]
 \Rew{\Appii} t_1\ t_2[x/u]$

$(t_1\ t_2[x/u])[y/v]\Rew{\Gc}
t_1\ t_2[x/u]$

\item $\Appiii$ and $\C$: similar to the previous one. 

\item $\Lamb$ and $\C$:

We have $y \notin \fv(u)\ \&\ x \notin \fv(v)$ so that 

$(\l z. t)[x/u][y/v] =_{\C} (\l z. t)[y/v] [x/u]$

$(\l z. t)[x/u][y/v] \Rew{\Lamb} (\l z. t[x/u])[y/v]$

can be closed by 

$(\l z. t)[y/v] [x/u] \Rewn{\Lamb}\l z. t[y/v] [x/u]$

$(\l z. t[x/u])[y/v] \Rew{\Lamb} \l z. t[x/u][y/v] =_{\C} \l z. t[y/v][x/u]$

\item $\Compii$ and $\C$:

We have $z \notin \fv(v)\ \&\ y \notin \fv(w)\ \&\ y \in
\fv(u)\ \&\ y \notin \fv(t)$ so that 

$t[x/u][y/v][z/w] =_{\C} t[x/u][z/w][y/v]$

$t[x/u][y/v][z/w] \Rew{\Compii} t[x/u[y/v]][z/w]$

If $z \in \fv(t)\ \&\  z \in \fv(u)$, then

$t[x/u][z/w][y/v] \Rew{\Compi}
t[z/w][x/u[z/w]][y/v] \Rew{\Compii}
t[z/w][x/u[z/w][y/v]]$

$t[x/u[y/v]][z/w] \Rew{\Compi} 
t[z/w] [x/u[y/v][z/w] ] =_{\C} t[z/w] [x/u[z/w][y/v]]$

If $z \notin \fv(t)\ \&\ z \in \fv(u)$, then

$t[x/u][z/w][y/v] \Rew{\Compii}
t[x/u[z/w]][y/v] \Rew{\Compii}
t[x/u[z/w][y/v]]$

$t[x/u[y/v]][z/w] \Rew{\Compi} 
t[x/u[y/v][z/w] ] =_{\C} t[x/u[z/w][y/v]]$

If $z \notin \fv(u)$, then

$t[x/u][z/w][y/v] =_{\C} t[z/w][x/u][y/v] \Rew{\Compii} t[z/w][x/u[y/v]]$

$t[x/u[y/v]][z/w] =_{\C} t[z/w][x/u[y/v]]$

\item $\Compi$ and $\C$:

We have $z \notin \fv(v)\ \&\ y \notin \fv(w)\ \&\ y \in
\fv(u)\ \&\ y \in \fv(t)$ so that

$t[x/u][y/v][z/w] =_{\C} t[x/u][z/w][y/v]$

$t[x/u][y/v][z/w] \Rew{\Compi} t[y/v][x/u[y/v]][z/w]$

If $z \in \fv(t)\ \&\ z \in \fv(u)$, then 

$t[x/u][z/w][y/v] \Rew{\Compi} t[z/w][x/u[z/w]][y/v] \Rew{\Compi}
t[z/w][y/v][x/u[z/w][y/v]]$

$t[y/v][x/u[y/v]][z/w]\Rew{\Compi} t[y/v][z/w][x/u[y/v][z/w]] =_{\C}
t[z/w][y/v][x/u[z/w][y/v]]$

If $z \notin \fv(t)\ \&\ z \in \fv(u)$, then

$t[x/u][z/w][y/v] \Rew{\Compii} t[x/u[z/w]][y/v] \Rew{\Compi}
t[y/v][x/u[z/w][y/v]]$

$t[y/v][x/u[y/v]][z/w]\Rew{\Compii} t[y/v][x/u[y/v][z/w]] =_{\C}
t[y/v][x/u[z/w][y/v]]$

If $z \notin \fv(u)$, then 

$t[x/u][z/w][y/v] =_{\C} t[z/w][x/u][y/v] \Rew{\Compi} t[z/w][y/v] [x/u[y/v]]$

$t[y/v][x/u[y/v]][z/w]  =_{\C} t[z/w][y/v][x/u[y/v]]$ 

\item $\Compii$ and $\C$:

We have $y \notin \fv(u)\ \&\ x \notin \fv(v)\ \&\ 
z \in \fv(v)\ \&\ z \notin \fv(t[x/u])$ so that

$t[x/u][y/v][z/w] =_{\C} t[y/v][x/u][z/w]$

$t[x/u][y/v][z/w] \Rew{\Compii} t[x/u][y/v[z/w]]$

can be closed by 

$t[y/v][x/u][z/w] =_{\C} t[y/v][z/w][x/u] \Rew{\Compii} t[y/v[z/w]][x/u]$

$t[x/u][y/v[z/w]]=_{\C} t[y/v[z/w]][x/u]$

\item $\Compi$ and $\C$:

We have $y \notin \fv(u)\ \&\ x \notin \fv(v)\ \&\ 
z \in \fv(v)\ \&\ z \in \fv(t[x/u])$ so that

$t[x/u][y/v][z/w] =_{\C} t[y/v][x/u][z/w]$

$t[x/u][y/v][z/w] \Rew{\Compi} t[x/u][z/w][y/v[z/w]]$

If $z \in \fv(t)\ \&\ z \in \fv(u)$, then 

$t[y/v][x/u][z/w] \Rew{\Compi} t[y/v][z/w][x/u[z/w]] \Rew{\Compi}
t[z/w][y/v[z/w]][x/u[z/w]]$

$t[x/u][z/w][y/v[z/w]] \Rew{\Compi} t[z/w][x/u[z/w]][y/v[z/w]]
=_{\C} t[z/w][y/v[z/w]][x/u[z/w]]$

If $z \notin \fv(t)\ \&\ z \in \fv(u)$, then 

$t[y/v][x/u][z/w]\Rew{\Compi} t[y/v][z/w][x/u[z/w]]\Rew{\Compii}t[y/v[z/w]][x/u[z/w]]$

$t[x/u][z/w][y/v[z/w]]\Rew{\Compii} t[x/u[z/w]][y/v[z/w]]
=_{\C} t[y/v[z/w]][x/u[z/w]]$

If $z \in \fv(t)\ \&\ z \notin \fv(u)$, then 

$t[y/v][x/u][z/w] =_{\C} t[y/v][z/w][x/u] \Rew{\Compi} t[z/w][y/v[z/w]][x/u]$

$t[x/u][z/w][y/v[z/w]]=_{\C} t[z/w][y/v[z/w]][x/u]$

\end{itemize}

\end{proof}

\ignore{
  \technicallemmaalc*
\begin{proof}
  If $z \notin \fv(t)$ then  $\ALC(t[z/u]) = \l \vec{y}.\ALC(t')[z/u]$
  with $t=\l  \vec{y}.t'$ and    $t'   \neq  \l y'.t''$.   Then    $\l
  \vec{y}.\ALC(t')[z/u]    \Rew{\les}   \l      \vec{y}.\ALC(t')[z/u']
  \Rewn{\les} \l \vec{y}.\ALC(t')[z/\ALC(u')] = \ALC(t[z/u'])$.

Otherwise,  $z \in \fv(t)$, and  we
reason by induction on  the structure of $t$. 
\report{
\begin{itemize}
\item $t = z$. 
$\ALC(t[z/u]) = z[z/u] \Rew{\les} z[z/u'] \Rewn{\les} z[z/\ALC(u')] = \ALC(t[z/u'])$.

\item $t = \l y. t_1$. 

$\ALC(t[z/u]) = \ALC((\l y. t_1)[z/u]) = \l y.\ALC(t_1[z/u])$\\
$\Rewplus{\les} $ (\ih) $\:\l y.\ALC(t_1[z/u']) = \ALC((\l y.t_1)[z/u'])$.

\item $t = (t_1\ t_2)$.  

There are three subcases.

\begin{itemize}
\item $z \in \fv(t_1)$, $z \in \fv(t_2)$.

$\ALC(t[z/u]) = \ALC((t_1\ t_2)[z/u]) = (\ALC(t_1[z/u])\ \ALC(t_2[z/u]))$\\
$\Rewplus{\les}$ (\ih\ twice) $(\ALC(t_1[z/u'])\ \ALC(t_2[z/u']))$\\
$ = \ALC(t_1[z/u']\ t_2[z/u']) = \ALC((t_1\ t_2)[z/u'])$.

\item $z \in \fv(t_1)$, $z \notin \fv(t_2)$. Similar.

\item $z \notin \fv(t_1)$, $z \in \fv(t_2)$. Similar.
\end{itemize}

\item  $t =  t_0[\vec{x}/\vec{t}]$ with  $\vec{x} =  x_1, \ldots,
  x_n$,  $\vec{t}  =  t_1,  \ldots,  t_n$, $n  >  0$,  $t_0  \neq
  v_1[y/v_2]$,   $x_j  \notin  \fv(t_i)$   for  all   $i,j \in [1..n]$ (the
  substitutions are independent).

%

Since  we are rewriting  modulo $\Es$  and the  substitutions are
independent, we choose an order for them in the proofs below. For
convenience,    we   write    $[x_i\ldots   x_j]$    instead   of
$[x_i/t_i][x_{i+1}/t_{i+1}]\cdots[x_j/t_j]$.

We now reason by cases on $t_0$.

\begin{itemize}
\item $t_0 = y$. Let $z \in \fv(t_i)$, $i \in [1..k]$, $z \notin
  \fv(t_{j})$, $j \in [k+1..n]$.  Let $y = z$. We have
\begin{eqnarray*}
&\ALC(t[z/u])\\ 
=&\ALC(z[\vec{x}/\vec{t}][z/u])\\
=&z[z/u][x_1/\ALC(t_1[z/u])]\cdots[x_k/\ALC(t_k[z/u])][x_{k+1}\ldots x_n]\\
\mbox{(\ih)} \Rewplus{\les} &z[z/u][x_1/\ALC(t_1[z/u'])]\cdots[x_k/\ALC(t_k[z/u'])][x_{k+1}\ldots x_n]\\
\Rewplus{\les}&z[z/\ALC(u')][x_1/\ALC(t_1[z/u'])]\cdots[x_k/\ALC(t_k[z/u'])][x_{k+1}\ldots x_n]\\
=&\ALC(z[\vec{x}/\vec{t}][z/u'])\\
=&\ALC(t[z/u'])
\end{eqnarray*}

W.l.o.g., the remaining  cases are $y = x_i$ for $i \in [1..k]$,
$y = x_{i}$
for $i \in [k+1..n]$, and
$y \neq z,x_i$ for all $i \in [1..n]$ and they are similar.

\item $t_0 = \l y. v$.

$\ALC(t[z/u]) = \ALC((\l y. v)[\vec{x}/\vec{t}][z/u]) = \l y.\ALC(v[\vec{x}/\vec{t}][z/u])\: \mbox{ (\ih)}$\\
$\Rewplus{\les} \l y. \ALC(v[\vec{x}/\vec{t}][z/u']) =\ALC((\l y. v)[\vec{x}/\vec{t}][z/u']) = \ALC(t[z/u'])$.

\item $t_0 = (v_1\ v_2)$.
We partition the set of substitutions $\{[x_1/t_1],\ldots,[x_n/t_n]\}$ as follows:
\begin{itemize}
\item $x_1 \ldots x_{k_1} \in \fv(v_1)$, $x_1 \ldots x_{k_1} \notin \fv(v_2)$;
\item $x_{k_1+1} \ldots x_{k_2} \in \fv(v_1)$, $x_{k_1+1} \ldots x_{k_2} \in \fv(v_2)$;
\item $x_{k_2+1} \ldots x_{k_3} \notin \fv(v_1)$, $x_{k_2+1} \ldots x_{k_3} \in \fv(v_2)$;
\item $x_{k_3+1} \ldots x_{k_4} \notin \fv(v_1) \cup \fv(v_2)$, $z \in \fv(t_{k_3+1}),\ldots, \fv(t_{k_4})$;
\item $x_{k_4+1} \ldots x_{k_n} \notin \fv(v_1) \cup \fv(v_2)$, $z \notin \fv(t_{k_4+1}),\ldots, \fv(t_{k_n})$.
\end{itemize}

There are three  cases depending on whether $z  \in \fv(v_1) \cup
\bigcup_{i  \in  [1..k_2]}\fv(t_i)$   or  $z  \in  \fv(v_2)  \cup
\bigcup_{i \in [k_1..k_3]}\fv(t_i)$. We treat the case where both
conditions are true. The other cases are similar.

The       use  of \ih\    below    does    not  follow directly   from
$\ALC(v_1[\vec{x}/\vec{t}][z/u])                        \Rewplus{\les}
\ALC(v_1[\vec{x}/\vec{t}][z/u'])$   but   can   be  derived  from  the
hypothesis by  considering  the structure of   $v_1$ as  $\l  \vec{y}. 
v_1'$ and using $x_{k_2+1}\ldots  x_n \notin \fv(v_1)$ (similarly  for
$v_2$).

\begin{eqnarray*}
&\ALC(t[z/u])\\
=&\ALC((v_1\ v_2)[\vec{x}/\vec{t}][z/u])\\
=&(\ALC(v_1[x_1\ldots x_{k_2}][z/u])\ \ALC(v_2[x_{k_1+1}\ldots x_{k_3}][z/u]))\\
&[x_{k_3+1}/\ALC(t_{k_3+1}[z/u])]\cdots[x_{k_4}/\ALC(t_{k_4}[z/u])][x_{k_4+1}\ldots x_{k_n}]\\
\mbox{(\ih)} \Rewplus{\les}&(\ALC(v_1[x_1\ldots x_{k_2}][z/u'])\ \ALC(v_2[x_{k_1+1}\ldots x_{k_3}][z/u']))\\
&[x_{k_3+1}/\ALC(t_{k_3+1}[z/u])]\cdots[x_{k_4}/\ALC(t_{k_4}[z/u])][x_{k_4+1}\ldots x_{k_n}]\\
\mbox{(\ih)} \Rewplus{\les}
&(\ALC(v_1[x_1\ldots x_{k_2}][z/u'])\ \ALC(v_2[x_{k_1+1}\ldots x_{k_3}][z/u']))\\
&[x_{k_3+1}/\ALC(t_{k_3+1}[z/u'])]\cdots[x_{k_4}/\ALC(t_{k_4}[z/u'])][x_{k_4+1}\ldots x_{k_n}]\\
=&\ALC(t[z/u'])
\end{eqnarray*}
\end{itemize}

\item  $t = t_0[\vec{x}/\vec{t}]$ with  $\vec{x} =  x_1, \ldots, x_n$,
  $\vec{t} = t_1,  \ldots, t_n$, $n  > 0$, $t_0 \neq v_1[y/v_2]$, $x_j
  \in    \fv(t_i)$ for   some   $i<j$   (the   substitutions are   not
  independent). This case degenerates into the previous case:

\begin{eqnarray*}
&\ALC(t[z/u])\\
=&\ALC((v_1\ v_2)[\vec{x}/\vec{t}][z/u])\\
=&\ALC((v_1\ v_2)[\vec{y}/\vec{s}][z/u])\\
&\mbox{such that the substitutions $[\vec{y}/\vec{s}]$ are independent}\\
\mbox{previous case }  \Rewplus{\les} &\ALC((v_1\ v_2)[\vec{y}/\vec{s}][z/u'])\\
=&\ALC((v_1\ v_2)[\vec{x}/\vec{t}][z/u'])\\
=&\ALC(t[z/u'])
\end{eqnarray*}
\end{itemize}
}
\end{proof}
}

\ignore{
\lessimulateslm*

\begin{proof}
 By induction on the definition of $t \Rew{\lm} t'$. By~\cref{l:tra-preserves-es} (twice), it is
 sufficient to prove that $t \Rew{\B\sm} t'$
 implies $\tra{t}
 \Rew{\B\sm} \tra{t'}$. First we consider the base case of the induction
 where reductions occur at the root.

\begin{itemize}
\item  $t=(\lambda x.u)\ v \Rew{\B} u[x/v]=t'$. 

  If $x \in \fv(u)$, then
  $\tra{(\lambda x.u)\ v}
  =  ((\lambda x.\tra{u})\ \tra{v}) [z/\tra{v}]
  \Rew{\B}  \tra{u}[x/\tra{v}][z/\tra{v}]
\Rewn{\ALC}  \ALC(\tra{u}[x/\tra{v}][z/\tra{v}])
  =  \tra{u[x/v]}$.

  If $x \notin \fv(u)$, then 
  $\tra{(\lambda x.u)\ v}
  =  ((\lambda x.\tra{u})\ \tra{v}) [z/\tra{v}]
  \Rew{\B}  \tra{u}[x/\tra{v}][z/\tra{v}]
\Rew{\Gc}  \tra{u}[x/\tra{v}]
\Rewn{\ALC}  \ALC(\tra{u}[x/\tra{v}])=  \tra{u[x/v]}
$.

\item $t=C\cwc{x}[x/u] \Rew{\us} C\cwc{u}[x/u]=t'$. 
We reason by  induction on the context  $C$.
\begin{itemize}
\item $C = \Box$.
  Then $t=x[x/u] \Rew{\us} u[x/u]=t'$ and
  $\tra{t}=x[x/\tra{u}][\fresh{x}/\tra{u}] \Rew{\Var}
  \tra{u}[\fresh{x}/\tra{u}] \Rewn{\ALC} \ALC(\tra{u}[\fresh{x}/\tra{u}])
  = \tra{u[x/u]}$.

\item $C = D\ v$.

We treat the case where $x \notin \fv(v)$. By the \ih, $\tra{D\cwc{x}[x/u]} = $\linebreak $\ALC(\tra{D\cwc{x}}[x/\tra{u}][\fresh{x}/\tra{u}]) \Rewplus{\les} \tra{D\cwc{u}[x/u]}$. We have 
$$
\begin{array}{rcll}
\tra{D\cwc{u}[x/u]}& =& \ALC(\tra{D\cwc{u}}[x/\tra{u}]) & x \notin \fv(D\cwc{u})\\
\tra{D\cwc{u}[x/u]}& =& \ALC(\tra{D\cwc{u}}[x/\tra{u}][\fresh{x}/\tra{u}]) & x \in \fv(D\cwc{u})
\end{array}
$$
so by~\cref{l:non-interfering-garbage},
$$
\begin{array}{ll}
\ALC(\tra{D\cwc{x}}[x/\tra{u}]) \Rewplus{\les} \ALC(\tra{D\cwc{u}}) & x \notin \fv(D\cwc{u})\\
\ALC(\tra{D\cwc{x}}[x/\tra{u}]) \Rewplus{\les} \ALC(\tra{D\cwc{u}}[x/\tra{u}]) & x \in \fv(D\cwc{u}).
\end{array}
$$
When $x \in \fv(D\cwc{u})$ then
\begin{eqnarray*}
&\tra{(D\cwc{x}\ v)[x/u]}\\
=& \ALC((\tra{D\cwc{x}}\ \tra{v})[z/\tra{v}][x/\tra{u}][\fresh{x}/\tra{u}])\\
=& (\ALC(\tra{D\cwc{x}}[x/\tra{u}])\ \ALC(\tra{v}))[z/\tra{v}][\fresh{x}/\tra{u}]\\
\Rewplus{\les} & (\ALC(\tra{D\cwc{u}}[x/\tra{u}])\ \ALC(\tra{v}))[z/\tra{v}][\fresh{x}/\tra{u}]\\
=& (\ALC(\tra{D\cwc{u}}[x/\tra{u}]\ \tra{v}))[z/\tra{v}][\fresh{x}/\tra{u}]\\
=& \ALC((\tra{D\cwc{u}}\ \tra{v})[z/\tra{v}][x/\tra{u}][\fresh{x}/\tra{u}])\\
=&\tra{(D\cwc{u} v)[x/u]}.
\end{eqnarray*}

Otherwise, if $x \notin \fv(D\cwc{u})$ then
\begin{eqnarray*}
&\tra{(D\cwc{x}\ v)[x/u]}\\
=& \ALC((\tra{D\cwc{x}}\ \tra{v})[z/\tra{v}][x/\tra{u}][\fresh{x}/\tra{u}])\\
=& (\ALC(\tra{D\cwc{x}}[x/\tra{u}])\ \ALC(\tra{v}))[z/\tra{v}][\fresh{x}/\tra{u}]\\
\Rewplus{\les} & (\ALC(\tra{D\cwc{u}})\ \ALC(\tra{v}))[z/\tra{v}][\fresh{x}/\tra{u}]\\
=& \ALC((\tra{D\cwc{u}}\ \tra{v})[z/\tra{v}][\fresh{x}/\tra{u}])\\
=&\tra{(D\cwc{u} v)[x/u]}
\end{eqnarray*}

\ShortLongProof
{
The case where $x \notin \fv(D\cwc{u})$ follows similarly as do the cases where $x \in \fv(v)$.
}
{
The cases where $x \in \fv(v)$ follow similarly.
}
\item $C = v\ D$. 
\ShortLongProof
{
Similar.
}
{

We treat the case where $x \notin \fv(v)$. Again we have
$$
\begin{array}{ll}
\ALC(\tra{D\cwc{x}}[x/\tra{u}]) \Rewplus{\les} \ALC(\tra{D\cwc{u}}) & x \notin \fv(D\cwc{u})\\
\ALC(\tra{D\cwc{x}}[x/\tra{u}]) \Rewplus{\les} \ALC(\tra{D\cwc{u}}[x/\tra{u}]) & x \in \fv(D\cwc{u}).
\end{array}
$$
When $x \in \fv(D\cwc{u})$ we have the following.
\begin{eqnarray*}
&\tra{(v\ D\cwc{x})[x/u]}\\
=& \ALC((\tra{v}\ \tra{D\cwc{x}})[z/\tra{D\cwc{x}}][x/\tra{u}][\fresh{x}/\tra{u}])\\
=& (\ALC(\tra{v})\ \ALC(\tra{D\cwc{x}}[x/\tra{u}]))[z/\ALC(\tra{D\cwc{x}}[x/\tra{u}])][\fresh{x}/\tra{u}]\\
\Rewplus{\les} & (\ALC(\tra{v})\ \ALC(\tra{D\cwc{u}}[x/\tra{u}]))[z/\ALC(\tra{D\cwc{u}}[x/\tra{u}])][\fresh{x}/\tra{u}]\\
= & \ALC(\tra{v}\ \tra{D\cwc{u}}[x/\tra{u}])[z/\ALC(\tra{D\cwc{u}}[x/\tra{u}])][\fresh{x}/\tra{u}]\\
= & \ALC((\tra{v}\ \tra{D\cwc{u}})[z/\tra{D\cwc{u}}][x/\tra{u}][\fresh{x}/\tra{u}])\\
=&\tra{(v\ D\cwc{u})[x/u]}
\end{eqnarray*}
The case where $x \notin \fv(D\cwc{u})$ follows similarly as do the cases where $x \in \fv(v)$.
}
\item $C = \l z. D$. 
\ShortLongProof
{
Similar.
}
{

When $x \in \fv(D\cwc{u})$ we have:
\begin{eqnarray*}
& \tra{C\cwc{x}[x/u]}\\
=& \ALC((\l z. \tra{D\cwc{x}})[x/\tra{u}][\fresh{x}/\tra{u}]) \\
=& \l z.\ALC(\tra{D\cwc{x}}[x/\tra{u}][\fresh{x}/\tra{u}])\\
=& \l z. \tra{D\cwc{x}[x/u]}\\
\Rewplus{\les} (\ih)& \l z. \tra{D\cwc{u}[x/u]} \\
=& \l z.\ALC(\tra{D\cwc{u}}[x/\tra{u}][\fresh{x}/\tra{u}])\\
=& \ALC((\l z. \tra{D\cwc{u}})[x/\tra{u}][\fresh{x}/\tra{u}]) \\
=& \tra{C\cwc{u}[x/u]}
\end{eqnarray*}

When $x \notin \fv(D\cwc{u})$ we have:
\begin{eqnarray*}
& \tra{C\cwc{x}[x/u]}\\
=& \ALC((\l z. \tra{D\cwc{x}})[x/\tra{u}][\fresh{x}/\tra{u}]) \\
=& \l z.\ALC(\tra{D\cwc{x}}[x/\tra{u}][\fresh{x}/\tra{u}])\\
=& \l z. \tra{D\cwc{x}[x/u]}\\
\Rewplus{\les} (\ih)& \l z. \tra{D\cwc{u}[x/u]} \\
=& \l z.\ALC(\tra{D\cwc{u}}[x/\tra{u}])\\
=& \ALC((\l z. \tra{D\cwc{u}})[x/\tra{u}]) \\
=& \tra{C\cwc{u}[x/u]}
\end{eqnarray*}
}
\item $C = D[z/v]$.

The case is further broken down depending on whether $z \in \fv(D\cwc{x})$, $x \in \fv(D\cwc{u})$, and $x \in \fv(v)$. 

$\tra{D\cwc{x}[x/u]} = \ALC(\tra{D\cwc{x}}[x/\tra{u}][\fresh{x}/\tra{u}]) \Rewplus{\les} \tra{D\cwc{u}[x/u]}$ by the  \ih. We have 
$$
\begin{array}{rcll}
\tra{D\cwc{u}[x/u]}& =& \ALC(\tra{D\cwc{u}}[x/\tra{u}]) & x \notin \fv(D\cwc{u})\\
\tra{D\cwc{u}[x/u]}& =& \ALC(\tra{D\cwc{u}}[x/\tra{u}][\fresh{x}/\tra{u}]) & x \in \fv(D\cwc{u})
\end{array}
$$
so by~\cref{l:non-interfering-garbage},
$$
\begin{array}{ll}
\ALC(\tra{D\cwc{x}}[x/\tra{u}]) \Rewplus{\les} \ALC(\tra{D\cwc{u}}) & x \notin \fv(D\cwc{u})\\
\ALC(\tra{D\cwc{x}}[x/\tra{u}]) \Rewplus{\les} \ALC(\tra{D\cwc{u}}[x/\tra{u}]) & x \in \fv(D\cwc{u})
\end{array}
$$


As $z \in \fv(D\cwc{x})$ and $z\neq x$ implies
$z \in \fv( \tra{D\cwc{x}}[x/\tra{u}])$ and $z \in \fv(\tra{D\cwc{u}})$,
then 
by~\cref{l:adding-independent-subsi}.\ref{l:adding-independent-subsia} and~\cref{l:adding-independent-subsi}.\ref{l:adding-independent-subsib} we have
$$
\begin{array}{ll}
\ALC(\tra{D\cwc{x}}[x/\tra{u}][z/v']) \Rewplus{\les} \ALC(\tra{D\cwc{u}}[z/v'])\\
\hfill x \notin \fv(D\cwc{u})\\
\ALC(\tra{D\cwc{x}}[x/\tra{u}][z/v']) \Rewplus{\les} \ALC(\tra{D\cwc{u}}[x/\tra{u}][z/v']) \\
\hfill x \in \fv(D\cwc{u}).
\end{array}
$$

\ShortLongProof
{
We take the case where $x \in \fv(D\cwc{u})$, $z \in \fv(D\cwc{x})$, $x \notin \fv(v)$. We assume below that $D\cwc{x} \neq \l y.s$ for some $s$ -- the alternative case can be similarly proved.
\begin{eqnarray*}
& \tra{C\cwc{x}[x/u]}\\
= & \ALC(\tra{D\cwc{x}}[z/\tra{v}][\fresh{z}/\tra{v}][x/\tra{u}][\fresh{x}/\tra{u}])\\
= & \ALC(\tra{D\cwc{x}}[x/\tra{u}][z/\tra{v}])[\fresh{z}/\tra{v}][\fresh{x}/\tra{u}]\\
\Rewplus{\les} &     
    \ALC(\tra{D\cwc{u}}[x/\tra{u}][z/\tra{v}])[\fresh{z}/\tra{v}][\fresh{x}/\tra{u}]\\
= & \ALC(\tra{D\cwc{u}}[z/\tra{v}][\fresh{z}/\tra{v}][x/\tra{u}][\fresh{x}/\tra{u}])\\
= & \tra{C\cwc{u}[x/u]}
\end{eqnarray*}
}
{
We assume below that $D\cwc{x} \neq \l y.s$ for some $s$ -- the alternative case can be similarly proved.

Let $x \in \fv(D\cwc{u})$, $z \in \fv(D\cwc{x})$, $x \in \fv(v)$. 
\begin{eqnarray*}
& \tra{C\cwc{x}[x/u]}\\
= & \ALC(\tra{D\cwc{x}}[z/\tra{v}][\fresh{z}/\tra{v}][x/\tra{u}][\fresh{x}/\tra{u}])\\
= & \ALC(\tra{D\cwc{x}}[x/\tra{u}][z/\tra{v}[x/\tra{u}]])[\fresh{z}/\tra{v}[x/\tra{u}]][\fresh{x}/\tra{u}]\\
\Rewplus{\les} &         \ALC(\tra{D\cwc{u}}[x/\tra{u}][z/\tra{v}[x/\tra{u}]])[\fresh{z}/\tra{v}[x/\tra{u}]][\fresh{x}/\tra{u}]\\
= & \ALC(\tra{D\cwc{u}}[z/\tra{v}][\fresh{z}/\tra{v}][x/\tra{u}][\fresh{x}/\tra{u}])\\
= & \tra{C\cwc{u}[x/u]}
\end{eqnarray*}

Let $x \in \fv(D\cwc{u})$, $z \in \fv(D\cwc{x})$, $x \notin \fv(v)$. 
\begin{eqnarray*}
& \tra{C\cwc{x}[x/u]}\\
= & \ALC(\tra{D\cwc{x}}[z/\tra{v}][\fresh{z}/\tra{v}][x/\tra{u}][\fresh{x}/\tra{u}])\\
= & \ALC(\tra{D\cwc{x}}[x/\tra{u}][z/\tra{v}])[\fresh{z}/\tra{v}][\fresh{x}/\tra{u}]\\
\Rewplus{\les} &     
    \ALC(\tra{D\cwc{u}}[x/\tra{u}][z/\tra{v}])[\fresh{z}/\tra{v}][\fresh{x}/\tra{u}]\\
= & \ALC(\tra{D\cwc{u}}[z/\tra{v}][\fresh{z}/\tra{v}][x/\tra{u}][\fresh{x}/\tra{u}])\\
= & \tra{C\cwc{u}[x/u]}
\end{eqnarray*}

Let $x \notin \fv(D\cwc{u})$, $z \notin \fv(D\cwc{x})$, $x \in \fv(v)$. 
\begin{eqnarray*}
& \tra{C\cwc{x}[x/u]}\\
= & \ALC(\tra{D\cwc{x}}[z/\tra{v}][x/\tra{u}][\fresh{x}/\tra{u}])\\
= & \ALC(\tra{D\cwc{x}}[x/\tra{u}][z/\tra{v}[x/\tra{u}]])[\fresh{x}/\tra{u}]\\
\Rewplus{\les} &         \ALC(\tra{D\cwc{u}}[z/\tra{v}[x/\tra{u}]])[\fresh{x}/\tra{u}]\\
= & \ALC(\tra{D\cwc{u}}[z/\tra{v}][x/\tra{u}][\fresh{x}/\tra{u}])\\
= & \tra{C\cwc{u}[x/u]}
\end{eqnarray*}

Let $x \notin \fv(D\cwc{u})$, $z \notin \fv(D\cwc{x})$, $x \notin \fv(v)$. 
\begin{eqnarray*}
& \tra{C\cwc{x}[x/u]}\\
= & \ALC(\tra{D\cwc{x}}[z/\tra{v}][x/\tra{u}][\fresh{x}/\tra{u}])\\
= & \ALC(\tra{D\cwc{x}}[x/\tra{u}][z/\tra{v}])[\fresh{x}/\tra{u}]\\
\Rewplus{\les} &     
    \ALC(\tra{D\cwc{u}}[z/\tra{v}])[\fresh{x}/\tra{u}]\\
= & \ALC(\tra{D\cwc{u}}[z/\tra{v}][\fresh{x}/\tra{u}])\\
= & \tra{C\cwc{u}[x/u]}
\end{eqnarray*}

The remaining four cases are similar.
}

\item $C = v[z/D]$.

$\tra{D\cwc{x}[x/u]} = \ALC(\tra{D\cwc{x}}[x/\tra{u}][\fresh{x}/\tra{u}]) \Rewplus{\les} \tra{D\cwc{u}[x/u]}$ by the \ih. We assume below that $v \neq \l y.s$ for some $s$ -- the alternative is similarly explained.

The  case is further broken down  depending on whether $z \in \fv(v)$,
$x \in \fv(D\cwc{u})$,  and $x  \in \fv(v)$.  For the  cases where  $x
\notin \fv(D\cwc{u})$, we use  the    equations and reductions    from~\cref{figEqSimproofi}.    For   the    cases    where   $x   \in
\fv(D\cwc{u})$,  we   use     the equations  and     reductions   from~\cref{figEqSimproofii}.

\begin{figure}[p]
\begin{eqnarray}
\nonumber  &\tra{D\cwc{u}[x/u]}\\
\nonumber  &=\ALC(\tra{D\cwc{u}}[x/\tra{u}]).\\
\nonumber  & \mbox{Thus, by~\cref{l:non-interfering-garbage},}\\
\nonumber  &\ALC(\tra{D\cwc{x}}[x/\tra{u}])\\
\nonumber  &\Rewplus{\les} \ALC(\tra{D\cwc{u}}).\\
\nonumber  & \mbox{Thus, by~\cref{c:reduce-copies-of-subs},}\\
\nonumber  &\ALC(\tra{v}[z/\ALC(\tra{D\cwc{x}}[x/\tra{u}])])\\
&\Rewplus{\les}  \ALC(\tra{v}[z/\ALC(\tra{D\cwc{u}})]) \label{EqSimproof1}\\
\nonumber  & \mbox{and}\\
\nonumber  &\ALC(\tra{v}[x/\tra{u}])][z/\ALC(\tra{D\cwc{x}}[x/\tra{u}])])\\
&\Rewplus{\les}  \ALC(\tra{v}[x/\tra{u}])][z/\ALC(\tra{D\cwc{u}})])~\label{EqSimproof2}\\
\nonumber  & \mbox{Let $\fresh{z}$ be fresh for the equations above. Then:}\\
\nonumber  &\ALC(\tra{v}[z/\ALC(\tra{D\cwc{x}}[x/\tra{u}])][\fresh{z}/\ALC(\tra{D\cwc{x}}[x/\tra{u}])])\\
\nonumber =&\ALC(\tra{v}[z/\ALC(\tra{D\cwc{x}}[x/\tra{u}])])[\fresh{z}/\ALC(\tra{D\cwc{x}}[x/\tra{u}])]\\
\nonumber \Rewplus{\les} &
 \ALC(\tra{v}[z/\ALC(\tra{D\cwc{u}})])[\fresh{z}/\ALC(\tra{D\cwc{x}}[x/\tra{u}])]\\
\nonumber \Rewplus{\les} &
 \ALC(\tra{v}[z/\ALC(\tra{D\cwc{u}})])[\fresh{z}/\ALC(\tra{D\cwc{u}})]\\
= &\ALC(\tra{v}[z/\ALC(\tra{D\cwc{u}})][\fresh{z}/\ALC(\tra{D\cwc{u}})])~\label{EqSimproof3}\\
\nonumber  & \mbox{and}\\
\nonumber  &\ALC(\tra{v}[x/\tra{u}][z/\ALC(\tra{D\cwc{x}}[x/\tra{u}])][\fresh{z}/\ALC(\tra{D\cwc{x}}[x/\tra{u}])])\\
\nonumber =&\ALC(\tra{v}[x/\tra{u}][z/\ALC(\tra{D\cwc{x}}[x/\tra{u}])])[\fresh{z}/\ALC(\tra{D\cwc{x}}[x/\tra{u}])]\\
\nonumber \Rewplus{\les} &
 \ALC(\tra{v}[x/\tra{u}][z/\ALC(\tra{D\cwc{u}})])[\fresh{z}/\ALC(\tra{D\cwc{x}}[x/\tra{u}])]\\
\nonumber \Rewplus{\les} &
 \ALC(\tra{v}[x/\tra{u}][z/\ALC(\tra{D\cwc{u}})])[\fresh{z}/\ALC(\tra{D\cwc{u}})]\\
=&
\ALC(\tra{v}[x/\tra{u}][z/\ALC(\tra{D\cwc{u}})][\fresh{z}/\ALC(\tra{D\cwc{u}})])~\label{EqSimproof4} \\
\end{eqnarray}
\caption{Equations for the cases where $x$ is not free in $\fv(D{\cwc{u}})$}
\label{figEqSimproofi}
\end{figure}

\begin{figure}[p]
\begin{eqnarray}
\nonumber &\tra{D\cwc{u}[x/u]}\\
\nonumber &= \ALC(\tra{D\cwc{u}}[x/\tra{u}][\fresh{x}/\tra{u}]).\\
\nonumber & \mbox{Thus, by~\cref{l:non-interfering-garbage},}\\
\nonumber &\ALC(\tra{D\cwc{x}}[x/\tra{u}])\\
\nonumber &\Rewplus{\les} \ALC(\tra{D\cwc{u}}[x/\tra{u}]).\\
\nonumber & \mbox{Thus, by~\cref{c:reduce-copies-of-subs},}\\
\nonumber &\ALC(\tra{v}[z/\ALC(\tra{D\cwc{x}}[x/\tra{u}])])\\
&\Rewplus{\les} \ALC(\tra{v}[z/\ALC(\tra{D\cwc{u}}[x/\tra{u}])]) \label{EqSimproof5}\\
\nonumber  & \mbox{and}\\
\nonumber &\ALC(\tra{v}[x/\tra{u}])][z/\ALC(\tra{D\cwc{x}}[x/\tra{u}])])\\
&\Rewplus{\les} \ALC(\tra{v}[x/\tra{u}])][z/\ALC(\tra{D\cwc{u}}[x/\tra{u}])])~\label{EqSimproof6}\\
\nonumber  & \mbox{Let $\fresh{z}$ be fresh for the equations above. Then:}\\
\nonumber  &\ALC(\tra{v}[z/\ALC(\tra{D\cwc{x}}[x/\tra{u}])][\fresh{z}/\ALC(\tra{D\cwc{x}}[x/\tra{u}])])\\
\nonumber =&\ALC(\tra{v}[z/\ALC(\tra{D\cwc{x}}[x/\tra{u}])])[\fresh{z}/\ALC(\tra{D\cwc{x}}[x/\tra{u}])]\\
\nonumber \Rewplus{\les} &
 \ALC(\tra{v}[z/\ALC(\tra{D\cwc{u}}[x/\tra{u}])])[\fresh{z}/\ALC(\tra{D\cwc{x}}[x/\tra{u}])]\\
\nonumber \Rewplus{\les} &
 \ALC(\tra{v}[z/\ALC(\tra{D\cwc{u}}[x/\tra{u}])])[\fresh{z}/\ALC(\tra{D\cwc{u}}[x/\tra{u}])]\\
= &\ALC(\tra{v}[z/\ALC(\tra{D\cwc{u}}[x/\tra{u}])][\fresh{z}/\ALC(\tra{D\cwc{u}}[x/\tra{u}])])~\label{EqSimproof7}\\
\nonumber  & \mbox{and}\\
\nonumber  &\ALC(\tra{v}[x/\tra{u}][z/\ALC(\tra{D\cwc{x}}[x/\tra{u}])][\fresh{z}/\ALC(\tra{D\cwc{x}}[x/\tra{u}])])\\
\nonumber =&\ALC(\tra{v}[x/\tra{u}][z/\ALC(\tra{D\cwc{x}}[x/\tra{u}])])[\fresh{z}/\ALC(\tra{D\cwc{x}}[x/\tra{u}])]\\
\nonumber \Rewplus{\les} & \ALC(\tra{v}[x/\tra{u}][z/\ALC(\tra{D\cwc{u}}[x/\tra{u}])])[\fresh{z}/\ALC(\tra{D\cwc{x}}[x/\tra{u}])]\\
\nonumber \Rewplus{\les} &
 \ALC(\tra{v}[x/\tra{u}][z/\ALC(\tra{D\cwc{u}}[x/\tra{u}])])[\fresh{z}/\ALC(\tra{D\cwc{u}}[x/\tra{u}])]\\
=& \ALC(\tra{v}[x/\tra{u}][z/\ALC(\tra{D\cwc{u}}[x/\tra{u}])][\fresh{z}/\ALC(\tra{D\cwc{u}}[x/\tra{u}])])~\label{EqSimproof8}
\end{eqnarray}
\caption{Equations for the cases where $x$ is free in $D{\cwc{u}}$}
\label{figEqSimproofii}
\end{figure}

The structure of the proofs is the same. Depending on $z \in \fv(v)$,  $x \notin \fv(v)$, or $x \in \fv(D\cwc{u})$, we use a different equation from the figures as follows:

\begin{center}
\begin{tabular}{|c|c|c|}
\hline
& $x \notin \fv(D\cwc{u})$ & $x \in \fv(D\cwc{u})$
\\
\hline
$z \notin \fv(v)$ and $x \notin \fv(v)$ & Use (\ref{EqSimproof1}) & Use (\ref{EqSimproof5})\\
\hline
$z \notin \fv(v)$ and $x \in \fv(v)$ & Use (\ref{EqSimproof2}) & Use (\ref{EqSimproof6})\\
\hline
$z \in \fv(v)$ and $x \notin \fv(v)$ & Use (\ref{EqSimproof3}) & Use (\ref{EqSimproof7})\\
\hline
$z \in \fv(v)$ and $x \in \fv(v)$ & Use (\ref{EqSimproof4}) & Use (\ref{EqSimproof8})\\
\hline
\end{tabular}
\end{center}

\ShortLongProof
{

For example, we take the case where $z \in \fv(v)$,  $x \notin \fv(v)$, and $x \in \fv(D\cwc{u})$. 
\begin{eqnarray*}
& \tra{C\cwc{x}[x/u]}\\
= & \ALC(\tra{v}[z/\tra{D\cwc{x}}][\fresh{z}/\tra{D\cwc{x}}][x/\tra{u}][\fresh{x}/\tra{u}])\\
= & \ALC(\tra{v}[z/\ALC(\tra{D\cwc{x}}[x/\tra{u}])][\fresh{z}/\ALC(\tra{D\cwc{x}}[x/\tra{u}])][\fresh{x}/\tra{u}])\\
\Rewplus{\les} &
\ALC(\tra{v}[z/\ALC(\tra{D\cwc{u}}[x/\tra{u}])][\fresh{z}/\ALC(\tra{D\cwc{u}}[x/\tra{u}])][\fresh{x}/\tra{u}])\\
= & \ALC(\tra{v}[z/\tra{D\cwc{u}}][\fresh{z}/\tra{D\cwc{u}}][x/\tra{u}][\fresh{x}/\tra{u}])\\
= & \tra{C\cwc{u}[x/u]}
\end{eqnarray*}
}
{

Thus for example,
let $z \in \fv(v)$,  $x \notin \fv(v)$, and $x \in \fv(D\cwc{u})$. Then,
\begin{eqnarray*}
& \tra{C\cwc{x}[x/u]}\\
= & \ALC(\tra{v}[z/\tra{D\cwc{x}}][\fresh{z}/\tra{D\cwc{x}}][x/\tra{u}][\fresh{x}/\tra{u}])\\
= & \ALC(\tra{v}[z/\ALC(\tra{D\cwc{x}}[x/\tra{u}])][\fresh{z}/\ALC(\tra{D\cwc{x}}[x/\tra{u}])][\fresh{x}/\tra{u}])\\
&\mbox{(use (\ref{EqSimproof7}) and~\cref{l:adding-independent-subsi}.\ref{l:adding-independent-subsia})}\\
\Rewplus{\les} & \ALC(\tra{v}[z/\ALC(\tra{D\cwc{u}}[x/\tra{u}])][\fresh{z}/\ALC(\tra{D\cwc{u}}[x/\tra{u}])][\fresh{x}/\tra{u}])\\
= & \ALC(\tra{v}[z/\tra{D\cwc{u}}][\fresh{z}/\tra{D\cwc{u}}][x/\tra{u}][\fresh{x}/\tra{u}])\\
= & \tra{C\cwc{u}[x/u]}
\end{eqnarray*}

Let $z \notin \fv(v)$,  $x \notin \fv(v)$, and $x \notin \fv(D\cwc{u})$.
\begin{eqnarray*}
& \tra{C\cwc{x}[x/u]}\\
= & \ALC(\tra{v}[z/\tra{D\cwc{x}}][x/\tra{u}][\fresh{x}/\tra{u}])\\
= & \ALC(\tra{v}[z/\ALC(\tra{D\cwc{x}}[x/\tra{u}])][\fresh{x}/\tra{u}])\\
&\mbox{(use (\ref{EqSimproof1}) and~\cref{l:adding-independent-subsi}.\ref{l:adding-independent-subsia})}\\
\Rewplus{\les} & \ALC(\tra{v}[z/\ALC(\tra{D\cwc{u}})][\fresh{x}/\tra{u}])\\
= & \ALC(\tra{v}[z/\tra{D\cwc{u}}][\fresh{x}/\tra{u}])\\
= & \tra{C\cwc{u}[x/u]}
\end{eqnarray*}

The other six cases proceed in a similar way. 
}
\end{itemize}

\item  $u[x/v] \Rew{\Gc} u$.

Let us write $u$ as $\l \vec{y}.u'$,  where $u'$ is not a
$\l$-abstraction. 

 As $\fv(t) = \fv(\tra{t})$ and all bound variables are distinct, $x \notin \fv(u')$. Then

$$
\begin{array}{ll}
  & \tra{u[x/v]}\\
= & \ALC(\tra{\l \vec{y}.u'}[x/\tra{v}])\\
= & \ALC((\l \vec{y}.  \tra{ u'})[x/\tra{v}])\\
= & \ALC(\l\vec{y}. \tra{u'}[x/\tra{v}])\\
= & \l\vec{y}. \ALC(\tra{ u'}[x/\tra{v}])\\
= & \l \vec{y}.\tra{u'}[x/\tra{v}]\\
\Rew{\Gc} & \l \vec{y}.\tra{u'}\\
= & \tra{u}
\end{array}
$$

\end{itemize}

The inductive cases are the following:

\begin{itemize}
\item $t = (u_1\ u_2) \Rew{\lm} (u'_1\ u_2) = t'$. By the \ih\  we have the diagram below.
$$
\SelectTips{cm}{}
\xymatrix{
(u_1\ u_2) \ar@{-> }[r]_>(0.965){\lm} \ar[d]_{\nombretra} & (u'_1\ u_2)  \ar[d]_{\nombretra}\\
(\tra{u_1}\ \tra{u_2})[z/\tra{u_2}] \ar@{-> }[r]^>(0.88)+_>(0.925){\les}  & (\tra{u'_1}\ \tra{u_2})[z/\tra{u_2}]
}
$$

\item $t = (u_1\ u_2) \Rew{\lm} (u_1\ u_2') = t'$. Similar.
\item $t = \l y. u_1 \Rew{\lm} \l y. u'_1  = t'$. Similar.

\item $t = u_1[x/u_2] \Rew{\lm} u_1'[x/u_2] = t'$.

By the \ih, $\tra{u_1} \Rewplus{\les} \tra{u_1'}$. We break the case over three subcases.

\begin{itemize}
\item $x \notin \fv(u_1)$. 

$\tra{u_1[x/u_2]} = \ALC(\tra{u_1}[x/\tra{u_2}])$. 

By~\cref{l:adding-independent-subsi}.\ref{l:adding-independent-subsia}, $\ALC(\tra{u_1}[x/\tra{u_2}]) \Rewplus{\les} \ALC(\tra{u_1'}[x/\tra{u_2}])$ $= \tra{u_1'[x/u_2]}.$

\item $x \in \fv(u_1)$, $x \in \fv(u_1')$. 

$\tra{u_1[x/u_2]} = \ALC(\tra{u_1}[x/\tra{u_2}][\fresh{x}/\tra{u_2}])$.

By~\cref{l:adding-independent-subsi}.\ref{l:adding-independent-subsib}, $\ALC(\tra{u_1}[x/\tra{u_2}]) \Rewplus{\les} \ALC(\tra{u_1'}[x/\tra{u_2}])$ 
then by~\cref{l:adding-independent-subsi}.\ref{l:adding-independent-subsia},
$$
\begin{array}{rl}
&\ALC(\tra{u_1}[x/\tra{u_2}][\fresh{x}/\tra{u_2}])\\
=&\ALC(\ALC(\tra{u_1}[x/\tra{u_2}])[\fresh{x}/\tra{u_2}])\\
\Rewplus{\les}& \ALC(\ALC(\tra{u_1'}[x/\tra{u_2}])[\fresh{x}/\tra{u_2}])\\
=& \ALC(\tra{u_1'}[x/\tra{u_2}][\fresh{x}/\tra{u_2}])\\
=& \tra{u_1'[x/u_2]}
\end{array}
$$

\item $x \in \fv(u_1)$, $x \notin \fv(u_1')$. 

$\tra{u_1[x/u_2]} = \ALC(\tra{u_1}[x/\tra{u_2}][\fresh{x}/\tra{u_2}])$.

By~\cref{l:adding-independent-subsi}.\ref{l:adding-independent-subsic}, $\ALC(\tra{u_1}[x/\tra{u_2}]) \Rewplus{\les} \tra{u_1'}$ then by~\cref{l:adding-independent-subsi}.\ref{l:adding-independent-subsia},
$$
\begin{array}{rl}
&\ALC(\tra{u_1}[x/\tra{u_2}][\fresh{x}/\tra{u_2}])\\
=&\ALC(\ALC(\tra{u_1}[x/\tra{u_2}])[\fresh{x}/\tra{u_2}])\\
\Rewplus{\les} & \ALC(\tra{u_1'}[\fresh{x}/\tra{u_2}])\\
=_{\alpha}& \tra{u_1'[x/u_2]}
\end{array}
$$
\end{itemize}

\item $t = u_1[x/u_2] \Rew{\lm} u_1[x/u_2'] = t'$.

By the \ih, $\tra{u_2} \Rewplus{\les} \tra{u_2'}$. We break the proof over two subcases.

\begin{itemize}
\item $x \notin \fv(u_1)$. 

$\tra{u_1[x/u_2]} = \ALC(\tra{u_1}[x/\tra{u_2}])$. 

By~\cref{c:reduce-copies-of-subs}, $\ALC(\tra{u_1}[x/\tra{u_2}]) \Rewplus{\les} \ALC(\tra{u_1}[x/\tra{u_2'}])$ $= \tra{u_1[x/u_2']}$.

\item $x \in \fv(u_1)$. 

$\tra{u_1[x/u_2]} = \ALC(\tra{u_1}[x/\tra{u_2}][\fresh{x}/\tra{u_2}])$. 

By~\cref{c:reduce-copies-of-subs},
$\ALC(\tra{u_1}[x/\tra{u_2}]) \Rewplus{\les}
\ALC(\tra{u_1}[x/\tra{u_2'}])$. 

Let's  write  $\ALC(\tra{u_1}[x/\tra{u_2}])$  as $\l  \vec{y}.v$,
where $v$ is  not a $\l$-abstraction. Then we  have $\l \vec{y}.v
\Rewplus{\les}  \l  \vec{y}.v' =  \ALC(\tra{u_1}[x/\tra{u_2'}])$.
Therefore,

$$
\begin{array}{rl}
&\ALC(\tra{u_1}[x/\tra{u_2}][\fresh{x}/\tra{u_2}])\\
=& \l \vec{y}.v [\fresh{x}/\tra{u_2}]\\
\Rewplus{\les}& \l \vec{y}.v' [\fresh{x}/\tra{u_2}]\\
\Rewplus{\les}& \l \vec{y}.v' [\fresh{x}/\tra{u_2'}]\\
=&\ALC(\tra{u_1}[x/\tra{u_2'}][\fresh{x}/\tra{u_2'}])\\
=&\tra{u_1[x/u_2']}.
\end{array}
$$
\end{itemize}
\end{itemize}
\end{proof}
}

\section{Proofs of~\cref{s:charact-les-lm}}
\label{app:types}

\environmentstableorderoftypes*

\begin{proof}
By induction on the derivation of $\Gam, x:B \vd_\TST  t:A$. 

\begin{itemize}
\item If $\Gam, x:B \vd_\TST  t:A$ is an axiom, then $t$ is a variable
  $y$. We consider two cases.
      \begin{itemize}
      \item $x \neq y$. Then $y:A \in \Gam$.
        If $\TST$ is an additive
        system, then $\Gam, x:C
        \vd_\TST  y:A$ is also an axiom in $\TST$.
        If $\TST$ is a multiplicative 
        system, then $\Gam, x:B \vd_\TST  t:A$ cannot be an axiom, so
        this case does not apply. 
      \item $x = y$. Thus  $A=B$ (otherwise the judgement cannot be
        an axiom). We reason by induction on $C \ll B$.
        \begin{itemize}
        \item $C=B \ll B$. Trivial.
        \item $C=B \cap B' \ll B$. Then $\Gam, x:B \cap B' \vd_\TST 
          x:A$ follows from the axiom  $\Gam, x:B \cap B' \vd_\TST 
          x:B \cap B'$ by the typing rule $\ine$.
        \item $C=B' \cap B \ll B$. Similar.
        \item $C \ll D \ll B$. By two applications of the second \ih\  we have
          sequently $\Gam, x:D \vd_\TST  x:A$ and thus $\Gam, x:C \vd_\TST 
          x:A$.
        \item $B = B_1 \cap B_2$ and $C \ll B_1\ \&\ C \ll B_2$.
        Then $\Gam, x:C \vd_\TST 
          x:B_i$ follows from the axiom  $\Gam, x:B_i  \vd_\TST 
          x:B_i$. Then $\Gam, x:C \vd_\TST 
          x:B_1 \cap B_2$ follows  by the typing rule $\ine$.
        \end{itemize}
      \end{itemize}      
\item If $\Gam, x:B \vd_\TST  t:A$ is not an axiom, then the property
  follows straightforwardly by the first \ih. 
\end{itemize}
\end{proof}

\weakening*

\begin{proof}
\begin{description}
\item [$\Rightarrow$] By induction on  the derivation of
  $\Gam \vd_{\addls}  t:A$. The base case  $(\axiom)$ is trivial.
  For the  inductive case, we  break the proof over  the possible
  rules:
\begin{description}
\item[(\app)]
Suppose the derivations ends with 
              \[ \begin{prooftree}
                 \Gam \vd t:B \> A \sep  \Gam \vd u:B
                 \justifies{\Gam \vd (t\ u):A}
                 \end{prooftree} \] 
       
      By the \ih\  we have $\Gam, \Del \vd t:B \> A$ and $\Gam, \Del \vd u:B$. We then apply $(\app)$.
\item[(\abs)]
Let $A=B \> C$. Suppose the derivations ends with 
              \[ \begin{prooftree}
                 \Gam,x:B \vd t:C
                 \justifies{\Gam \vd \l x.t:B \> C}
                 \end{prooftree} \] 
       
      By the \ih\  we have $\Gam, \Del, x:B \vd t:C$. We then apply $(\app)$.
\end{description}
The rest of the cases are similar.
\item[$\Leftarrow$]
 By  induction on   the derivation of
  $\Gam, \Del \vd_{\addls}  t:A$.
\end{description}
\end{proof}

\additivegenerationlemma*

\begin{proof}
  The  right to  left implications   follow  from the  typing rules in
  the additive systems  and~\cref{l:ll-and-typing} and~\cref{l:ll-preserves-typing}.

The left to right implication of the first four points are
shown by induction on the typing derivation of the left part.
\begin{enumerate}
\item $\Gam \vd_\TST x:A$.
      \begin{itemize}
      \item If $x:A \in \Gam$ (so that the typing is an axiom),
            then $B=A$.
      \item Suppose  $A = C_1 \cap C_2$ and the derivations ends with 

              \[ \begin{prooftree}
                 \Gam \vd_\TST x:C_1 \sep  \Gam \vd_\TST x:C_2
                 \justifies{\Gam \vd_\TST x:C_1 \cap C_2}
                 \end{prooftree} \] 

            By the \ih\  there is $B_1 \ll C_1$ and $B_2 \ll C_2$
            s.t. $x:B_1, x:B_2 \in \Gam$, thus $B_1=B_2$
            and $B_1 \ll C_1 \cap C_2$ concludes the proof of
            this case. 
      \item Suppose  the derivations ends with 

              \[ \begin{prooftree}
                 \Gam \vd_\TST x: A \cap A'
                 \justifies{\Gam \vd_\TST x:A}
                 \end{prooftree} \] 
    
            By the \ih\  there is $B \ll A \cap A'$ 
            s.t. $x:B \in \Gam$. By transitivity 
            $B \ll A$ which concludes the proof of this case.
      \end{itemize}
\item $\Gam \vd_\TST t\ u:A$.
  \begin{itemize}
  \item If the derivation ends with
        \[ \begin{prooftree}
            \Gam \vd_\TST t: A'  \> A \sep \Gam \vd_\TST u:A'
            \justifies{\Gam \vd_\TST t\ u:A}
           \end{prooftree}\]

        then the property immediately holds. 
  \item If the  derivation ends with
        \[ \begin{prooftree}
           \Gam \vd_\TST t\ u:C_1  \sep \Gam \vd_\TST t\ u: C_2
           \justifies{\Gam \vd_\TST t\ u:C_1 \cap C_2}
           \end{prooftree}\]

        By the \ih\  there are $A_i, B_i, i \in \un{n}$ s.t.
        $\capp{n} A_i \ll C_1$ and $\Gam \vd_\TST t: B_i \> A_i$ and
        $\Gam \vd_\TST u:B_i$ for all $i \in \un{n}$. Also, there are $A'_i, B'_i, i \in \un{n'}$ s.t.
        $\capp{n'} A'_i \ll C_2$ and $\Gam \vd_\TST t: B'_i \> A'_i$ and
        $\Gam \vd_\TST u:B'_i$ for all $i \in \un{n}$. Since  $\capp{n} A_i \cap \capp{n'} A'_i
        \ll C_1 \cap C_2$, 
         this concludes this case. 

  \item If the  derivation ends with
        \[ \begin{prooftree}
           \Gam \vd_\TST t\ u:A \cap B
           \justifies{\Gam \vd_\TST t\ u:A}          
           \end{prooftree}\]

       By the \ih\  there are $A_i, B_i, i \in \un{n}$ s.t.
        $\capp{n} A_i \ll A \cap B$ and $\Gam \vd_\TST t: B_i \> A_i$ and
        $\Gam \vd_\TST u:B_i$ for all $i \in \un{n}$. Since $\capp{n} A_i \ll A$ this
        concludes this case. 
  \end{itemize}
\item $\Gam \vd_\TST t[x/u]:A$.
   \begin{itemize}
  \item If the derivation ends with
        \[ \begin{prooftree}
            \Gam \vd_\TST u: B \sep \Gam, x:B \vd_\TST t:A
            \justifies{\Gam \vd_\TST t[x/u]:A}
           \end{prooftree}\]

        then the property immediately holds. 
  \item If the  derivation ends with
        \[ \begin{prooftree}
           \Gam \vd_\TST t[x/u]:C_1  \sep \Gam \vd_\TST t[x/u]: C_2
           \justifies{\Gam \vd_\TST t[x/u]:C_1 \cap C_2}
           \end{prooftree}\]

        By the \ih\  there are $A_i, B_i, i \in \un{n}$ s.t.
        $\capp{n} A_i \ll C_1$ and $\Gam \vd_\TST u: B_i$ and
        $\Gam, x:B_i \vd_\TST t:A_i$ for all $i \in \un{n}$. Also  there are $A'_i, B'_i, i \in \un{n'}$ s.t.
        $\capp{n'} A'_i \ll C_2$ and $\Gam \vd_\TST u: B'_i$ and
        $\Gam, x:B'_i \vd_\TST t:A'_i$ for all $i \in \un{n}$.
         Since  $\capp{n} A_i \cap \capp{n'} A'_i
        \ll C_1 \cap C_2$, 
         this concludes this case. 

  \item If the  derivation ends with
        \[ \begin{prooftree}
           \Gam \vd_\TST t[x/u]:A \cap B
           \justifies{\Gam \vd_\TST t[x/u]:A}          
           \end{prooftree}\]

       By the \ih\  there are $A_i, B_i, i \in \un{n}$ s.t.
        $\capp{n} A_i \ll A \cap B$ and $\Gam, x:B_i \vd_\TST t: A_i$ and
        $\Gam \vd_\TST u:B_i$ for all $i \in \un{n}$. Since $\capp{n} A_i \ll A$ this
        concludes this case. 
  \end{itemize}
\item  $\Gam \vd_\TST\l x. t:A$.
  \begin{itemize}
  \item If $A = A_1 \> A_2$ and the derivation ends with
        \[ \begin{prooftree}
            \Gam, x:A_1 \vd_\TST t: A_2
            \justifies{\Gam \vd_\TST\l x. t:A_1 \> A_2}
           \end{prooftree}\]
        then the property immediately holds. 
  
  \item If $A = C_1 \cap C_2$ the derivation ends with
             \[ \begin{prooftree}
            \Gam \vd_\TST\l x. t:C_1 \sep \Gam \vd_\TST\l x. t:C_2
            \justifies{\Gam \vd_\TST\l x. t:C_1 \cap C_2}
           \end{prooftree}\]

        By the \ih\  there are $A_i, B_i, i \in
  \un{n}$ s.t. $\capp{n} (A_i \> B_i) \ll C_1$
and    $\Gam, x:A_i \vd_\TST t:B_i$ for all $i \in \un{n}$. Also, there are
       $A'_i, B'_i, i \in
  \un{n'}$ s.t. $\capp{n'} (A'_i \> B'_i) \ll C_2$
and    $\Gam, x:A'_i \vd_\TST t:B'_i$ for all $i \in \un{n}$.
        Since  $\capp{n} (A_i \> B_i) \cap \capp{n'} (A'_i \> B'_i)  \ll C_1 \cap C_2$, 
         this concludes this case. 
  \item If the derivation ends with
             \[ \begin{prooftree}
            \Gam \vd_\TST\l x. t:A \cap B
            \justifies{\Gam \vd_\TST\l x. t:A}
          \end{prooftree}\] 

          By  the  \ih\  there  are  $A_i,  B_i, i  \in \un{n}$  s.t.
          $\capp{n} (A_i \>  B_i) \ll A \cap B$  and $\Gam, x:A_i
          \vd_\TST t:B_i$ for all $i \in \un{n}$. Since $\capp{n} (A_i \> B_i) \ll A$ this
          concludes this case.

  \end{itemize}
      
\end{enumerate}

The left to right implication of point~\ref{last-add} follows from point~\ref{lambda-case-add}
and~\cref{l:relation-between-ll-types}. Indeed,
If $\Gam \vd_\TST\l x. t:B \> C$, then point~\ref{lambda-case-add}
gives  $\Gam, x:B_i \vd_\TST t:C_i$ for $\capp{n} (B_i \> C_i) \ll B
\> C$. \Cref{l:relation-between-ll-types} gives
$B \> C  =B_j \> C_j$ for some $j \in \un{n}$, thus 
$\Gam, x:B \vd_\TST t:C$. 
\end{proof}

\multiplicativegenerationlemma*
\begin{proof}
The right to left implications follow from the typing rules in in the
multiplicative systems in~\cref{f:typing-rules-l},~\cref{l:ll-and-typing}
and~\cref{l:ll-preserves-typing}:
\begin{enumerate}
\item $\Gam = x:B, B \ll A$.

Use $(\axiom)$ to prove $\Gam \vd_\TST x:B$ and then apply~\cref{l:ll-and-typing}.

\item $\Gam = \Gam_1 \uplus \Gam_2, \Gam_1 = \fv(t)$ and $\Gam_2 = \fv(u)$ where $\Gam_1 = \fv(t)$ and $\Gam_2 = \fv(u)$ 
and there exist $A_i, B_i \in \un{n}$ s.t.
$\capp{n} A_i \ll A$ and $\Gam_1 \vd_\TST t: B_i \> A_i$ and
$\Gam_2 \vd_\TST u:B_i$.

Applying $(\app)$, we  have $\Gam_1 \uplus \Gam_2 \vd_\TST t\ u: A_i$
for  all $i$.  Then, using  $n-1$ applications  of  $(\ini)$ proves
$\Gam_1 \uplus \Gam_2 \vd_\TST t\ u: \capp{n} A_i$. As $\capp{n} A_i
\ll A$,  we have $\Gam_1  \uplus \Gam_2 \vd_\TST t\ u: A$  by~\cref{l:ll-and-typing}.

\item $\Gam = \Gam_1 \uplus \Gam_2$,
where $\Gam_1 = \fv(t) \setminus \set{x}$ and $\Gam_2 = \fv(u)$
and there exist $A_i, B_i \in \un{n}$ s.t.
$\capp{n} A_i \ll A$ and  $\forall i \in
\un{n}$ $\Gam_2 \vd_\TST u:B_i$ and 
either $x \notin \fv(t)\ \&\ \Gam_1 \vd_\TST t: A_i$ or 
$x \in \fv(t)\ \&\ \Gam_1, x:B_i \vd_\TST t: A_i$.

We can apply $(\ini)$ repeatedly to prove $\Gam_2 \vd_\TST u:\capp{n} B_i$.

If $x  \notin \fv(t)$  then we can  apply $(\ini)$  repeatedly to
prove  $\Gam_1 \vd_\TST t:\capp{n} A_i$  so, by  $(\subsii)$, $\Gam_1
\uplus \Gam_2 \vd_\TST t[x/u]:\capp{n} A_i$.

If  $x  \in \fv(t)$  then  by~\cref{l:ll-preserves-typing},
$\Gam_1,  x:\capp{n} B_i  \vd_\TST t:A_i$  for all  $i$. We  can apply
$(\ini)$  repeatedly   to  prove  $\Gam_1,   x:\capp{n}  B_i  \vd_\TST
t:\capp{n}  A_i$ so,  by  $(\subsi)$, $\Gam_1  \uplus \Gam_2  \vd_\TST
t[x/u]:\capp{n} A_i$.

By~\cref{l:ll-and-typing}, $\Gam_1 \uplus \Gam_2 \vd_\TST t[x/u]:A$.

\item $\Gam = \fv(\l x.t)$ and there exist $B_i, C_i \in
  \un{n}$ s.t. $\capp{n} (B_i \> C_i) \ll A$
and   $\forall i \in
\un{n}$  either $x \notin \fv(t)\ \&\ \Gam \vd_\TST t:C_i$ or 
$x \in \fv(t)\ \&\ \Gam, x:B_i \vd_\TST t:C_i$. 

If $x \notin \fv(t)$ then,  by $(\absii)$, $\Gam \vd_\TST \l x.t:
B_i \> C_i$ for all $i$. We can then apply $(\ini)$ repeatedly to
prove $\Gam \vd_\TST t:\capp{n} (B_i \> C_i)$.

If $x \in \fv(t)$ then,  by $(\absi)$, $\Gam \vd_\TST\l x.t: B_i
\> C_i$  for all  $i$. We can  then apply $(\ini)$  repeatedly to
prove $\Gam \vd_\TST t:\capp{n} (B_i \> C_i)$.

By~\cref{l:ll-and-typing}, $\Gam \vd_\TST\l x.t:A$.

\item Use $(\absi)$ or $(\absii)$.
\end{enumerate}

The left to right implication of points~\ref{lax}-\ref{lambda-case}
is by induction on the typing derivation of the left part.
\begin{enumerate}
\item $\Gam \vd_\TST x:A$.
      \begin{itemize}
      \item If $\Gam = x:A$ (so that the typing is an axiom),
            then $B=A$.
      \item Suppose  $A = C_1 \cap C_2$ and the derivations ends with 

              \[ \begin{prooftree}
                 \Gam \vd_\TST x:C_1 \sep  \Gam \vd_\TST x:C_2
                 \justifies{\Gam \vd_\TST x:C_1 \cap C_2}
                 \end{prooftree} \] 

            By the \ih\  there is $B_1 \ll C_1$ and $B_2 \ll C_2$
            s.t. $\Gam=x:B_1$ and $\Gam=x:B_2$, thus $B_1=B_2$
            and $B_1 \ll C_1 \cap C_2$ concludes the proof of
            this case. 
      \item Suppose  the derivations ends with 

              \[ \begin{prooftree}
                 \Gam \vd_\TST x: A \cap A'
                 \justifies{\Gam \vd_\TST x:A}
                 \end{prooftree} \] 
    
            By the \ih\  there is $B \ll A \cap A'$ 
            s.t. $\Gam = x:B$. By transitivity 
            $B \ll A$ which concludes the proof of this case.
      \end{itemize}
\item $\Gam \vd_\TST t\ u:A$.
  \begin{itemize}
  \item If the derivation ends with
        \[ \begin{prooftree}
            \Gam_1 \vd_\TST t: A'  \> A \sep \Gam_2 \vd_\TST u:A'
            \justifies{\Gam_1 \uplus \Gam_2 \vd_\TST t\ u:A}
           \end{prooftree}\]

        then the property immediately holds. 
  \item If the  derivation ends with
        \[ \begin{prooftree}
           \Gam \vd_\TST t\ u:C_1  \sep \Gam \vd_\TST t\ u: C_2
           \justifies{\Gam \vd_\TST t\ u:C_1 \cap C_2}
           \end{prooftree}\]

        By the \ih\  there are $A_i, B_i, i \in \un{n}$ s.t.
        $\capp{n} A_i \ll C_1$ and $\Gam_1 \vd_\TST t: B_i \> A_i$ and
        $\Gam_2 \vd_\TST u:B_i$ for all $i \in \un{n}$. Also, there are $A'_i, B'_i, i \in \un{n'}$ s.t.
        $\capp{n'} A'_i \ll C_2$ and $\Gam_1 \vd_\TST t: B'_i \> A'_i$ and
        $\Gam_2 \vd_\TST u:B'_i$ for all $i \in \un{n}$. Since $\capp{n} A_i \cap \capp{n'} A'_i
        \ll C_1 \cap C_2$, 
         this concludes this case. 

  \item If the  derivation ends with
        \[ \begin{prooftree}
           \Gam \vd_\TST t\ u:A \cap B
           \justifies{\Gam \vd_\TST t\ u:A}          
           \end{prooftree}\]

       By the \ih\  there are $A_i, B_i, i \in \un{n}$ s.t.
        $\capp{n} A_i \ll A \cap B$ and $\Gam_1 \vd_\TST t: B_i \> A_i$ and
        $\Gam_2 \vd_\TST u:B_i$ for all $i \in \un{n}$. Since $\capp{n} A_i \ll A$ this
        concludes this case. 
  \end{itemize}
\item $\Gam \vd_\TST t[x/u]:A$.
   \begin{itemize}
  \item If $\Gam = \Gam_1 \uplus \Gam_2$ and the derivation ends with
        \[ \begin{prooftree}
            \Gam_2 \vd_\TST u: B \sep \Gam_1, x:B \vd_\TST t:A
            \justifies{\Gam_1 \uplus \Gam_2 \vd_\TST t[x/u]:A}
           \end{prooftree}\]

        then the property immediately holds. 
  \item If  $\Gam = \Gam_1 \uplus \Gam_2$ and the derivation ends with
        \[ \begin{prooftree}
            \Gam_2 \vd_\TST u: B \sep \Gam_1 \vd_\TST t:A
            \justifies{\Gam \vd_\TST t[x/u]:A}
           \end{prooftree}\]

        then the property immediately holds. 
  \item If the  derivation ends with
        \[ \begin{prooftree}
           \Gam \vd_\TST t[x/u]:C_1  \sep \Gam \vd_\TST t[x/u]: C_2
           \justifies{\Gam \vd_\TST t[x/u]:C_1 \cap C_2}
           \end{prooftree}\]

        By the \ih\  there are $A_i, B_i, i \in \un{n}$ s.t.
        $\capp{n} A_i \ll C_1$ and $\Gam_2 \vd_\TST u: B_i$ and
        $\Gam_1, x:B_i \vd_\TST t:A_i$ or $\Gam_1 \vd_\TST t:A_i$ for all $i \in \un{n}$. 
        Also  there are $A'_i, B'_i, i \in \un{n'}$ s.t.
        $\capp{n'} A'_i \ll C_2$ and $\Gam_2 \vd_\TST u: B'_i$ and
        $\Gam_2, x:B'_i \vd_\TST t:A'_i$ or $\Gam_2 \vd_\TST t:A'_i$ for all $i \in \un{n}$.
         Since  $\capp{n} A_i \cap \capp{n'} A'_i
        \ll C_1 \cap C_2$, 
         this concludes this case. 

  \item If the  derivation ends with
        \[ \begin{prooftree}
           \Gam \vd_\TST t[x/u]:A \cap B
           \justifies{\Gam \vd_\TST t[x/u]:A}          
           \end{prooftree}\]

       By the \ih\  there are $A_i, B_i, i \in \un{n}$ s.t.
        $\capp{n} A_i \ll A \cap B$ and $\Gam_1, x:B_i \vd_\TST t:
        A_i$ or $\Gam_1 \vd_\TST t: A_i$ and
        $\Gam_2 \vd_\TST u:B_i$ for all $i \in \un{n}$. Since $\capp{n} A_i \ll A$ this
        concludes this case. 
  \end{itemize}
\item  $\Gam \vd_\TST\l x. t:A$.
  \begin{itemize}
  \item If $A = A_1 \> A_2$ and the derivation ends with
        \[ \begin{prooftree}
            \Gam, x:A_1 \vd_\TST t: A_2
            \justifies{\Gam \vd_\TST\l x. t:A_1 \> A_2}
           \end{prooftree}\]
        then the property immediately holds. 
  
  \item If $A = A_1 \> A_2$ and the derivation ends with
        \[ \begin{prooftree}
            \Gam \vd_\TST t: A_2
            \justifies{\Gam \vd_\TST\l x. t:A_1 \> A_2}
           \end{prooftree}\]
        then the property immediately holds. 
  
  \item If $A = C_1 \cap C_2$ the derivation ends with
             \[ \begin{prooftree}
            \Gam \vd_\TST\l x. t:C_1 \sep \Gam \vd_\TST\l x. t:C_2
            \justifies{\Gam \vd_\TST\l x. t:C_1 \cap C_2}
           \end{prooftree}\]

        By the \ih\  there are $A_i, B_i, i \in
  \un{n}$ s.t. $\capp{n} (A_i \> B_i) \ll C_1$
and    $\Gam, x:A_i \vd_\TST t:B_i$ or $\Gam \vd_\TST t:B_i$ for all $i \in \un{n}$. Also, there are
       $A'_i, B'_i, i \in
  \un{n'}$ s.t. $\capp{n'} (A'_i \> B'_i) \ll C_2$
and    $\Gam, x:A'_i \vd_\TST t:B'_i$ or $\Gam \vd_\TST t:B'_i$ for all $i \in \un{n}$.
        Since  $\capp{n} (A_i \> B_i) \cap \capp{n'} (A'_i \> B'_i)  \ll C_1 \cap C_2$, 
         this concludes this case. 
  \item If the derivation ends with
             \[ \begin{prooftree}
            \Gam \vd_\TST\l x. t:A \cap B
            \justifies{\Gam \vd_\TST\l x. t:A}
          \end{prooftree}\] 

          By  the  \ih\  there  are  $A_i,  B_i, i  \in \un{n}$  s.t.
          $\capp{n} (A_i \>  B_i) \ll A \cap B$  and $\Gam, x:A_i
          \vd_\TST t:B_i$ or $\Gam  \vd_\TST t:B_i$ for all $i \in \un{n}$. Since $\capp{n} (A_i \> B_i) \ll A$ this
          concludes this case.

  \end{itemize}
      
\end{enumerate}

The left to right implication of point~\ref{last} follows from point~\ref{lambda-case}
and~\cref{l:relation-between-ll-types}. Indeed,
If $\Gam \vd_\TST\l x. t:B \> C$, then point~\ref{lambda-case}
gives  $\Gam, x:B_i \vd_\TST t:C_i$ (resp. $\Gam \vd_\TST t:C_i$) for $\capp{n} (B_i \> C_i) \ll B
\> C$. \Cref{l:relation-between-ll-types} gives
$B \> C  =B_j \> C_j$ for some $j \in \un{n}$, thus 
$\Gam, x:B \vd_\TST t:C$ (resp. $\Gam \vd_\TST t:C$). 
\end{proof}


\end{document}